\documentclass[12pt]{article}
\usepackage[latin9]{inputenc}
\usepackage[letterpaper]{geometry}
\usepackage{verbatim}
\usepackage{amsmath}
\usepackage{amsthm}
\usepackage{amssymb}
\usepackage{graphicx}
\usepackage{setspace}
\usepackage[authoryear]{natbib}
\usepackage[unicode=true,
 bookmarks=false,
 breaklinks=false,pdfborder={0 0 1},backref=false,colorlinks=false]
 {hyperref}
\hypersetup{
 colorlinks,citecolor=blue,bookmarksnumbered=true,urlcolor=Indigo,linkcolor=blue}

\makeatletter

\providecommand{\tabularnewline}{\\}

\theoremstyle{plain}
\newtheorem{assumption}{\protect\assumptionname}
\theoremstyle{definition}
 \newtheorem{example}{\protect\examplename}
\theoremstyle{remark}
\newtheorem{rem}{\protect\remarkname}
\theoremstyle{plain}
\newtheorem{prop}{\protect\propositionname}
\theoremstyle{plain}
\newtheorem{thm}{\protect\theoremname}
\theoremstyle{definition}
\newtheorem{defn}{\protect\definitionname}
\theoremstyle{plain}
\newtheorem{lem}{\protect\lemmaname}

\usepackage{bbm}
\usepackage{amsmath}
\usepackage{amsthm}
\usepackage{amssymb}
\usepackage{setspace}
\usepackage{comment}
\usepackage{graphicx}
\usepackage{rotating}
\usepackage[authoryear]{natbib}
\usepackage{xcolor}

\allowdisplaybreaks

\makeatother

\providecommand{\assumptionname}{Assumption}
\providecommand{\definitionname}{Definition}
\providecommand{\examplename}{Example}
\providecommand{\lemmaname}{Lemma}
\providecommand{\propositionname}{Proposition}
\providecommand{\remarkname}{Remark}
\providecommand{\theoremname}{Theorem}

\begin{document}
\global\long\def\a{\alpha}%
\global\long\def\b{\beta}%
\global\long\def\g{\gamma}%
\global\long\def\d{\delta}%
\global\long\def\e{\epsilon}%
\global\long\def\l{\lambda}%
\global\long\def\t{\theta}%
\global\long\def\o{\omega}%
\global\long\def\s{\sigma}%
\global\long\def\G{\Gamma}%
\global\long\def\D{\Delta}%
\global\long\def\L{\Lambda}%
\global\long\def\T{\Theta}%
\global\long\def\O{\Omega}%
\global\long\def\R{\mathcal{R}}%
\global\long\def\N{\mathbb{N}}%
\global\long\def\Q{\mathbb{Q}}%
\global\long\def\I{\mathbb{I}}%
\global\long\def\P{\mathbb{P}}%
\global\long\def\E{\mathbb{E}}%
\global\long\def\B{\mathbb{\mathbb{B}}}%
\global\long\def\S{\mathbb{\mathbb{S}}}%
\global\long\def\V{\mathbb{\mathbb{V}}\text{ar}}%
\global\long\def\GG{\mathbb{G}}%
\global\long\def\TT{\mathbb{T}}%
\global\long\def\X{{\bf X}}%
\global\long\def\cX{\mathscr{X}}%
\global\long\def\cY{\mathscr{Y}}%
\global\long\def\cA{\mathscr{A}}%
\global\long\def\cB{\mathscr{B}}%
\global\long\def\cF{\mathscr{F}}%
\global\long\def\cM{\mathscr{M}}%
\global\long\def\cN{\mathcal{N}}%
\global\long\def\cG{\mathcal{G}}%
\global\long\def\cC{\mathcal{C}}%
\global\long\def\sp{\,}%
\global\long\def\es{\emptyset}%
\global\long\def\mc#1{\mathscr{#1}}%
\global\long\def\ind{\mathbf{\mathbbm1}}%
\global\long\def\indep{\perp}%
\global\long\def\any{\forall}%
\global\long\def\ex{\exists}%
\global\long\def\p{\partial}%
\global\long\def\cd{\cdot}%
\global\long\def\Dif{\nabla}%
\global\long\def\imp{\Rightarrow}%
\global\long\def\iff{\Leftrightarrow}%
\global\long\def\up{\uparrow}%
\global\long\def\down{\downarrow}%
\global\long\def\arrow{\rightarrow}%
\global\long\def\rlarrow{\leftrightarrow}%
\global\long\def\lrarrow{\leftrightarrow}%
\global\long\def\abs#1{\left|#1\right|}%
\global\long\def\norm#1{\left\Vert #1\right\Vert }%
\global\long\def\rest#1{\left.#1\right|}%
\global\long\def\bracket#1#2{\left\langle #1\middle\vert#2\right\rangle }%
\global\long\def\sandvich#1#2#3{\left\langle #1\middle\vert#2\middle\vert#3\right\rangle }%
\global\long\def\turd#1{\frac{#1}{3}}%
\global\long\def\ellipsis{\textellipsis}%
\global\long\def\sand#1{\left\lceil #1\right\vert }%
\global\long\def\wich#1{\left\vert #1\right\rfloor }%
\global\long\def\sandwich#1#2#3{\left\lceil #1\middle\vert#2\middle\vert#3\right\rfloor }%
\global\long\def\abs#1{\left|#1\right|}%
\global\long\def\norm#1{\left\Vert #1\right\Vert }%
\global\long\def\rest#1{\left.#1\right|}%
\global\long\def\inprod#1{\left\langle #1\right\rangle }%
\global\long\def\ol#1{\overline{#1}}%
\global\long\def\ul#1{\underline{#1}}%
\global\long\def\td#1{\tilde{#1}}%
\global\long\def\bs#1{\boldsymbol{#1}}%
\global\long\def\upto{\nearrow}%
\global\long\def\downto{\searrow}%
\global\long\def\pto{\overset{p}{\longrightarrow}}%
\global\long\def\dto{\overset{d}{\longrightarrow}}%
\global\long\def\asto{\overset{a.s.}{\longrightarrow}}%

\setlength{\abovedisplayskip}{6pt} 
\setlength{\belowdisplayskip}{6pt}
\title{Identification in Nonlinear Dynamic Panel \\Models under Partial Stationarity\thanks{We are grateful to Jason Blevins, Xiaohong Chen, Bryan Graham, Jiaying
Gu, Robert de Jong, Shakeeb Khan, Kyoo il Kim, Louise Laage, Xiao
Lin, Eric Mbakop, Chris Muris, Adam Rosen, Frank Schorfheide, Liyang
Sun, Elie Tamer, Valentin Verdier, Jeffrey Wooldridge, as well as numerous
seminar and conference participants for helpful comments and suggestions.}}
\author{Wayne Yuan Gao\thanks{Department of Economics, University of Pennsylvania, 133 S 36th St,
Philadelphia, PA 19104, USA. Email: waynegao@upenn.edu}$\ \ $and Rui Wang\thanks{This work was completed while Rui Wang was at Ohio State University, prior to joining Amazon. Department of Economics, Ohio State University, 1945 N High St, Columbus, OH 43210, USA. Email: rwang.econ@gmail.com} }

%Department of Economics, Ohio State University, 1945 N High St, Columbus, OH 43210, USA.

\maketitle
%\noindent
\begin{abstract}
This paper provides a general identification approach for a wide range of nonlinear panel data models, including binary choice, ordered response, and other types of limited dependent variable models. Our approach accommodates dynamic models with any number of lagged dependent variables as well as other types of endogenous covariates. Our identification strategy relies on a partial stationarity condition, which allows for not only an unknown distribution of errors, but also temporal dependencies in errors. We derive partial identification results under flexible model specifications and establish sharpness of our identified set in the binary choice setting. We demonstrate the robust finite-sample performance of our approach using Monte Carlo simulations, and apply the approach to the empirical analysis of income categories using various ordered choice models.

%The approach is applied to study...

\noindent \textbf{~} 
\end{abstract}
\newpage{}

\section{\label{sec:intr}Introduction}

This paper provides a general and unified identification approach
for a wide range of panel data models with limited dependent variables,
including various discrete (binary, multinomial, and ordered) choice
models and censored outcome models. In particular, our approach accommodates dynamic models with any number of lagged dependent variables and contemporaneously endogenous covariates. Moreover, the identification
approach does not impose parametric distributions on unobserved heterogeneity, nor on the exact form of endogeneity, thus allowing for more flexible model specifications.     

To fix ideas, we start with the following dynamic binary choice model,
which is on its own of considerable theoretical and applied interest.
Section \ref{sec:gene} generalizes the approach to other limited
dependent variable models. Specifically, consider 
\begin{equation}
Y_{it}=\ind\left\{ W_{it}^{'}\t_{0}+\alpha_{i}+\epsilon_{it}\geq0\right\} ,\label{eq:bin}
\end{equation}
where $Y_{it}\in\left\{ 0,1\right\} $ denotes a binary outcome variable
for individual $i=1,2,...$ and time $t=1,...,T$, $W_{it}\in\R^{d_{w}}$
denotes a vector of observed covariates, $\alpha_{i}\in\R$ denotes
the unobserved fixed effect for individual $i$, and $\epsilon_{it}$
denotes the unobserved time-varying error term for individual $i$
at time $t$. The objective is to identify the parameter $\t_{0}$\footnote{We discuss in Appendix \ref{subsec:ID_CF} how our results can be used to derive bounds on certain counterfactual parameters.}
using a panel of observed variables $\left(Y_{i},W_{i}\right){}_{i=1}^{n}$,
where $W_{i}:=\left(W_{i1},...,W_{iT}\right)$, and similarly for
$Y_{i}$.  We focus on short panels, where the number of time periods
$T\geq2$ is fixed and finite.

The identification of model \eqref{eq:bin} has been explored in the
literature under various assumptions. For example, \citet{chamberlain1980}
examines identification under the logistic distribution of $\e_{it}$
and the independence of $\e_{it}$ with respect to $\left(\alpha_{i},W_{i}\right)$.
Subsequently, \citet{manski1987} relaxes the distributional assumption
and employs the following conditional stationarity of $\epsilon_{it}$
to achieve identification: 
\begin{equation}
\epsilon_{is}\,\sim\,\epsilon_{it}\mid\alpha_{i},W_{i}\quad\forall s,t=1,...,T\label{eq:cond_sta}
\end{equation}
This condition is also referred to as ``group stationarity'' or
``group homogeneity'' and has been exploited in studies such
as \citet{chernozhukov2013}, \citet*{shi2018} and \citet{pakes2019}.\footnote{To be precise, condition \eqref{eq:cond_sta} is often stated in the
following weaker ``pairwise'' version in the literature, 
\[
\epsilon_{is}\,\sim\,\epsilon_{it}\mid\alpha_{i},W_{is},W_{it},\ \quad\forall s,t=1,...,T,
\]
where only covariate realizations from the two periods $\left(s,t\right)$
are conditioned on. However, the difference between condition \eqref{eq:cond_sta}
and the pairwise version above usually only leads to a minor adaptation
of the results in the aforementioned papers (as well as in the current
one). See Remark \ref{rem:PairPS} for a follow-up discussion.} Condition \eqref{eq:cond_sta} does not impose parametric restrictions
on the distributions of $\e_{it}$ and allows for dependence between the
fixed effect $\alpha_{i}$ and the covariates $W_{i}$. However, condition
\eqref{eq:cond_sta} does impose substantial restriction on the dependence
between $W_{i}$ and the time-varying error term $\e_{it}$: it effectively
requires that all covariates in $W_{i}$ are exogenous with respect
to the time varying error $\epsilon_{it}$.\footnote{For instance, suppose $W_{it}=\left(Z_{it},X_{it}\right)$ and $\E\left[\rest{\e_{it}}W_{i}\right]=X_{it}^{'}\eta$,
then the conditional distributions of $\e_{it}$ and $\e_{is}$ cannot
be the same as long as $X_{it}^{'}\eta\neq X_{is}^{'}\eta$. Hence
condition \eqref{eq:cond_sta} fails in general.}

In many economic applications, certain components of the observable
covariates $W_{i}$, namely $X_{i}$, may exhibit endogeneity issues in violation of condition \eqref{eq:cond_sta}. For
example, in a \emph{dynamic} setting where $X_{it}$ includes the
lagged outcome variable $Y_{i,t-1}$, then the dependence  of $Y_{i,t-1}$
on $\e_{i,t-1}$
arises immediately, violating the condition \eqref{eq:cond_sta}.\footnote{We note that lagged outcome variable $Y_{it}$ are sometimes referred to as ``predetermined'' or ``weakly or sequentially exogenous'' covariates in panel data literature. We clarify that in this paper, we use the words  ``exogenous'' and ``endogenous'' in the stronger ``all-period'' sense, or more precisely, in the sense of whether condition \eqref{eq:cond_sta} holds when conditioned on the covariate in question.} For another example, if $X_{it}$ includes ``price''
or other variables that may be endogenously chosen by economic agents
after observing $\e_{it}$, then $X_{it}$ would be correlated with the
contemporaneous $\e_{it}$, so the exogeneity restriction imposed by
condition \eqref{eq:cond_sta} will again fail to hold.

In this paper, we instead impose and exploit a weaker version of  
condition \eqref{eq:cond_sta}  by excluding all endogenous components
of $W_{i}$ from the conditioning set. To be precise,
from now on we suppose that we can decompose $W_{it}$ as:
\[
W_{it}\equiv\left(Z_{it},X_{it}\right),
\]
where $Z_{it}$ is of dimension $d_{z}$, and $X_{it}$ is of dimension
$d_{x}$ with $d_{w}=d_{z}+d_{x}$. Our ``partial stationarity''
assumption is then formulated as follows:
\begin{equation}
\e_{is}\sim\e_{it}\mid\a_{i},Z_{i},\quad\forall s,t=1,...,T.\label{eq:part_stat}
\end{equation}
Our partial stationarity condition \eqref{eq:part_stat}, as its name
suggests, only requires that the errors are stationary conditional
on the realizations of a subvector of the covariates (i.e., the exogenous
covariates denoted by $Z_{i}$) while allowing the remaining covariates
(denoted by $X_{i}$) to be endogenous in arbitrary manners.\footnote{Our identification strategy and results can be easily adapted under
the alternative ``pairwise partial stationarity'' condition $\e_{is}\sim\e_{it}\mid\a_{i},Z_{is},Z_{it}.$
See Remarks \ref{rem:PairPS} and \ref{rem:PairPS_prop1} for follow-up
discussions.} In short, condition \eqref{eq:part_stat} imposes exogeneity conditions
only on a subset of covariates, which we label as ``$Z$'' and will thereafter be referred to as the ``exogenous covariates''.\footnote{Condition \eqref{eq:part_stat} also accommodates the standard stationarity
assumption conditional on all covariates.} %We can interpret condition
%\eqref{eq:part_stat} as an assumption of the existence of \emph{some}
%covariates being exogenous.

We describe how to exploit the partial stationarity condition \eqref{eq:part_stat}
to derive the identified set on the model parameters $\theta_{0}$
through a class of conditional moment inequalities, which take the
form of lower and upper bounds for the conditional distribution $\e_{it}+\a_{i}\mid Z_{i}$,
solely as functions of observed variables and the model parameters
$\t_{0}$. We show that these bounds must have nonzero intersections
over time under the partial stationarity assumption, thereby forming
a class of identifying restrictions for the parameter $\t_{0}$. Conditional
on the exogenous covariates $Z_{i}$, our class of inequalities is
indexed by a scalar $c\in\R$, which implicitly traces out all possible
values that the parametric index $Z_{it}^{'}\b_{0}+X_{it}^{'}\g_{0}$
can take. That said, we show how the\emph{ effective }number of identifying
restrictions can be reduced to be finite when $X_{it}$ has finite
support, a condition naturally satisfied in the important special
case of ``$p$-th order autoregressive'' dynamic binary choice models,
where $X_{it}$ consists of lagged outcome variables $Y_{i,t-1},Y_{i,t-2},...,Y_{i,t-p}$
that are by construction discrete.

We demonstrate the sharpness of the identified set we derived for
binary choice models. More precisely, we show that for any $\t$
that satisfies all the conditional moment inequalities we derived,
we can construct an observationally equivalent joint distribution
of the observed and unobserved variables in our model. Our proof of
sharpness consists of three main steps: we begin by demonstrating
``per-period'' sharpness, and then progressively generalize the
result from ``per-period'' to ``all-period'' sharpness, and from
discrete $X_{it}$ to general $X_{it}$. A key innovation in our proof
technique is using an explicit, simple, and general construction that
shows how marginal/aggregate stationarity restrictions and joint choice
probability restrictions can be satisfied simultaneously, which might
be of independent and wider use. 
%While our main result is about set identification, we also provide sufficient conditions for the point identification of the coefficients on exogenous covariates (under scale normalization) as well as the signs of the coefficients on endogenous covariates.

Our identification strategy based on partial stationarity applies
more broadly beyond the context of dynamic binary choice models. In
Section \ref{sec:gene}, we demonstrate its applicability in a general
nonseparable semiparametric model, and show how
it can be applied to a range of alternative limited dependent variable
models, such as ordered response models, multinomial choice models,
and censored outcome models. The results of our approach accommodate
both static and dynamic settings across all these models.

We characterize the identified set using a collection of conditional
moment inequalities, based on which estimation and inference can be
conducted using established econometric methods in the literature,
such as \citet*{chernozhukov2007}, \citet{andrews2013}, and \citet*{chernozhukov2013}.
Through Monte Carlo simulations, we demonstrate that our identification
method yields informative and robust finite-sample confidence intervals
for coefficients in both static and dynamic models.

\subsubsection*{Literature Review}

Our paper contributes directly to the line of econometric literature
on semiparametric panel discrete choice models. Dating back to \citet{manski1987},
a series of work exploits ``full'' stationarity conditions for identification,
such as \citet{abrevaya2000rank}, \citet*{chernozhukov2013}, \citet*{shi2018},
\citet{pakes2019}, \citet*{khan2021inference},
\citet{gao2020}, \citet{wang2022semiparametric}, and \citet*{botosaru2023identification}.
As discussed above, full stationarity conditions given all observable
covariates effectively require that all covariates are exogenous with
no dynamic effects (i.e., lagged dependent variables). In contrast,
we exploit the ``partial'' stationarity condition, allowing
for lagged dependent variables, as well as other endogenous covariates.

In the literature on dynamic discrete choice models, our paper is
most closely related to \citet*[KPT thereafter]{khan2023identification},
who studies the following dynamic panel binary choice model 
\begin{equation}
Y_{it}=\ind\left\{ Z_{it}^{'}\b_{0}+Y_{i,t-1}\g_{0}+\a_{i}+\e_{it}\geq0\right\} ,\label{eq:KPT}
\end{equation}
where the one-period lagged dependent variable $Y_{i,t-1}\in\{0,1\}$
serves as the endogenous covariate, and $Z_{it}$ are exogenous covariates.
KPT exactly imposes the ``partial stationarity'' condition \eqref{eq:part_stat}
in the specific context of \eqref{eq:KPT}, and derives the sharp
identified set for $\t_{0}$ by explicitly enumerating the realizations
of the one-period lagged outcome variable $Y_{i,t-1}$ (across two
periods $t,s$). In contrast, our model \eqref{eq:bin}, along with
the ``partial stationarity'' condition, is stated with more general
specifications of the endogenous covariates $X_{it}$. The covariates
$X_{it}$ can include more than one lagged dependent variables (e.g.
$Y_{i,t-1},Y_{i,t-2},$ ...) and other endogenous variables (such
as ``price'' if $Y_{it}$ represents the purchase of a particular
product), which may be continuously valued. Consequently, our identification
strategy is substantially different from that of KPT, and can be applied
more broadly to various other dynamic nonlinear panel models. In the
specific model \eqref{eq:KPT}, we show that the identifying restrictions
we derived are equivalent to those derived in KPT and thus both approaches
lead to sharp identification. Relatedly, \citet{mbakop2023identification}
proposes a computational algorithm to derive conditional moment inequalities
in a general class of dynamic discrete choice models (potentially
with multiple lags). The focus of \citet{mbakop2023identification}
is on scenarios where the lagged discrete outcome variables are the
only endogenous covariates in the model, and the proposed algorithm
relies on the discreteness of these variables. Relative to these works,
our paper introduces an analytic approach that directly applies to
a more general class of dynamic binary choice models, as well as other
types of models with continuous limited dependent variables and any
number of endogenous covariates, regardless of whether they are discrete
or continuous and whether they take the form of lagged outcome variables
or not.

Our identification strategy relies on a type of stationarity condition,
while alternative approaches utilize other notions of exogeneity.
For example, \citet{honore2000panel} provides identification by exploiting
events where the exogenous covariates stay the same across two periods:
they consider both the logit case and a semiparametric case, but both
under the independence between time-changing errors and the lagged
dependent variable, as well as the intertemporal independence of errors.
Additionally, \citet{aristodemou2021semiparametric} exploits an alternative
type of ``full independence'' assumption for identification in dynamic binary
choice models. The ``full independence'' assumption essentially
requires that the time-varying errors from all time periods and the
exogenous variables from all time periods are independent (conditional
on initial conditions), but does not make intertemporal restrictions
on the errors (such as stationarity). Hence, such `full independence''
assumption and the partial stationarity assumption in our paper do
not nest each other as special cases. \citet*{chesher2023identification}
applies the framework of generalized instrumental variables \citep{chesher2017generalized}
to the context of various dynamic discrete choice models with fixed
effects, and utilizes a similar ``full independence'' assumption
\citep{aristodemou2021semiparametric} for identification.\footnote{Our identification strategy shares some conceptual similarity with
the idea of generalized instrumental variable (GIV) in \citet*{chesher2017generalized},
who proposes a general approach for representing the identified set
of structural models with endogeneity. %as an infinite collection of inequalities about conditional probabilities of the structural errorstaking values in an arbitary given set (conditional on the exogenous variables). 
\citet*{chesher2017generalized}, \citet*{chesher2020generalized},
and \citet*{chesher2023identification} demonstrate how the GIV framework
can be applied to various settings, but focus mostly on the use of
exclusion restrictions and/or full independence assumptions. In this
paper, we neither impose exclusion restrictions nor independence assumptions
but instead explore identification under a partial stationarity condition.} More differently, some other papers work with sequential exogeneity
in various dynamic panel models and provide (non-)identification results
under different model restrictions. For example, \citet{shiu2013identification}
imposes a high-level invertibility condition along with a restriction
that rules out the dependence of covariates on past dependent variables.
More recently, \citet*{bonhomme2023identification} investigates panel
binary choice models with a single binary predetermined covariate
under \emph{sequential exogeneity}, whose evolution may depend on
the past history of outcome and covariate values. The sequential exogeneity
condition considered in these papers and the partial stationarity
condition in ours again do not nest each other as special cases: in
particular, our current paper accommodates contemporaneously endogenous
covariates that violate sequential exogeneity. In summary, the key
assumptions, identification strategy, and identification results of
these studies are substantially different from and thus not directly
comparable to those in our current paper.

%More differently, a recent paper by \citet*{bonhomme2023identification} investigates the identification of panel binary choice models under \emph{sequential exogeneity}. Specifically, they focuses on a stylized setting with a single binary predetermined covariate, whose evolution may be dependent on the past history of outcome and covariate values. Again, their setting and results are very different from and not directly comparable to ours: on one hand, in their setting with a single binary predetermined regressor, there is no ``partially stationary" covariate as in our setting; on the other hand, our identification strategy does not rely on discreteness of outcome and covariates, and more importantly can accommodate contemporaneously endogenous covariates that do not satistfy sequential exogeneity as in \citet*{bonhomme2023identification}.

\begin{comment}
That said, our identification strategy is conceptually related to
the idea of generalized instrumental variable (GIV) in \citet*{chesher2017generalized},
who proposes a general approach for representing the identified set
of a structural model with endogeneity as an infinite collection of
inequalities about conditional probabilities of the structural errors
taking values in an arbitary given set (conditional on the exogenous
variables). discuss how the approach can be used to obtain sharp identified
set under various model specifications and structural assumptions,
and their subsequent work applies the approach to specific settings. 
\end{comment}

Our paper is also complementary to the literature that studies
dynamic logit models with fixed effects. This literature typically assumes that time-varying errors are conditionally independent across time, independent from all other variables, and follow the logistic distribution. The
logit assumption in panel settings has long been studied, such
as in \citet*{chamberlain1984panel} and \citet*{chamberlain2010binary}.
We do not impose the logit assumption, nor require conditional independence
across time, and our identification strategy is very different from
those based on the logit assumption.

Our paper also contributes to the general panel data literature on
linear and nonlinear models with and without endogeneity and dynamics.
Most relatedly, \citet{botosaru2017binarization} proposes a binarization
strategy for general panel data models with fixed effects without
requiring time homogeneity, but focuses on static settings. \citet{botosaru2022}
considers a model where the outcome variable is generated as a strictly
monotone (and thus invertible) transformation of a linear model, and
they exploit time homogeneity in conditional means (instead of the
whole distributions) for identification. Our current paper, with a
focus on discrete choice models, does not require strict monotonicity
and invertibility. 

It should be acknowledged that our partial stationarity condition \eqref{eq:part_stat}, as well as the full stationarity condition \eqref{eq:cond_sta}, cannot handle random coefficients as in \cite*{berry1995automobile} and \cite{mcfadden2000mixed}. On the other hand, our approach can be naturally extended to handle nonadditive and functional fixed effects as discussed in Remark \ref{rem:scalaradd} and Section \ref{subsec:generic}. It would be interesting in future research to investigate whether our approach can be further adapted to incorporate heterogeneous coefficients. 

~

The rest of the paper is organized as follows. Section \ref{sec:bin}
studies the sharp identification of panel binary choice models with
endogenous covariates. Section \ref{sec:gene} demonstrates how our
key identification strategy generalizes to a wide range of dynamic
nonlinear panel data models, such as ordered response models, multinomial
choice models, and censored outcome models. Section \ref{sec:simu}
presents simulation results about the finite-sample performances of
our approach, and Section \ref{sec:appl} explores the empirical application
of income categories using various ordered response models. We conclude
with Section \ref{sec:conc}.

\section{\label{sec:bin}Dynamic Binary Choice Model}

\subsection{Model}

To explain the partial stationarity and our key identification strategy,
we start with the canonical binary choice model, which is of wide
theoretical interest itself. In Section \ref{sec:gene}, we explain
how our identification strategy can be applied more generally.

Specifically, consider the same binary model as introduced in \eqref{eq:bin}:
$$
Y_{it}=\ind\left\{ W_{it}^{'}\t_{0}+\alpha_{i}+\epsilon_{it}\geq0\right\}.
$$ 
Recall that we decompose $W_{it}\equiv\left(Z_{it}^{'},X_{it}^{'}\right)^{'}$, and, throughout this paper, we will refer to $Z_{i}$ as ``exogenous covariates'',
and refer to $X_{i}$ as ``endogenous covariates''. The exact difference
between $Z_{i}$ and $X_{i}$ is formalized through the ``\emph{partial stationarity}'' condition, which we now state as a formal assumption:
\begin{assumption}[Partial Stationarity]
 \label{assu:PartStat} The conditional distribution of $\e_{it}\mid Z_{i},\a_{i}$
is stationary over time, i.e., 
\[
\e_{it}\mid Z_{i},\a_{i}\stackrel{d}{\sim}\e_{is}\mid Z_{i},\a_{i}\quad\forall t,s=1,...,T.
\]
\end{assumption}
Assumption \ref{assu:PartStat} essentially requires that the (conditional)
distribution of $\e_{it}$ stays the same across all time periods
$t=1,...,T$ even if $Z_{i}$ realize to different values. To illustrate,
suppose that there are only two periods $t=1,2$, and that $Z_{i1},Z_{i2}$
realize to two values $z_{1},z_{2}$, respectively, with $z_{1}<z_{2}$.
Then Assumption \ref{assu:PartStat} requires that $\e_{i1}$ and
$\e_{i2}$ still have the same (conditional) distributions: in particular,
$\e_{i1}$ cannot be stochastically smaller (or larger) than $\e_{i2}$
because of $z_{1}<z_{2}$. Hence, Assumption \ref{assu:PartStat}
can be thought as a definition of the ``exogeneity'' of the covariates
$Z_{it}$ in our context.

In contrast, Assumption \ref{assu:PartStat} imposes no such restrictions
on the (potentially) endogenous covariates $X_{i}$. In fact, since
$X_{i}$ does not appear in Assumption \ref{assu:PartStat} at all,
here we are completely agnostic about the dependence structure between
$\e_{i}$ and $X_{i}$: in particular, the conditional distribution
of $\e_{it}$ is allowed to vary across $t$ arbitrarily for any particular
realization of $X_{i}$. As a result, different forms of endogeneity
in $X_{i}$ can be incorporated under our framework in a unified manner,
as we illustrate in the examples below. 
\begin{example}[Dynamic Effects via Lagged Outcomes]
\label{exa:dyn} Consider the following ``AR(1)'' dynamic binary
choice model studied in \citet*[KPT thereafter]{khan2023identification}:
\[
Y_{it}=\ind\left\{ Z_{it}^{'}\b_{0}+Y_{i,t-1}\g_{0}+\a_{i}+\e_{it}\geq0\right\} ,
\]
which is a special case of our model with $X_{it}$ set to be the
one-period lagged binary outcome variable $Y_{i,t-1}$. Here, $X_{it}$
is endogenous since $X_{it}\equiv Y_{i,t-1}$ and $\e_{i,t-1}$ is
by construction positively correlated with $Y_{i,t-1}$ for any $t$,
and thus the distribution of $\e_{it}$ cannot be stationary across
$t$ when conditioned on the realizations of $Y_{i0},...,Y_{i,T-1}$.
For example, given $Y_{i0}=Y_{i1}=1,$ $Y_{i2}=0$ (and $Z_{i},\a_{i}$),
the conditional distribution of $\e_{i1}$ will naturally be different
from that of $\e_{i2}$. To obtain identification under the endogeneity
of $Y_{i,t-1}$, KPT imposes the stationarity of $\e_{it}$ conditional
on the exogenous covariates $Z_{i}$ only, which coincides with our
``partial stationarity'' condition (Assumption \ref{assu:PartStat})
when specialized to their setting.

A natural generalization of the AR(1) model above in KPT is the following
``AR($p$)'' model, which is again a special case of our model with
$X_{it}$ taken to be the vector of $p$ lagged outcomes $Y_{i,t-1},...,Y_{i,t-p}$:
\[
Y_{it}=\ind\left\{ Z_{it}^{'}\b_{0}+\sum_{j=1}^{p}Y_{i,t-j}\g_{j}+\a_{i}+\e_{it}\geq0\right\} .
\]
Similarly, $X_{it}$ is endogenous here due to dependence on $\e_{i,t-1},...,\e_{t-p}$,
which can again be handled in our framework under the ``partial stationarity''
assumption. While it is not clear how the identification results in
KPT can be easily generalized to the AR($p$) model above, we show
in the next subsection how our identification strategy provides a
simple and unified approach to derive moment inequalities regardless
of the exact specifications of $X_{it}$. 
\end{example}
\begin{example}[Contemporaneously Endogenous Covariates]
\label{exa:ContempEndo} Alternatively, consider the following binary
choice model with contemporaneously endogenous covariates: 
\begin{align*}
Y_{it} & =\ind\left\{ Z_{it}^{'}\b_{0}+X_{it}^{'}\g_{0}+\a_{i}+\e_{it}\geq0\right\} ,\\
X_{it} & =\phi\left(Z_{it},u_{it}\right)
\end{align*}
where $\phi$ is an unknown ``first-stage'' function and $u_{it}$
is allowed to be arbitrarily correlated with $\e_{it}$. For example,
$X_{it}$ may be a ``price '' variable that is strategically chosen
by a decision maker after observing the current-period error $\e_{it}$,
which generates contemporaneous dependence between $X_{it}$ and $\e_{it}$.
Even though contemporaneous endogeneity of this type is very different
in nature from the dynamic endogeneity discussed in the previous example,
it also induces non-stationarity of $\e_{it}$ when conditioned on
$X_{i}$: for example, if $X_{it}$ and $\e_{it}$ are positively
correlated, then, conditional on $X_{i1}<X_{i2}$, it is unreasonable
to assume the distribution of $\e_{i1}$ is the same as $\e_{i2}$.
That said, such type of contemporaneous endogeneity can also be handled
in our framework under the ``partial stationarity'' condition (Assumption
\ref{assu:PartStat}). 
\end{example}
\begin{rem}[Combination of Dynamic and contemporaneous Endogeneity]
 We separately discussed two types of endogenous covariates, dynamic
covariates (lagged outcome variables) and contemporaneously endogenous
covariates, in the two examples above, but our identification strategy
also applies if both types of endogenous covariates are present together,
since our identification strategy works generally under ``partial
stationarity'', which does not impose or exploit any restrictions
on the form of endogeneity between $\e_{it}$ and $X_{i}$.
\end{rem}
\begin{rem}[Full Stationarity as Special Case]
 Obviously, the standard ``full stationarity'' condition \eqref{eq:cond_sta}
is nested under ``partial stationarity'' condition (Assumption \ref{assu:PartStat})
as a special case, where the endogenous covariate $X_{it}$ contains
no variables. Hence, ``full stationarity'' is in general stronger
than ``partial stationarity''. 
\end{rem}
\begin{rem}[Focus on Time-Varying Endogeneity]
 Technically, our partial stationarity condition also allows some
endogeneity between $\e_{it}$ and $Z_{i}$, as long as such endogeneity
is time-invariant. This is because Assumption \ref{assu:PartStat}
is stated conditional on the full vector $Z_{i}=\left(Z_{i1},...,Z_{iT}\right)$
and the time-invariant fixed effect $\a_{i}$. Hence, as long as the
conditional distribution of $\e_{it}$ depends on $Z_{i1},...,Z_{iT}$
and $\a_{i}$ in a time-invariant manner, the stationarity of $\e_{it}$
can still hold. That said, since in empirical applications we are
mostly interested in ``time-varying endogeneity'', such as the dynamic
and contemporaneous endogeneity discussed in the examples above, in this
paper we refer to $Z_{i}$ as ``exogenous'' even though it may be
endogenous in a time-invariant manner, and only call $X_{i}$, which
features time-varying endogeneity, the ``endogenous'' covariates. 
\end{rem}
\begin{rem}[Pairwise Version of Partial Stationarity]
\label{rem:PairPS} In Assumption \ref{assu:PartStat}, we impose
partial stationarity of $\e_{it}$ conditional on $Z_{it}$ from all
periods $t=1,...,T$. Alternatively, we could impose partial stationarity
in a ``pairwise'' version, conditional on $\left(Z_{it},Z_{is}\right)$
from any pair of time periods $\left(t,s\right)$ only: 
\begin{equation}
\text{\textbf{Pairwise Partial Stationarity}:\ensuremath{\quad}}\e_{it}\mid Z_{it},Z_{is},\a_{i}\stackrel{d}{\sim}\e_{is}\mid Z_{it},Z_{is},\a_{i},\quad\forall t,s=1,...,T.\label{eq:PartStat_ts}
\end{equation}
Clearly, the ``pairwise'' version is equivalent to the ``all-periods''
version when $T=2$, but is weaker when $T\geq3$. Our identification
strategy applies under both versions of partial stationarity, though
the identification results and the corresponding proofs have slightly
different representations. Essentially, conditioning on all-period
covariate realizations would be replaced with conditioning the realizations
in any specific pair of periods. See Remark \ref{rem:PairPS_prop1}
at the end of Section \ref{subsec:ind} for a follow-up discussion. 
\end{rem}
\begin{rem}[Initial Conditions in Dynamic Settings]
 \label{rem:initial_cond} In dynamic settings where $X_{it}$ includes lagged
outcome variables such as $Y_{i,t-1}$, the treatment of the initial
condition $Y_{i0}$ warrants some additional discussion. Our current
setup \eqref{eq:bin} treats $X_{it}$ (and the lagged outcome
variables involved) as \emph{observed}\footnote{If only $\left(Y_{i1},...,Y_{iT}\right)$ are observed,
we can truncate the time periods to satisfy this requirement. For
example, in the AR(1) setting, we can treat $Y_{i1}$ as the initial
condition $Y_{i0}$ and relabel periods 2 as period 1.} and \emph{endogenous}. However, one may consider alternative setups
where $Y_{i0}$ is treated as unobserved and/or exogenous. In Appendix
\eqref{subsec:Yi0}, we explain how our approach can be adapted to
such settings.
\end{rem}
\begin{rem}[Scalar Additivity]\label{rem:scalaradd}
 We work with the binary choice model \eqref{eq:bin} with scalar-additive
fixed effects $\a_{i}$ and error $\e_{it}$. This restriction is
unnecessary: We explain in Section \ref{sec:gene} that our identification
strategy does not rely at all on the scalar-additivity of $\a_{i}$
and $\e_{it}$. However, in this section we stick with the scalar-additive
representation \eqref{eq:bin}, since it is the most standard
specification (or notation) that is adopted in a wide body of work on
binary choice models. It thus provides a context in which most clearly
we can explain our partial stationarity condition in relation to
previous work. 
\end{rem}

\subsection{\label{subsec:ind} Key Identification Strategy}

We now explain our key identification strategy based on the partial
stationarity condition. Write $v_{it}:=-\left(\epsilon_{it}+\alpha_{i}\right)$
so that model \eqref{eq:bin} can be rewritten as 
\[
Y_{it}=\ind\left\{ v_{it}\leq W_{it}^{'}\t_{0}\right\} .
\]
For any constant $c\in\R$, consider first the event 
\[
Y_{it}=1\text{ and }W_{it}^{'}\t_{0}\leq c.
\]
Whenever the event above happens, we must have $v_{it}\leq W_{it}^{'}\t_{0}\leq c$,
implying that $v_{it}\leq c$. Formally, the above can be summarized
by the following inequality: 
\begin{equation}
\begin{aligned}Y_{it}\ind\left\{ W_{it}^{'}\t_{0}\leq c\right\}  & =\ind\left\{ v_{it}\leq W_{it}^{'}\t_{0}\right\} \ind\left\{ W_{it}^{'}\t_{0}\leq c\right\} \leq\ind\left\{ v_{it}\leq c\right\} \end{aligned}
\label{eq:vleqc_LB}
\end{equation}
Symmetrically, we can also consider the ``flipped'' event 
\[
Y_{it}=0\text{ and }W_{it}^{'}\t_{0}\geq c,
\]
which implies $v_{it}>c$:
\begin{align*}
\left(1-Y_{it}\right)\ind\left\{ W_{it}^{'}\t_{0}\geq c\right\} =\  & \ind\left\{ v_{it}>W_{it}^{'}\t_{0}\right\} \ind\left\{ W_{it}^{'}\t_{0}\geq c\right\} \\
\leq\  & \ind\left\{ v_{it}>c\right\} \equiv1-\ind\left\{ v_{it}\leq c\right\} 
\end{align*}
Rearranging the above, we have
\begin{equation}
\ind\left\{ v_{it}\leq c\right\} \leq1-\left(1-Y_{it}\right)\ind\left\{ W_{it}^{'}\t_{0}\geq c\right\} .\label{eq:vleqc_UB}
\end{equation}
Next, taking conditional expectations of \eqref{eq:vleqc_LB} and
\eqref{eq:vleqc_UB} given $Z_{i}=z$, we have
\begin{align}
\P\left(\rest{Y_{it}=1,\ W_{it}^{'}\t_{0}\leq c}z\right)\leq\  & \P\left(\rest{v_{it}\leq c}z\right)\nonumber \\
=\  & \P\left(\rest{v_{is}\leq c}z\right)\nonumber \\
\leq\  & 1-\P\left(\rest{Y_{is}=0,\ W_{is}^{'}\t_{0}\geq c}z\right)\label{eq:Bound_ts}
\end{align}
where ``$|z$'' is a shorthand for ``$Z_{i}=z$'' that we will
use throughout the paper. Note that the middle equality of \eqref{eq:Bound_ts}
follows from the partial stationarity condition (Assumption \ref{assu:PartStat}).\footnote{Specifically, observe that Assumption \ref{assu:PartStat} implies
the partial stationarity of $v_{it}$ given $Z_{i}$, since 
\begin{align*}
\P\left(\rest{\a_{i}+\e_{it}\leq c}z\right) & =\E\left[\rest{\P\left(\rest{\a_{i}+\e_{it}\leq c}z,\a_{i}\right)}\right]\\
 & =\E\left[\rest{\P\left(\rest{\a_{i}+\e_{is}\leq c}z,\a_{i}\right)}\right]=\P\left(\rest{\a_{i}+\e_{is}\leq c}z\right)
\end{align*}
for any $c$, and hence $v_{it}\mid Z_{i}\stackrel{d}{\sim}v_{is}\mid Z_{i}.$} Essentially, in the above we exploit the joint occurrence of $v_{it}\leq W_{it}^{'}\t_{0}$
and $W_{it}^{'}\t_{0}\leq c$ to deduce an implication on the composite
error $v_{it}\leq c$ that is free of the endogenous covariates $X_{it}$,
and then leverage the partial stationarity of $v_{it}$ given $Z_{i}$
for intertemporal comparisons.

Since the lower and upper bounds in \eqref{eq:Bound_ts} hold for
any $t$ and $s$, we summarize the identifying restrictions \eqref{eq:Bound_ts}
across all time periods in the following proposition.
\begin{prop}[Identified Set]
\label{prop:bin} Write 
\begin{align}
L_{t}\left(\rest cz,\t\right) & :=\P\left(\rest{Y_{it}=1,\ W_{it}^{'}\t\leq c}z\right)\nonumber \\
U_{t}\left(\rest cz,\t\right) & :=1-\P\left(\rest{Y_{it}=0,\ W_{it}^{'}\t\geq c}z\right)\label{eq:Lt_Ut}
\end{align}
and
\begin{equation}
\ol L\left(\rest cz;\t\right):=\max_{t=1,...,T}L_{t}\left(c|z;\t\right),\quad\ul U\left(\rest cz;\t\right):=\min_{t=1,...,T}U_{t}\left(c|z;\t\right),\label{eq:Lbar_Ubar}
\end{equation}
Define $\T_{I}$ as the set of $\t\in\R^{d_{w}}$ such that 
\begin{equation}
\ol L\left(\rest cz,\t\right)\leq\ul U\left(\rest cz,\t\right),\quad\forall c\in\R,\ \forall z\in{\cal Z}:=\text{Supp\ensuremath{\left(Z_{i}\right)}},\label{eq:ID_Set}
\end{equation}
Then, under model \eqref{eq:bin} and Assumption \ref{assu:PartStat},
$\t_{0}\in\T_{I}$.
\end{prop}
\begin{rem}
We note that, once conditioned on $z\equiv\left(z_{1},...,z_{T}\right)$,
the randomness in $W_{it}^{'}\t=z_{t}^{'}\b+X_{it}^{'}\g$ lies purely
in $X_{it}$ given $z$, and thus it is equivalent to write 
\[
L_{t}\left(\rest cz,\t\right):=\P\left(\rest{Y_{it}=1,\ z_{t}^{'}\b+X_{it}^{'}\g\leq c}z\right)
\]
and similarly for $U_{t}$. We will continue to use the notation $W_{it}^{'}\t_{0}$
for simplicity, but would like to emphasize this degeneracy of $Z_{it}^{'}\b$
given $Z_{i}=z.$ In particular, this means that $z_{t}^{'}\b$ can
be ``absorbed'' into the constant $c$, in a sense that will become
clearer below.
\end{rem}
Proposition \ref{prop:bin} characterizes the identified set $\T_{I}$
for $\t_{0}$ as restrictions on the conditional joint distribution
of $Y_{it}$ and $X_{it}$ given $z$. More specifically, the restrictions
in \eqref{eq:ID_Set} can be regarded as a collection of conditional
moment inequalities that relate $\ind\left\{ Y_{it}=1,\ W_{it}^{'}\t\leq c\right\} $
and $\ind\left\{ Y_{it}=0,\ W_{it}^{'}\t\geq c\right\} $ conditional
on $z$.

Proposition \ref{prop:bin} holds regardless of whether the endogenous
covariates $X_{it}$ are discrete or continuous. When $X_{it}$ are
continuous (taking a continuum of values), then Proposition \ref{prop:bin}
requires that condition \eqref{eq:ID_Set} hold for a continuum of
constants $c\in\R$, so that (the information in) the whole joint
distribution of the binary variable $Y_{it}$ and the continuous variable
$W_{it}^{'}\t=z_{t}^{'}\b+X_{it}^{'}\g$ can be captured by the collection
of joint distributions of $\left(Y_{it},\ind\left\{ W_{it}^{'}\t\leq c\right\} \right)$
across all possible cutoff values $c$.

However, when $X_{it}$ are discrete, such as in the AR$\left(p\right)$
dynamic model where $X_{it}$ consists of $p$ lagged binary outcome
variables, there is no need to evaluate \eqref{eq:ID_Set} at all
possible values of $c\in\R$, since the inequalities in \eqref{eq:ID_Set}
can only bind at finitely many values of $c$. We formalize this observation
via the following Proposition. 
\begin{prop}[Identified Set with Discrete Endogenous Covariates]
 \label{prop:disc} Suppose that the endogenous covariate $X_{it}$
can only take finite number of values in $\left\{ \ol x_{1},...,\ol x_{K}\right\} $
across all time periods $t=1,...,T$. Then $\T_{I}=\T_{I}^{disc},$
where $\T_{I}^{disc}$ consists of all $\t=\left(\b^{'},\g^{'}\right)^{'}\in\R^{d_{z}}\times\R^{d_{x}}$
that satisfy condition \eqref{eq:ID_Set} for any 
\begin{equation}
c\in\left\{ z_{t}^{'}\b+\ol x_{k}^{'}\g:k=1,...,K,t=1,...,T\right\} ,\label{eq:c_KTset}
\end{equation}
and for any $z\in{\cal Z}$. 
\end{prop}
Proposition \ref{prop:disc} shows that the discreteness of the endogenous
covariates $X_{it}$ help reduce the infinite number of inequality
restrictions in Proposition \ref{prop:bin} to finitely many, or more
precisely, $KT$ ones (conditional on $z$).

The case of discrete $X_{it}$ is conceptually important, since it
nests the dynamic AR$\left(p\right)$ model widely studied in the
literature as a special case. Clearly, when $X_{it}$ consists of
$p$ (finitely many) lagged binary outcome variables $Y_{i,t-1},...,Y_{i,t-p}$,
then $X_{it}$ by construction can only take $K=2^{p}$ discrete values.
Specialized further to the AR(1) model in KPT, Proposition \ref{prop:disc}
shows that the identified set $\T_{I}$ is characterized by $2T$
conditional restrictions, which is drastically smaller than the $9T\left(T-1\right)$
conditional restrictions listed in KPT (even when $T$ is small).
\begin{rem}
\label{rem:PairPS_prop1} Following up on Remark \ref{rem:PairPS},
if pairwise partial stationarity is adopted, then Propositions \ref{prop:bin}
and \ref{prop:disc} continue to hold with \eqref{eq:ID_Set} adapted
to the following ``pairwise'' version: 
\begin{equation}
\P\left(\rest{Y_{it}=1,\ W_{it}^{'}\t\leq c}z_{ts}\right)\leq1-\P\left(\rest{Y_{is}=0,\ W_{is}^{'}\t\geq c}z_{ts}\right),\label{eq:ID_Set_PS}
\end{equation}
for all $(t,s)$, where ``$|z_{ts}$'' denotes conditioning on the
event $\left(Z_{it},Z_{is}\right)=\left(z_{t},z_{s}\right)=:z_{ts}$.
Relative to \eqref{eq:ID_Set}, the statement in \eqref{eq:ID_Set_PS}
reflects the fact that pairwise partial stationarity is imposed on
all pairs of time periods separately instead of all $T$ time periods
jointly. It is straightforward to verify that the identification arguments
above, in particular \eqref{eq:vleqc_LB}-\eqref{eq:Bound_ts}, carry
over with all conditional probabilities/expectations taken conditional
on $z_{ts}$ instead of $z$. 
\end{rem}

\subsection{\label{subsec:Sharp}Sharpness}

So far we have only shown that $\T_{I}$ is a valid identified set
for $\t_{0}$. However, it is not yet clear whether it has incorporated
all the available information for $\t_{0}$ under the current model
specification. We now proceed to establish the sharpness of $\T_{I}$
under appropriate conditions.

We start with the discrete case where the support of $X_{it}$ is
assumed to be finite. Remarkably, the sharpness of our identified
set can be established without any additional assumptions in this
case.
\begin{thm}[Sharpness: Discrete Case]
\label{thm:sharp_disc} Suppose that $X_{it}$ only takes finitely
many values for each $t$. Then, under model \eqref{eq:bin} and
Assumption \ref{assu:PartStat}, the identified set $\T_{I}^{disc}$
is sharp. 
\end{thm}
The formal definition of sharpness, along with the complete proof
of Theorem \ref{thm:sharp_disc}, are available in Appendix \ref{subsec:pf_sharp_disc}.
In short, we show (by construction) that, for each $\t\in\T_{I}\backslash\left\{ \t_{0}\right\} $,
there exists a data generating process (DGP) that satisfies Assumption
\ref{assu:PartStat} and produces the same joint distribution of observable
data $\left(Y_{i},W_{i}\right)$ under model \eqref{eq:bin} with
parameter $\t$. Theorem \ref{thm:sharp_disc} demonstrates that our
key identification strategy based on the bounding of (endogenous)
parametric index by arbitrary constants, as described in Section \ref{subsec:ind},
is able to extract all the available information for $\t_{0}$ from
the model and the observable data, and thus it is impossible to further
differentiate $\t_{0}$ from alternatives in the identified set $\T_{I}$
under model \eqref{eq:bin} and our assumption of partial stationarity
(without further restrictions).

Theorem \ref{thm:sharp_disc} immediately implies that, in the special
case of dynamic AR$\left(p\right)$ models where $X_{it}$ consists
of discrete lagged outcomes, our characterization of the identified
set $\T_{I}^{disc}$ in Proposition \ref{prop:disc} is sharp. In
particular, our result generalizes the corresponding result in KPT,
which focuses on the AR(1) model. Furthermore, KPT characterizes the
sharp identified set via $9T\left(T-1\right)$ conditional restrictions,
the derivation of which is based on an exhaustive enumeration of lagged
outcome realizations $Y_{i,t-1}$. In this paper we adopt an entirely
different (and much more general) identification strategy, and arrive
at a characterization of the identified set by $2T$ conditional restrictions,
which we also show to be sharp by Theorem \ref{thm:sharp_disc}. Since
our model and assumption specialize exactly to that in KPT under the
AR(1) specification, it follows that our $2T$ restrictions must be
able to reproduce all the $9T\left(T-1\right)$ restrictions in KPT.
This demonstrates that our identification strategy not only applies
more generally than the one in KPT, but also leads to a more elegant
characterization of the sharp identified set with much fewer restrictions.
We provide a more detailed explanation about this point in the next
subsection.

Another conceptually remarkable, or surprising, feature of Proposition
\ref{prop:disc} and Theorem \ref{thm:sharp_disc} is that they are
established without reference to the exact nature, or interpretation,
of the endogenous covariates $X_{it}$. The identified set $\T_{I}$
we characterized is valid and sharp regardless of whether $X_{it}$
are specified as lagged outcome variables, contemporaneously endogenous
covariates, or a combination of both.

Our proof of sharpness consists of two main steps. First, we show
for each $\t\in\T_{I}\backslash\left\{ \t_{0}\right\} $ how to construct
the per-period marginal distributions of errors that match the per-period
marginal choice probabilities. Second, we show how to combine the
$T$ per-period marginal distributions into an all-period joint distribution
that matches the all-period joint choice probabilities, so that observational
equivalence holds.

The proof techniques we exploited are also different from, and thus
novel relative to, those used in the related work that leverages stationarity-type
conditions for partial identification, such as \citet*{pakes2019}
for static multinomial choice model and KPT for dynamic AR(1) model.
Instead of showing existence only, we provide a more explicit construction
of the joint distribution of the latent variables, which is valid
regardless of the exact type of endogeneity in $X_{it}$. In particular,
a key challenge in proving sharpness based on stationarity-type conditions
lies in that stationarity imposes only aggregate restrictions (via
integrals/sums) of the joint distribution of errors, which is rather
implicit to work with. A key innovation in our proof technique is
to show how marginal/aggregate stationarity restrictions and joint
choice probability restrictions can be satisfied simultaneously by
an explicit, simple and general construction, which might be of independent
and wider use.

~

Next, we seek to establish the sharpness of our identification set
in the case where certain or all components of $X_{it}$ may be continuous.
Below we present an additional set of regularity conditions for the
continuous case and the corresponding sharpness result, followed by
a discussion of the conditions and the result.
\begin{assumption}[Regularity Conditions for the Continuous Case]
 \label{assu:Cts}Suppose that: 
\begin{itemize}
\item[(a)] $\rest{W_{it}^{'}\t_{0}}z$ is continuously distributed with strictly
positive density on a bounded interval support for each $t$. 
\item[(b)] $\P\left(\rest{Y_{it}=1}W_{i}=w\right)\in\left(0,1\right)$ for each
$t$. 
\item[(c)] $\ol L\left(\rest cz,\t_{0}\right)=\ul U\left(\rest cz,\t_{0}\right)$
only for $c$'s in a set of Lebesgue measure $0$.
\end{itemize}
\end{assumption}
\begin{thm}[Sharpness: Continuous Case]
\label{thm:sharp_cts} Let $\T_{I}^{cts}$ be the set of $\t$ such
that model \eqref{eq:bin}, Assumptions \ref{assu:PartStat} and
\ref{assu:Cts} all hold with $\t$ in lieu of $\t_{0}$. Then $\T_{I}^{cts}$
is sharp. 
\end{thm}
Theorem \ref{thm:sharp_cts} establishes the sharpness of our identification
set under the additional regularity conditions imposed in Assumption
\ref{assu:Cts}. The proof, presented in Appendix \ref{subsec:pf_sharp_cts},
follows the general construction strategy used in the discrete-case
proof, with some key adaptions to handle several continuity and measure-zero
issues arising in the continuous case. Such adaptions utilize the
conditions in Assumption \ref{assu:Cts}, which we now explain in
more details.

Assumption \ref{assu:Cts}(a) can be effectively regarded as a setup
of the continuous-case model. In our current context, given $z$,
the induced index $W_{it}^{'}\t_{0}=z_{t}^{'}\b_{0}+X_{it}^{'}\g_{0}$
is what enter most directly into our model, rather than $X_{it}$
per se. Part of Assumption \ref{assu:Cts}(a) states that $W_{it}^{'}\t_{0}$
is continuously distributed on a bounded interval, which can be satisfied
with various lower-level conditions on $X_{it}$. For example, if
$X_{it}|z$ is continuously distributed on a bounded and connected
support with nonempty interior, and if $\g_{0}$ is restricted to
lie within a bounded set (which can be imposed as a scale normalization
without loss of generality), then $X_{it}^{'}\g_{0}|z$ is continuously
distributed on a bounded connected interval. If in addition $X_{it}|z$
is assumed to have a density that is strictly positive (almost) everywhere
on its support, then the induced density of $X_{it}^{'}\g_{0}|z$
will also be (almost) everywhere strictly positive. Note also Assumption
\ref{assu:Cts}(a) may also be satisfied if some (but not all) components
of $X_{it}$ are discrete, as long as some other component(s) of $X_{it}$
is continuously distributed with nonzero coefficient and a sufficiently
large support.

Assumption \ref{assu:Cts}(b), along with the assumptions of connectedness
(interval representation of the support) and strictly positive densities
(strictly increasing CDFs) for $X_{it}^{'}\g_{0}|z$ in Assumption
\ref{assu:Cts}(a), are imposed mainly as simplifying restrictions
that are not conceptually necessary but allow for a more convenient
notation. Essentially, they jointly imply that the per-period CCPs
on the left-hand and right-hand sides of \eqref{eq:ID_Set} are continuous
and strictly increasing in $c$ on connected intervals, leading to
simpler notation in the proof via the use of the inverse function
and the intermediate value theorem. Without these conditions, we would
need to handle ``flat regions'', ``jump points'', and ``continuously
increasing regions'' separately and then combine them together to
produce the final result, which  should be achievable using
a  combination of the proof techniques in the discrete case
and the continuous case.

Assumption \ref{assu:Cts}(c) is a key condition for the validity
of our adapted construction in the continuous case, but it is admittedly
the most nonstandard and implicit one, which warrants further explanation.
Effectively, Assumption 2(c) rules out certain ``knife-edge'' degenerate
DGPs that result in a ``flat region of contact'' between $\ol L$
and $\ul U$, though the exact form of such degeneracy can be rather
complicated given the nonlinear nature of the binary choice model
and the generality of the endogeneity we incorporate. Below we provide some intuition for why Assumption \ref{assu:Cts}(c) should be regarded as a relatively mild condition. For more detail, see Appendix \ref{subsec:cond_A2c} for two illustrative examples where Assumption \ref{assu:Cts}(c) (almost) trivially holds, as well as a general sufficient condition for Assumption \ref{assu:Cts}(c).  

Note that under Assumption \ref{assu:Cts}(a)(b), we have $L_{t}\left(c|z,\t_{0}\right)<U_{t}\left(c|z,\t_{0}\right)$
with strict inequality, so $\ol L\left(\rest cz,\t_{0}\right)\leq\ul U\left(\rest cz,\t_{0}\right)$
can only hold with equality if there exist two different periods $t\neq s$
such that $L_{t}\left(c|z,\t_{0}\right)=U_{s}\left(c|z,\t_{0}\right)$.
If this holds for all $c$'s in a small open interval, i.e. with positive
Lebesgue measure in violation of Assumption 2(c), then we can deduce
that their derivatives in $c$ must also match, i.e., 
\begin{equation}
L_{t}^{'}\left(\rest cz,\t_{0}\right)=U_{s}^{'}\left(\rest cz,\t_{0}\right)\label{eq:L'=00003DU'}
\end{equation}
on an open interval, with $L_{t}^{'}$ and $U_{s}^{'}$ given by
\begin{align}
L_{t}^{'}\left(\rest cz,\t_{0}\right) & =\P\left(Y_{it}=1|X_{it}^{'}\g_{0}=c-z_{t}^{'}\b_{0},Z_{i}=z\right)\pi_{t}\left(\rest{c-z_{t}^{'}\b_{0}}z\right),\nonumber \\
U_{s}^{'}\left(\rest cz,\t_{0}\right) & =\P\left(Y_{is}=0|X_{is}^{'}\g_{0}=c-z_{s}^{'}\b_{0},Z_{i}=z\right)\pi_{s}\left(\rest{c-z_{s}^{'}\b_{0}}z\right),\label{eq:LU_deriv}
\end{align}
where $\pi_{t}\left(\rest{\cd}z\right)$ denotes the conditional pdf
of $X_{it}^{'}\g_{0}$ given $Z_{i}=z$.

Consequently, $L_{t}^{'}=U_{s}^{'}$ on an open interval essentially
means that the density-weighted CCPs on the right-hand sides of \eqref{eq:LU_deriv}
must change continuously in $c$ in \emph{exactly the same functional
form on a continuum, despite all of the following}: (i) $L_{t}^{'}$
is defined on the event $Y_{it}=1$ while $U_{s}^{'}$ is defined
as on the event $Y_{is}=0$, which are ``flipped'' events that may
generally vary with $c$ in different manners, (ii) the values of
$z_{t}^{'}\b_{0}$ and $z_{s}^{'}\b_{0}$ can be different, so the
conditioning events are generally different for $L_{t}$ and $U_{s}$
as well, (iii) the conditional distribution of $X_{it}$ given $Z_{i}=z$
may be different (nonstationary) across periods $t$, so $\pi_{t}\left(\rest{c-z_{t}^{'}\b_{0}}z\right)$
and $\pi_{s}\left(\rest{c-z_{s}^{'}\b_{0}}z\right)$ may be different
even if $z_{t}^{'}\b_{0}=z_{s}^{'}\b_{0}$, (iv) the dependence structure
between $X_{it}$ and $\e_{it}$ may vary across $t$. For all these
reasons, it appears unlikely that \eqref{eq:L'=00003DU'}
can hold for a continuum of $c$, except under carefully designed ``knife-edge'' DGPs. Hence, we consider Assumption 2(c) as a mild condition. See Appendix \ref{subsec:cond_A2c} for more examples and details.

% Even if it is possible at all, it probably requires a very carefully designed ``knife-edge'' DGP for \eqref{eq:L'=00003DU'} to hold on a continuum.

% While we acknowledge that there might be an alternative proof approach
% that establishes sharpness in the continuous case under weaker conditions
% than those imposed in Assumption \ref{assu:Cts}, we hope that Assumption
% \ref{assu:Cts} and Theorem \ref{thm:sharp_cts} demonstrate the conceptual
% point that the inequality restrictions generated by our ``bounding-by-$c$''
% technique is able to extract the continuum of identifying information
% under the continuous case, which does not appear obviously true to
% us ex ante.

\subsection{\label{subsec:Recon}Reconciliation with Related Work}

Our identifying restrictions in \eqref{eq:ID_Set} and \eqref{eq:c_KTset}
have a somewhat ``nonstandard'' representation in terms of (conditional)
joint probabilities of $Y_{it}$ and $X_{it}$ (given $Z_{i}$), instead
of conditional probabilities of $Y_{it}$ given $X_{it}$ (such as
lagged outcomes), which are more usually found in the related literature.
Hence, we provide a more detailed discussion about the content and
interpretation of our identifying restrictions, as well as a more
explicit explanation of how they relate to existing results in the
related literature.

\subsubsection*{Reconciliation with \citet{manski1987}}

Consider first the special case where there are \emph{no} endogenous
covariates $X_{it}$, or in other words, $X_{it}$ is degenerate.
In this case, our ``partial stationarity'' condition specializes
to the ``full stationarity'' condition \eqref{eq:cond_sta} as in
\citet{manski1987}. However, our identifying restriction \eqref{eq:ID_Set}
still has a different form than the identifying restriction in \citet{manski1987}.
To illustrate, focus on any two periods $\left(t,s\right)$, and observe
that our identifying restriction becomes:
\begin{equation}
\P\left(\rest{Y_{it}=1,\ z_{t}^{'}\b_{0}\leq c}z\right)\leq1-\P\left(\rest{Y_{is}=0,\ z_{s}^{'}\b_{0}\geq c}z\right),\ \forall c,\label{eq:ID_ts}
\end{equation}
while the ``maximum-score-type'' identifying restrictions in \citet{manski1987}
are of the form 
\begin{equation}
z_{s}^{'}\b_{0}\geq z_{t}^{'}\b_{0}\ \iff\ \P\left(\rest{Y_{is}=1}z\right)\geq\P\left(\rest{Y_{it}=1}z\right).\label{eq:MaxScore}
\end{equation}
The ``maximum-score-type'' identifying restriction \eqref{eq:MaxScore}
has a quite intuitive and interpretable representation: across two
periods $\left(t,s\right)$ under full stationarity, the conditional
choice probability at period $s$ is larger if and only if the index
$z_{s}^{'}\b_{0}$ is larger. In contrast, our restriction \eqref{eq:ID_ts}
has a somewhat twisted representation even in this simple setting.

However, a closer look reveals that our \eqref{eq:ID_ts} reproduces Manski's ``maximum-score-type'' identifying restrictions
in the current context. To see this, notice that, by setting $c=z_{t}^{'}\b_{0}$
in \eqref{eq:ID_ts}, we obtain 
\begin{align*}
\P\left(\rest{Y_{it}=1}z\right)=\, & \P\left(\rest{Y_{it}=1}z\right)\ind\left\{ z_{t}^{'}\b_{0}\leq z_{t}^{'}\b_{0}\right\} \leq1-\P\left(\rest{Y_{is}=0}z\right)\ind\left\{ z_{s}^{'}\b_{0}\geq z_{t}^{'}\b_{0}\right\} 
\end{align*}
Hence, if $z_{s}^{'}\b_{0}\geq z_{t}^{'}\b_{0}$, i.e., the left-hand
side of \eqref{eq:MaxScore} holds, then the above implies that 
\[
\P\left(\rest{Y_{it}=1}z\right)\leq1-\P\left(\rest{Y_{is}=0}z\right)=\P\left(\rest{Y_{is}=1}z\right),
\]
which becomes exactly the right-hand side of \eqref{eq:MaxScore}.
Switching $t$ with $s$ in the argument above produces the other
implication $z_{s}^{'}\b_{0}\leq z_{t}^{'}\b_{0}$ $\imp$ $\P\left(\rest{Y_{is}=1}z\right)\leq\P\left(\rest{Y_{it}=1}z\right)$.
Together these exactly constitute the ``if-and-only-if'' restriction
in \eqref{eq:MaxScore}. Hence, even though our inequality restriction
\eqref{eq:ID_ts} looks different from the more intuitive ``maximum-score-type''
restriction, they both incorporate the same information.

\subsubsection*{Reconciliation with KPT}

Now, consider the dynamic AR(1) model as studied in KPT, where the
only endogenous covariate is the one-period lagged outcome variable,
i.e., $X_{it}:=Y_{i,t-1}$.

To illustrate, first focus on any two periods $\left(t,s\right)$
only, and observe that in this case our identifying restriction becomes
\begin{equation}
\P\left(\rest{Y_{it}=1,\ z_{t}^{'}\b_{0}+Y_{i,t-1}\g_{0}\leq c}z\right)\leq1-\P\left(\rest{Y_{is}=0,\ z_{s}^{'}\b_{0}+Y_{i,s-1}\g_{0}\geq c}z\right),\ \forall c.\label{eq:ID_ts_KPT}
\end{equation}
Under the same model and assumption, KPT derives the following 9 inequality
implications for $\left(t,s\right)$:\footnote{We adapt the notation in KPT to our current notation, and state these
9 inequalities as strict inequalities, which lead to a simpler and
more focused explanation. The equivalence between our restriction
and the KPT restrictions still hold if their inequalities are stated
in the weak form.}

KPT(i): $\P\left(\rest{Y_{it}=1}z\right)>\P\left(\rest{Y_{is}=1}z\right)$
$\imp$ $\left(z_{t}-z_{s}\right)^{'}\b_{0}+\left|\g_{0}\right|>0.$

KPT(ii): $\P\left(\rest{Y_{it}=1}z\right)>1-\P\left(\rest{Y_{i,s}=0,Y_{i,s-1}=1}z\right)$
$\imp$ $\left(z_{t}-z_{s}\right)^{'}\b_{0}-\min\left\{ 0,\g_{0}\right\} >0.$

KPT(iii): $\P\left(\rest{Y_{it}=1}z\right)>1-\P\left(\rest{Y_{i,s}=0,Y_{i,s-1}=0}z\right)$
$\imp$ $\left(z_{t}-z_{s}\right)^{'}\b_{0}+\max\left\{ 0,\g_{0}\right\} >0.$

KPT(iv): $\P\left(\rest{Y_{it}=1,Y_{it-1}=1}z\right)>\P\left(\rest{Y_{is}=1}z\right)$
$\imp$ $\left(z_{t}-z_{s}\right)^{'}\b_{0}+\max\left\{ 0,\g_{0}\right\} >0.$

KPT(v): $\P\left(\rest{Y_{it}=1,Y_{it-1}=1}z\right)>1-\P\left(\rest{Y_{is}=0,Y_{i,s-1}=1}z\right)$
$\imp$ $\left(z_{t}-z_{s}\right)^{'}\b_{0}>0.$

KPT(vi): $\P\left(\rest{Y_{it}=1,Y_{it-1}=1}z\right)>1-\P\left(\rest{Y_{is}=0,Y_{i,s-1}=0}z\right)$
$\imp$ $\left(z_{t}-z_{s}\right)^{'}\b_{0}+\g_{0}>0.$

KPT(vii): $\P\left(\rest{Y_{it}=1,Y_{it-1}=0}z\right)>1-\P\left(\rest{Y_{is}=0}z\right)$
$\imp$ $\left(z_{t}-z_{s}\right)^{'}\b_{0}-\min\left\{ 0,\g_{0}\right\} >0.$

KPT(viii): $\P\left(\rest{Y_{it}=1,Y_{it-1}=0}z\right)>1-\P\left(\rest{Y_{is}=0\text{\ensuremath{,}}Y_{i,s-1}=1}z\right)$
$\imp$ $\left(z_{t}-z_{s}\right)^{'}\b_{0}-\g_{0}>0.$

KPT(ix): $\P\left(\rest{Y_{it}=1,Y_{it-1}=0}z\right)>1-\P\left(\rest{Y_{is}=0,Y_{i,s-1}=0}z\right)$
$\imp$ $\left(z_{t}-z_{s}\right)^{'}\b_{0}>0.$

~

\noindent In a way, the 9 inequality restrictions in KPT above are
similar to the ``maximum-score restrictions'', in the sense that
all of them take the form of logical implications between intertemporal
comparisons of various conditional probabilities and intertemporal
differences of covariate indexes.

Using a very different identification strategy than the one in KPT,
we arrived at our inequality restriction \eqref{eq:ID_ts_KPT}, which
looks very different from the collection of 9 inequality restrictions
in KPT. At first sight it is not clear how \eqref{eq:ID_ts_KPT} relates
to and compares with the 9 KPT restrictions. However, a closer look
again reveals that our restriction \eqref{eq:ID_ts_KPT} can reproduce
all the 9 restrictions in KPT, and thus incorporate all the information
therein in a unified format.

Take KPT(v) as an illustration and suppose that the left-hand side
of KPT(v) holds, then it implies 
\begin{equation}
\P\left(\rest{Y_{it}=1,Y_{it-1}=1}z\right)>1-\P\left(\rest{Y_{is}=0,Y_{i,s-1}=1}z\right).\label{eq:KPTv_LHS}
\end{equation}

With $X_{it}=Y_{i,t-1}$, our inequality restriction \eqref{eq:ID_ts_KPT}
can be equivalently rewritten as follows, 
\begin{align}
 & \P\left(\rest{Y_{it}=1,\ Y_{i,t-1}=1}z\right)\ind\left\{ z_{t}^{'}\b_{0}+\g_{0}\leq c\right\} +\P\left(\rest{Y_{it}=1,\ Y_{i,t-1}=0}z\right)\ind\left\{ z_{t}^{'}\b_{0}\leq c\right\} \nonumber \\
\leq\  & 1-\P\left(\rest{Y_{is}=0,\ Y_{i,s-1}=1}z\right)\ind\left\{ z_{s}^{'}\b_{0}+\g_{0}\geq c\right\} -\P\left(\rest{Y_{is}=0,\ Y_{i,s-1}=0}z\right)\ind\left\{ z_{s}^{'}\b_{0}\geq c\right\} ,\label{eq:ID_KPT_expand}
\end{align}
where the realization of $Y_{i,t-1}$ is explicitly enumerated as
in KPT.

Note that we can further relax condition \eqref{eq:ID_KPT_expand}
by dropping the two probabilities $\P\left(\rest{Y_{it}=1,\ Y_{i,t-1}=0}z\right)\ind\left\{ z_{t}^{'}\b_{0}\leq c\right\} $
and $\P\left(\rest{Y_{is}=0,\ Y_{i,s-1}=0}z\right)\ind\left\{ z_{s}^{'}\b_{0}\geq c\right\} $
as it makes the lower bound smaller and the upper bound larger: 
\begin{align*}
 & \P\left(\rest{Y_{it}=1,\ Y_{i,t-1}=1}z\right)\ind\left\{ z_{t}^{'}\b_{0}+\g_{0}\leq c\right\} \\
\leq\  & 1-\P\left(\rest{Y_{is}=0,\ Y_{i,s-1}=1}z\right)\ind\left\{ z_{s}^{'}\b_{0}+\g_{0}\geq c\right\} .
\end{align*}
Then, the statement that $\ind\left\{ z_{t}^{'}\b_{0}+\g_{0}\leq c\right\} $
and $\ind\left\{ z_{s}^{'}\b_{0}+\g_{0}\geq c\right\} $ both hold
is precisely equivalent to the following statement of 
\[
z_{t}^{'}\b_{0}\leq z_{s}^{'}\b_{0}\ \imp\ \P\left(\rest{Y_{it}=1,\ Y_{i,t-1}=1}z\right)\leq\ 1-\P\left(\rest{Y_{is}=0,\ Y_{i,s-1}=1}z\right).
\]
By contraposition, it leads to exactly the same implication of KPT(v):
\[
\P\left(\rest{Y_{it}=1,\ Y_{i,t-1}=1}z\right)>\ 1-\P\left(\rest{Y_{is}=0,\ Y_{i,s-1}=1}z\right)\Longrightarrow z_{t}^{'}\b_{0}>z_{s}^{'}\b_{0}.
\]
Hence, we have shown that \eqref{eq:ID_KPT_expand} implies KPT(v).

Similarly, it is shown in Appendix \ref{appe:RelateKPT} that \eqref{eq:ID_KPT_expand}
implies all 9 restrictions in KPT. In fact, the representation \eqref{eq:ID_KPT_expand}
reveals why there are precisely 9 KPT-type restrictions. The two period-$t$
indicators $\ind\left\{ z_{t}^{'}\b_{0}+\g_{0}\leq c\right\} $ and
$\ind\left\{ z_{t}^{'}\b_{0}\leq c\right\} $ in the upper expression
of \eqref{eq:ID_KPT_expand} may take 3 ``useful''\footnote{The 4th combination, $\ind\left\{ z_{t}^{'}\b_{0}+\g_{0}\leq c\right\} =\ind\left\{ z_{t}^{'}\b_{0}\leq c\right\} =0$,
will make the upper expression of \eqref{eq:ID_KPT_expand} equal
to $0$, so that the inequality \eqref{eq:ID_KPT_expand} holds trivially.
Hence, this $\left(0,0\right)$ combination is not useful.} combinations $\left(1,0\right),\left(0,1\right)$ and $\left(1,1\right)$,
while the two period-$s$ indicators $\ind\left\{ z_{s}^{'}\b_{0}+\g_{0}\geq c\right\} $
and $\ind\left\{ z_{s}^{'}\b_{0}\geq c\right\} $ in the lower expression
of \eqref{eq:ID_KPT_expand} may also take 3 useful combinations.
Consequently, in total there are $3\times3=9$ useful combinations,
which exactly correspond to the $9$ left-hand-side suppositions in
the 9 KPT restrictions.

Hence, while our restriction \eqref{eq:ID_ts_KPT} appears very different
from the 9 KPT restrictions, it actually automatically incorporates
all the KPT restrictions. In particular, by treating the endogenous
covariate $X_{it}=Y_{i,t-1}$ as a random variable, our restriction
\eqref{eq:ID_ts_KPT} automatically aggregates the identifying information
across all possible realizations of $Y_{i,t-1}$, without the need
to explicitly consider each possibility separately.

Now, consider a general setting with $T\geq2$ periods. By our Proposition
\ref{prop:disc} and Theorem \ref{thm:sharp_disc}, the sharp identified
set can be characterized by $2T$ restrictions, which are generated
by evaluating \eqref{eq:ID_Set} at each $c$ of the $2T$ points
in $\left\{ z_{t}^{'}\b,z_{t}^{'}\b+\g:t=1,...,T\right\} $. In contrast,
across $T$ periods the KPT approach produces $9T\left(T-1\right)$
restrictions, which are generated by imposing the $9$ KPT restrictions
across all ordered time pairs $\left(t,s\right)$. Hence, our approach
provides a much simpler characterization of the sharp identified set,
using a significantly smaller number of restrictions. For example,
with $T=2$ periods, we have $4$ restrictions while KPT has $18$;
with $T=3$, we have $6$ restrictions while KPT has 54. Hence, the
reduction in the number of restrictions relative to KPT is quite remarkable.

~

In summary, while the representation of our identifying restrictions
in Propositions \ref{prop:bin} and \ref{prop:disc} may appear somewhat
unusual in the first place, it actually becomes equivalent to the more
familiar representations in the specialized settings of \citet{manski1987}
and KPT.

\section{\label{sec:gene}Extensions}

The key idea underlying our identification strategy extends beyond the binary choice model and can be applied to a wide range of nonlinear panel data models with dynamics and endogeneity. We first present our general identification strategy in a
generic semiparametric model (Section 3.1), and then demonstrate how this strategy
can be applied and adapted to ordered response (Section 3.2), multinomial
choice (Section 3.3) and censored outcome (Appendix B.3) settings.

\subsection{Identification Strategy in a Generic Model\label{subsec:generic}}

We start with a generic semiparametric model that helps convey the
generality of our key identification strategy
\begin{align}
Y_{it} & =G\left(W_{it}^{'}\t_{0},\,\a_{i},\,\e_{it}\right),\label{eq:model_gen}
\end{align}
where $Y_{it}\in\mathcal{Y}$ can be either a discrete or continuous
variable, $\a_{i}$ is the individual fixed effect of arbitrary dimension,
$\e_{it}$ is the time-varying error of arbitrary dimension, $W_{it}$
is a vector of observable covariates, $\t_{0}\in\mathcal{R}^{d_{w}}$
is a conformable vector of parameters, and the function $G$ is allowed
to be unknown, nonseparable but assumed to satisfy the following:
\begin{assumption}[Index Monotonicity]
 \label{assu:Mono} The mapping $\d\longmapsto G\left(\d,\a,\e\right)$
is weakly increasing in $\d\in\mathcal{R}$ for each realization of
$\left(\a,\e\right)$. 
\end{assumption}
Note that, we can obtain the binary choice model in Section \ref{sec:bin}
by setting $\alpha_{i},\e_{it}$ to be scalar-valued, and $G\left(W_{it}^{'}\t_{0},\,\a_{i},\,\e_{it}\right)=\ind\left\{ W_{it}^{'}\t_{0}+\a_{i}+\e_{it}\geq0\right\} $,
where $G$ is by construction weakly increasing in $W_{it}^{'}\t_{0}$. 

As before, we decompose $W_{it}$, and correspondingly
$\t_{0}$, into two components, $W_{it}=\left(Z_{it}^{'},X_{it}^{'}\right)^{'}$
and $\t_{0}=\left(\b_{0}^{'},\g_{0}^{'}\right)^{'}$, and impose the
partial stationarity condition (Assumption \ref{assu:PartStat}).
We now show how partial stationarity can be exploited in conjunction
with weak monotonicity (Assumption \ref{assu:Mono}) to obtain identifying
restrictions in the presence of endogeneity.

Let $\mathcal{Y}$ denote the support of $Y_{it}$. For any $c\in\mathcal{R}$
and $y\in\mathcal{Y}$, observe that
\[
\begin{aligned}\ind\left\{ Y_{it}\leq y,\ W_{it}^{'}\t_{0}\geq c\right\}  & =\ind\left\{ G\left(W_{it}^{'}\t_{0},\,\a_{i},\,\e_{it}\right)\leq y,\ W_{it}^{'}\t_{0}\geq c\right\} \\
 & \leq\ind\left\{ G\left(c,\,\a_{i},\,\e_{it}\right)\leq y\right\} ,
\end{aligned}
\]
where the inequality holds by the monotonicity of the function $G$.
Symmetrically, we have
\[
\begin{aligned}\ind\left\{ Y_{it}>y,\ W_{it}^{'}\t_{0}\leq c\right\}  & =\ind\left\{ G\left(W_{it}^{'}\t_{0},\,\a_{i},\,\e_{it}\right)>y,\ W_{it}^{'}\t_{0}\leq c\right\} \\
 & \leq\ind\left\{ G\left(c,\,\a_{i},\,\e_{it}\right)>y\right\} \\
 & =1-\ind\left\{ G\left(c,\,\a_{i},\,\e_{it}\right)\leq y\right\} .
\end{aligned}
\]
which is equivalent to 
\[
\ind\left\{ G\left(c,\,\a_{i},\,\e_{it}\right)\leq y\right\} \leq1-\ind\left\{ Y_{it}>y,\ Z_{it}^{'}\b_{0}+X_{it}^{'}\g_{0}\leq c\right\} .
\]
The partial stationarity assumption $\e_{it}\mid Z_{i},\a_{i}\sim\e_{is}\mid Z_{i},\a_{i}$
implies the stationarity of the transform function $G$: $G\left(c,\a_{i},\e_{it}\right)\mid Z_{i},\a_{i}\sim G\left(c,\a_{i},\e_{is}\right)\mid Z_{i},\a_{i}$.
After integrating out $\a_{i}$, the stationarity condition persists
conditioned on $Z_{i}$ alone: 
\[
G\left(c,\a_{i},\e_{it}\right)\mid Z_{i}\sim G\left(c,\a_{i},\e_{is}\right)\mid Z_{i}.
\]
Combining the above derived bounds on $\ind\{G\left(c,\,\a_{i},\,\e_{it}\right)\leq y\}$,
we have

\begin{equation}
\begin{aligned} & \P\left(Y_{it}\leq y,\ Z_{it}^{'}\b_{0}+X_{it}^{'}\g_{0}\geq c\mid z\right)\\
\leq \  & \P\left(G\left(c,\,\a_{i},\,\e_{it}\right)\leq y\mid z\right)=\P\left(G\left(c,\,\a_{i},\,\e_{is}\right)\leq y\mid z\right)\\
\leq\  & 1-\P\left(Y_{is}>y,\ Z_{is}^{'}\b_{0}+X_{is}^{'}\g_{0}\leq c\mid z\right)=:U_{s}\left(\rest{c,y}z,\t_{0}\right)
\end{aligned}
\label{eq:Bounds_z}
\end{equation}
The key difference of the above and the corresponding identifying
restrictions in Section 2 lies in that the ``middle term'' in \eqref{eq:Bounds_z}
is no longer the conditional CDF of $\a_{i}+\e_{it}$, but the conditional
probability of $G\left(c,\,\a_{i},\,\e_{is}\right)\leq y$, with the
latter representation not dependent on scalar-additivity of fixed
effect $\a_{i}$ and time-varying errors $\e_{it}$. 

We summarize the identifying restrictions derived above by the following
proposition: 
\begin{prop}
\label{prop:gene} Define $\T_{I,gen}$ as the set of all $\t\in\R^{d_{w}}$
such that 
\begin{align}
\max_{t}\P\left(Y_{it}\leq y,\ Z_{it}^{'}\b+X_{it}^{'}\g\geq c\mid z\right)\leq1-\max_{s}\P\left(Y_{is}>y,\ Z_{is}^{'}\b+X_{is}^{'}\g\leq c\mid z\right),\label{eq:gen_id_ineq}
\end{align}
where for any $c\in\mathcal{R},$ $y\in\mathcal{Y}$, and any $z$.
Under model \eqref{eq:model_gen}, Assumptions \ref{assu:PartStat}
and \ref{assu:Mono}, $\t_{0}\in\T_{I,gen}$. 
\end{prop}
Note that in the binary choice setting of Section \ref{sec:bin},
it suffices to set $y=0$ in \eqref{eq:gen_id_ineq}, which then coincides
with the identifying results in Proposition \ref{prop:bin}. This also shows that the identified set does not change  at all, regardless of whether scalar-additivity of $\a_i$ and $\e_{it}$ is imposed or not in the binary choice model.

The results in Proposition \ref{prop:gene} generally hold regardless
of whether the dependent variable and the endogenous covariate are
discrete or continuous. The next proposition shows that additional
discreteness in either the dependent variable or endogenous covariates
can further simplify and reduce the number of the identifying conditions
in \eqref{eq:gen_id_ineq}.
\begin{prop}
\label{prop:gene_dis} When $X_{it}\in\left\{ \ol x_{1},...,\ol x_{K}\right\} $
for any $t$, then $\T_{I,gen}=\T_{I,gen}^{disc_{x}},$ where $\T_{I,gen}^{disc_{x}}$
consists of all $\t=\left(\b^{'},\g^{'}\right)^{'}$ that satisfy
condition \eqref{eq:gen_id_ineq} for any $c\in\left\{ z_{t}^{'}\b+\ol x_{k}^{'}\g:k=1,...,K,t=1,...,T\right\} $.
\end{prop}
Moreover, when $Y_{it}\in\left\{ \ol y_{1},...,\ol y_{K}\right\} $
with $\ol y_{j}\leq\ol y_{j+1}$ for any $t$, then $\T_{I,gen}=\T_{I,gen}^{disc_{y}},$
where $\T_{I,gen}^{disc_{y}}$ consists of all $\t=\left(\b^{'},\g^{'}\right)^{'}$
that satisfy condition \eqref{eq:gen_id_ineq} for any $y\in\left\{ \ol y_{1},...,\ol y_{K-1}\right\} $. 

Proposition \ref{prop:gene_dis} shows that for the general model,
when both the outcome and the endogenous variable are discrete, it
is sufficient to focus on a finite number of identifying restrictions.
The number of these restrictions is determined by the support of the
outcome variable and the covariate index. The proof of Proposition \ref{prop:gene_dis} follows the same reasoning
as Proposition \ref{prop:disc}, so it is omitted here. The central
idea is that for any point $c$ or $y$ outside the range specified
in Proposition \ref{prop:gene_dis}, we can find a point within the
specified range that provides weakly more informative results. Therefore,
the inclusion of these outside points would not provide additional
information for the identified set. 

\begin{rem} It is  natural to ask whether sharpness can be established in this general setup. While we do not present a formal result, we provide a discussion of this  in Appendix \ref{subsec:sharp_gen}.
\end{rem}

\subsection{Ordered Response Model\label{subsec:order}}

Consider that the outcome variable $Y_{it}$ takes $J$ ordered values:
$Y_{it}\in\left\{ y_{1},..,y_{J}\right\} $ with $y_{j}<y_{j+1}$.
Examples of such ordered outcomes include various income categories,
health outcomes, or levels of educational attainment. We study the
following panel ordered choice model: 
\begin{equation}
\begin{aligned}Y_{it}^{*} & =W_{it}'\t_{0}+v_{it},\\
Y_{it} & =\sum_{j=1}^{J}y_{j}\ind\left\{ b_{j}<Y_{it}^{*}\leq b_{j+1}\right\} ,
\end{aligned}
\label{model:order}
\end{equation}
where $Y_{it}^{*}$ denotes the latent dependent variable, and $Y_{it}$
denotes the ordered outcome which takes value $y_{j}$ when $Y_{it}^{*}\in(b_{j},b_{j+1}]$.
The threshold parameters satisfy $b_{1}=-\infty,b_{J+1}=+\infty$,
and the remaining threshold parameters $b_{j}$ (where $b_{j}\leq b_{j+1}$)
can be either known or unknown for $2\leq j\leq J-1$. The binary
choice model in \eqref{eq:bin} is nested with $J=2$ and $b_{2}=0$.

While the ordered response model \eqref{model:order} here can be
regarded as a special case of the generic model \eqref{eq:model_gen},
the special ``ordered cutoffs'' structure in \eqref{model:order}
contains more information than an unknown generic $G$ function in
\eqref{eq:model_gen}. As a result, even though the general identification
strategy in Section \ref{subsec:generic} still applies, we can adapt
the identification argument to the special additional structure imposed
here, obtaining a sharper result than a direct application of Proposition
\ref{prop:gene}. In particular, we explain why the line of our identification
arguments help us find such an adaption that exploits the special
model structure.

We now explain this in more details. Following the arguments in Section
\ref{subsec:generic}, we have
\[
\ind\left\{ Y_{it}\leq y_{j},\ b_{j+1}-W_{it}^{'}\t_{0}\leq c\right\} \leq\ind\left\{ v_{it}\leq c\right\} 
\]
For a given $c$, the above inequality holds for any response index
$j$. This immediately implies that we can take the largest one to
get a tighter lower bound:
\[
\max_{j}\ind\left\{ Y_{it}\leq y_{j},\ b_{j+1}-W_{it}^{'}\t_{0}\leq c\right\} \leq\ind\left\{ v_{it}\leq c\right\} .
\]
In addition, an inspection of the LHS reveals that the maximum is
attained at 
\[
j=\ol j_{c}\left(W_{it}\right):=\max\left\{ j:b_{j+1}-W_{it}^{'}\t_{0}\leq c\right\} 
\]
since such a (random) $j$ would maximize $\ind\left\{ Y_{it}\leq y_{j}\right\} $
subject to $b_{j+1}-W_{it}^{'}\t_{0}\leq c$. Consequently, we obtain
\begin{align}
\ind\left\{ v_{it}\leq c\right\}  & \geq\ind\left\{ Y_{it}\leq y_{\ol j_{c}\left(W_{it}\right)},\ b_{\ol j_{c}\left(W_{it}\right)+1}-W_{it}^{'}\t_{0}\leq c\right\} \nonumber \\
 & =\sum_{j=1}^{\ol j_{c}\left(W_{it}\right)}\ind\left\{ Y_{it}=y_{j},\ b_{\ol j_{c}\left(W_{it}\right)+1}-W_{it}^{'}\t_{0}\leq c\right\} \nonumber \\
 & =\sum_{j=1}^{J}\ind\left\{ Y_{it}=y_{j},\ b_{j+1}-W_{it}^{'}\t_{0}\leq c\right\} \label{eq:order_LB}
\end{align}
where the last equality holds since $\ind\left\{ b_{j+1}-W_{it}^{'}\t_{0}\leq c\right\} =0$
for any choice $j>\ol j_{c}\left(W_{it}\right)$. 

The final expression \eqref{eq:order_LB} is particularly nice for
three reasons: First, it aggregates the information aggregated from
different $y_{j}$ together to produce a tighter lower bound. Second,
the expression circumvent the need to compute the maximizer cutoff
$\ol j_{c}$. Third, it is represented as a linear sum (instead of
a maximum) so that conditional expectation of \eqref{eq:order_LB}
remains a linear sum of conditional expectations.

To see the advantage of the third point above, we take conditional
expectation of \eqref{eq:order_LB} given $z$ as before, obtaining
\[
\P\left(\rest{v_{it}\leq c}z\right)\geq\sum_{j=1}^{J}\P\left(\rest{Y_{it}=y_{j},\ b_{j+1}-W_{it}^{'}\t_{0}\leq c}z\right).
\]
where the RHS can be computed as a simple sum of CCPs about each ordered
value $y_{j}$. 

Similarly, we can derive an upper bound 
\[
\P\left(v_{is}\leq c\mid z\right)\leq1-\sum_{j=1}^{J}\P\left(\rest{Y_{is}=y_{j},b_{j}-W_{is}^{'}\t_{0}\geq c}z\right),
\]
which can be combined with the lower bound to yield the following
result. 
\begin{prop}
\label{prop:order} Define $\T_{I,order}$ as the set of $\t=\left(\b^{'},\g^{'}\right)^{'}$
such that 
\begin{align}
 & \max_{t=1,...,T}\sum_{j=1}^{J}\P\left(Y_{it}=y_{j},b_{j+1}-z_{t}^{'}\b-X_{it}^{'}\g\leq c\mid z\right)\nonumber \\
\leq\  & 1-\max_{s=1,...,T}\sum_{j=1}^{J}\P\left(Y_{is}=y_{j},b_{j}-z_{s}^{'}\b-X_{is}^{'}\g\geq c\mid z\right),\label{eq:order}
\end{align}
for any $c\in\mathcal{R}$ and any realization $z$ in the support
of $Z_{i}$. Under Assumptions \ref{assu:PartStat}, $\t_{0}\in\T_{I,order}$.
\end{prop}
We emphasize again that Proposition \ref{prop:order} is \emph{not
}a direct application of Proposition \ref{prop:gene}, since Proposition
\ref{prop:order} explicitly utilizes the special model structure
of the order response model to aggregate information from all response
index $j$ together to form tighter bounds for each $c$. In contrast,
a naive application of Proposition \ref{prop:gene} would yield 
\begin{multline*}
\max_{t=1,...,T}\P\left(Y_{it}\leq y_{j},b_{j+1}-z_{t}^{'}\b-X_{it}^{'}\g\leq c\mid z\right)\\
\leq1-\max_{s=1,...,T}\P\left(Y_{is}>y_{j},b_{j}-z_{s}^{'}\b-X_{is}^{'}\g\geq c\mid z\right),\forall j,\ \forall\left(c,z\right)
\end{multline*}
which remains valid but is a collection of bounds imposed on each
$j$ separately, thus is generally not as tight as the bounds in \eqref{eq:order}.

\subsection{Multinomial Choice Model \label{subsec:mul}}

In this subsection, we apply our key identification strategy to panel
multinomial choice model with endogeneity. Specifically, consider
a set of unordered choice alternatives $\mathcal{J}=\left\{ 0,1,...,J\right\} $.
Let $u_{ijt}$ denote the latent utility for individual $i$ of selecting
choice $j$ at time $t$, which depends on the three components: observed
covariate $W_{ijt}=(Z_{ijt}',X_{ijt}')'$, unobserved fixed effects
$\alpha_{ij}$, and unobserved time-varying preference shock $\e_{ijt}$.
%The utility of choice 0 (outside option) is normalized to zero: $u_{i0t}=0$.
Let $Y_{it}\in\mathcal{J}$ denote individual $i$'s choice at time
$t$. We study the following panel multinomial choice model: 
\[
\begin{aligned}u_{ijt} & =W_{ijt}^{'}\t_{0}+\alpha_{ij}+\epsilon_{ijt},\\
Y_{it} & =\arg\max_{j\in\mathcal{J}}u_{ijt},
\end{aligned}
\]
and impose the same partial stationarity assumption: 
\[
\e_{is}\mid Z_{i},\a_{i}\stackrel{d}{\sim}\e_{it}\mid Z_{i},\a_{i}\quad\text{for any}\ s,t\leq T.
\]
with $Z_{it}:=\{Z_{ijt}\}_{j\in\mathcal{J}}$,$\a_{i}:=\{\a_{ij}\}_{j\in\mathcal{J}}$
and $\e_{it}:=\{\e_{ijt}\}_{j\in\mathcal{J}}$ defined to collect
terms across all $J$ choice alternatives.

We emphasize that this model is not a special case of the generic
model \eqref{subsec:generic} in Subsection \eqref{subsec:generic},
since in the current model the $J$ outcome values are unordered,
and the model involves multiple indexes and multivariate monotonicity.
Hence, we cannot directly apply Proposition \ref{prop:gene} to the
current setting. That said, we explain how the key idea from Subsection
\eqref{subsec:generic} can again be adapted to obtain identification
result in the panel multinomial choice setting.

We start by looking at the indicator variable $Y_{it}^{j}:=\ind\{Y_{it}=j\}$
of choosing alternative $j$, which maintains a similar monotone structure
with Assumption \ref{assu:Mono}: 
\begin{align*}
Y_{it}^{j}=1\  & \iff\ W_{ijt}'\t_{0}+\alpha_{ij}+\epsilon_{ijt}\geq W_{ikt}'\t_{0}+\alpha_{ik}+\epsilon_{ikt},\ \forall k\in\mathcal{J}\\
 & \iff\ W_{ijt}'\t_{0}-W_{ikt}'\t_{0}\geq\alpha_{ik}+\epsilon_{ikt}-\alpha_{ij}-\epsilon_{ijt},\ \forall k\in\mathcal{J}
\end{align*}
and the new variable $Y_{it}^{j}$ is increasing in $W_{ijt}'\t_{0}-W_{ikt}'\t_{0}$
$\forall k\in\mathcal{J}$

More generally, for any subset $K\subset\mathcal{J}$, the indicator
variable $Y_{it}^{K}:=\ind\{Y_{it}\in K\}$ represents individual
$i$'s choice belonging to the subset $K$, given by
\begin{align*}
Y_{it}^{K}=1\  & \iff\ W_{ijt}'\t_{0}+\alpha_{ij}+\epsilon_{ijt}\geq W_{ikt}'\t_{0}+\alpha_{ik}+\epsilon_{ikt},\ \exists j\in K,\forall k\in\mathcal{J}\setminus K,\\
 & \iff\ W_{ijt}'\t_{0}-W_{ikt}'\t_{0}\geq\alpha_{ik}+\epsilon_{ikt}-\alpha_{ij}-\epsilon_{ijt}\ \exists j\in K,\forall k\in\mathcal{J}\setminus K
\end{align*}
and the variable $Y_{it}^{K}$ is increasing in $W_{ijt}'\t_{0}-W_{ikt}'\t_{0}$
for any $j\in K$ and $k\in\mathcal{J}\setminus K$.

Adapting the key identification idea to the current context, the identification results
for panel multinomial choice models are presented in the following
proposition.
\begin{prop}
\label{prop:mul} Define $\T_{I,mul}$ as the set that consists of all $\t=\left(\b^{'},\g^{'}\right)^{'}$
such that 
\begin{align}
 & \max_{t=1,...,T}\P\left(\rest{Y_{it}^{K}=1,\left(W_{ijt}-W_{ikt}\right)^{'}\t\leq c_{jk}\ \forall j\in K,k\in\mathcal{J}\setminus K}z\right)\nonumber \\
\leq & 1-\max_{t=1,...,T}\P\left(\rest{Y_{is}^{K}=0,\left(W_{ijs}-W_{iks}\right)^{'}\t\geq c_{jk}\ \forall j\in K,k\in\mathcal{J}\setminus K}z\right),\label{eq:ind_mul}
\end{align}
for any subset $K\subset\mathcal{J}$, any $c_{jk}\in\R$, any $j\in K$
and $k\in\mathcal{J}\setminus K$, and any realization $z$ in the
support of $Z_{i}$. Then, under Assumption \ref{assu:PartStat},
$\t_{0}\in\T_{I,mul}$.
\end{prop}

Below we show that Proposition \ref{prop:mul} specializes to the
corresponding result in \citet{pakes2019}, who focuses on the static
panel multinomial choice model without any endogeneity. Since \citet{pakes2019}
establishes the sharpness of their identification result under their
setup, our Proposition \ref{prop:mul} is also sharp (under their
static two-period setting).

However, a key improvement of our result relative to that in \citet{pakes2019}
is that Proposition \ref{prop:mul} allows for any type of endogeneity
including dynamic multinomial models with lagged dependent variable,
as well as the inclusion of contemporaneously endogenous variables
such as product prices. For example, consider the following dynamic
model: 
\[
u_{ijt}=Z_{ijt}'\beta_{0}+\ind\left\{ Y_{i,t-1}=j\right\} \g_{0,j}-P_{ijt}\l_{0}+\alpha_{ij}+\epsilon_{ijt}.
\]
where individual $i$'s utility at time $t$ can potentially depend
on their choices in the previous period $t-1$ and we allow the dynamic
effect $\g_{0,j}$ to vary across choices, and $P_{ijt}$ is the price
of product $j$ faced by consumer $i$ at time $t$. To the best of our knowledge,
no previous work has considered such generalization of \citet{pakes2019}
that can incorporate price endogeneity and past-choice dependence.
Even though Proposition \ref{prop:mul} is presented as a byproduct
of our general identification strategy, it nevertheless presents a
substantive progress in the related literature on panel multinomial
choice models.

\subsubsection*{Reconciliation with \citet{pakes2019}}

\label{subsec:mul_app} Next, we show that Proposition \ref{prop:mul}
specializes to those in \citet{pakes2019}, who focus on the static
panel multinomial choice model without any endogeneity. Since
\citet{pakes2019} establishes the sharpness of their identification
set in a two-period setting and our identification set reproduces theirs, the sharpness of our identification set follows immediately in this setting.

Formally, \citet{pakes2019} characterizes the sharp identified set
for $\t_{0}$ under the full stationarity assumption given all covariates:
\[
\e_{is}\mid W_{i},\a_{i}\stackrel{d}{\sim}\e_{it}\mid W_{i},\a_{i}.
\]
Under this condition, for two periods $\left(t,s\right)$ our identifying
condition in \eqref{eq:ind_mul} is simplified to
\begin{align}
 & \P\left(Y_{is}^{K}=1,(w_{js}-w_{ks})'\t_{0}\leq c_{jk}\ \forall j\in K,k\in\mathcal{J}\setminus K\mid w\right)\nonumber \\
\leq\  & 1-\P\left(Y_{it}^{K}=0,(w_{jt}-w_{kt})'\t_{0}\geq c_{jk}\ \forall j\in K,k\in\mathcal{J}\setminus K\mid w\right)\label{eq:mul_sta}
\end{align}
The above equation is only informative when $(w_{js}-w_{ks})'\t_{0}\leq c_{jk}\leq(w_{jt}-w_{kt})'\t_{0}$
for any $j\in K,k\in k\in\mathcal{J}\setminus K$; otherwise either
the upper bound becomes one or the lower bound becomes zero so that
condition \eqref{eq:mul_sta} holds for any $\t$. There exists one
value $c_{jk}$ satisfying the condition $(w_{js}-w_{ks})'\t_{0}\leq c_{jk}\leq(w_{jt}-w_{kt})'\t_{0}$
is equivalent to $(w_{js}-w_{ks})'\t_{0}\leq(w_{jt}-w_{kt})'\t_{0}$,
generating the following inequality: for any $K\subset\mathcal{J}$,
\begin{align*}
\text{If \ensuremath{\quad}} & (w_{js}-w_{ks})'\t_{0}\leq(w_{jt}-w_{kt})'\t_{0}\ \ \forall j\in K,k\in\mathcal{J}\setminus K\\
\text{then\ensuremath{\quad}} & \P\left(Y_{it}^{K}=0\mid w\right)\leq1-\P\left(Y_{is}^{K}=1\mid w\right)
\end{align*}
which becomes the same result in \citet{pakes2019} (Proposition 1,
P. 12): 
\begin{align*}
\text{If \ensuremath{\quad}} & (w_{js}-w_{ks})'\t_{0}\leq(w_{jt}-w_{kt})'\t_{0}\ \ \forall j\in K,k\in\mathcal{J}\setminus K\\
\text{then\ensuremath{\quad}} & \P\left(Y_{is}\in K\mid w\right)\leq\P\left(Y_{it}\in K\mid w\right)
\end{align*}
since $Y_{it}^{K}=1$ is equivalent to $Y_{it}\in K$ by the definition.

\section{Simulation}

\label{sec:simu}

In this section, we focus on the static ordered response model Section
\ref{subsec:order} and implement the kernel-based CLR inference approach
proposed in the papers by \citet*{chernozhukov2013} and \citet{chen2019breaking},
which was developed to construct confidence interval based on general
conditional moment inequalities. 

In Appendix \ref{subsec:Q_gamma}, we also conduct a simulation exercise of a different nature. We numerically compute and visualize the identified set under two DGP configurations in dynamic binary choice setting, but do not implement the finite-sample estimation and inference procedure. 

\subsection{Static Ordered Response Model}

\label{subsec:static}

This section explores a static ordered choice model with three choices
$Y_{it}\in\{1,2,3\}$. We consider the following two-period model
with $T=2$, and the latent dependent variable $Y_{it}^{*}$ is generated
as: 
\[
Y_{it}^{*}=Z_{it}^{1}\b_{01}+Z_{it}^{2}\b_{02}+\a_{i}+\e_{it},
\]
where the covariate $Z_{it}^{k}$ satisfies $Z_{it}^{k}\sim\mathcal{N}(0,\s_{z})$
for $k\in\{1,2\}$; the fixed effects $\a_{i}$ are given as $\a_{i}=\sum_{t=1}^{T}(Z_{it}^{1}+Z_{it}^{2})/(4*\s_{z}*T)$,
so they are correlated with the covariates; the error term $(\e_{i1},\e_{i2})$
follows the normal distribution $\mathcal{N}(\mu,\Sigma)$ with $\mu=(0,0)$
and $\Sigma=(1\ \rho;\rho\ 1)$. The true parameter is $\b_{0}:=(\b_{0,1},\b_{02})'=(1,1)'$,
the repetition number is $B=200$, and the sample size is $n=\{2000,8000\}$.
We consider three specifications for $\s_{z}\in\{1,1.5,2\}$ and $\rho\in\{0,0.25,0.5\}$.

The observed dependent variable $Y_{it}$ is given as 
\[
Y_{it}=1*(Y_{it}^{*}\leq b_{2})+2*(b_{2}<Y_{it}^{*}\leq b_{3})+3*(Y_{it}^{*}>b_{3}),
\]
where $b_{2}=-1$ and $b_{3}=1$. We write $Y_{i}:=(Y_{i1},Y_{i2})$ and $Z_{i}:=(Z_{i1},Z_{i2})$, 

Proposition \ref{prop:static_ordered} in Appendix \ref{subsec:order_stat}, as a corollary of Proposition \ref{prop:order} with pivotal choices of $c$'s, provides a characterization of the identified set for $\b_{0}$ under the static ordered choice model via
the following conditional moment inequalities: for $s\neq t\leq2$,
\[
E[g(Z_{i},Y_{i};\b_{0})\mid z]\geq0,
\]
where 
\[
g(Z_{i},Y_{i};\b)=\left\{ \begin{aligned} & \ind\{b_{2}-Z_{is}'\b\geq b_{2}-Z_{it}'\b\}(\ind\{Y_{is}=1\}-\ind\{Y_{it}=1\});\\
 & \ind\{b_{2}-Z_{is}'\b\geq b_{3}-Z_{it}'\b\}(\ind\{Y_{is}=1\}-\ind\{Y_{it}\in\{1,2\}\});\\
 & \ind\{b_{3}-Z_{is}'\b\geq b_{2}-Z_{it}'\b\}(\ind\{Y_{is}\in\{1,2\}\}-\ind\{Y_{it}=1\});\\
 & \ind\{b_{3}-Z_{is}'\b\geq b_{3}-Z_{it}'\b\}(\ind\{Y_{is}\in\{1,2\}\}-\ind\{Y_{it}\in\{1,2\}\}),
\end{aligned}
\right.
\]
upon which our estimation and inference exercise in this subsection will be based.

The first element $\b_{01}$ of the parameter $\b_{0}$ is normalized
to one, and we are interested in conducting inference for the parameter
$\b_{02}$ using the CLR approach. Tables \ref{table:static_sigma}
and \ref{table:static_rho} report the average confidence interval
(CI) for $\b_{02}$, the coverage probability (CP), the average length
of the CI (length), the power of the test at zero (power), and the
mean absolute deviation of the lower bound ($l_{MAD}$) and upper
bound ($u_{MAD}$) of the CI.

\begin{table}[!htbp]
\centering %\setlength{\tabcolsep}{2mm}{
\caption{Performance of $\protect\b_{02}$ under different values of $\protect\s_{z}$
($\rho=0.25$)}
\label{table:static_sigma} %
\begin{tabular}{c|cccccc}
\hline 
$\s_{z}$  & CI  & CP  & length  & power  & $l_{MAD}$  & $u_{MAD}$ \tabularnewline
\hline 
 & \multicolumn{6}{c}{ $N=2000$}\tabularnewline
\hline 
$\s_{z}=1$  & {[}0.537, 1.760{]}  & 0.876  & 1.222  & 1.000 & 0.476  & 0.784 \tabularnewline
$\s_{z}=1.5$  & {[}0.556, 1.768{]}  & 0.934 & 1.212 & 1.000  & 0.454  & 0.773 \tabularnewline
$\s_{z}=2$  & {[}0.567, 1.791{]}  & 0.950  & 1.224  & 1.000  & 0.440  & 0.796 \tabularnewline
\hline 
 & \multicolumn{6}{c}{$N=8000$}\tabularnewline
\hline 
$\s_{z}=1$  & {[}0.570, 1.532{]}  & 0.939  & 0.962  & 1.000  & 0.439  & 0.548 \tabularnewline
$\s_{z}=1.5$  & {[}0.607, 1.561{]}  & 0.975  & 0.954  & 1.000  & 0.398  & 0.563 \tabularnewline
$\s_{z}=2$  & {[}0.618, 1.571{]}  & 0.985  & 0.953  & 1.000  & 0.383  & 0.573 \tabularnewline
\hline 
\end{tabular}
\end{table}

\begin{table}[!htbp]
\centering %\setlength{\tabcolsep}{2mm}{
\caption{Performance of $\protect\b_{02}$ under different values of $\rho$
($\protect\s_{z}=1$)}
\label{table:static_rho} %
\begin{tabular}{c|cccccc}
\hline 
$\rho$  & CI  & CP  & length  & power  & $l_{MAD}$  & $u_{MAD}$ \tabularnewline
\hline 
 & \multicolumn{6}{c}{ $N=2000$}\tabularnewline
\hline 
$\rho=0$  & {[}0.537, 1.755{]}  & 0.895  & 1.218  & 1.000 & 0.476  & 0.773 \tabularnewline
$\rho=0.25$  & {[}0.537, 1.760{]}  & 0.876  & 1.222  & 1.000 & 0.476  & 0.784 \tabularnewline
$\rho=0.5$  & {[}0.511, 1.765{]}  & 0.909  & 1.254  & 1.000  & 0.497 & 0.785 \tabularnewline
\hline 
 & \multicolumn{6}{c}{ $N=8000$}\tabularnewline
\hline 
$\rho=0$  & {[}0.584, 1.553{]}  & 0.933  & 0.969  & 1.000  & 0.436  & 0.568 \tabularnewline
$\rho=0.25$  & {[}0.570, 1.532{]}  & 0.939  & 0.962  & 1.000  & 0.439  & 0.548 \tabularnewline
$\rho=0.5$  & {[}0.573, 1.526{]}  & 0.934  & 0.954  & 1.000  & 0.442  & 0.541 \tabularnewline
\hline 
\end{tabular}
\end{table}

As shown in Tables \ref{table:static_sigma} and \ref{table:static_rho},
our approach exhibits robust performance across various specifications
of standard deviation $\s$ and correlation coefficients $\rho$.
The coverage probabilities of the 95\% confidence interval (CI) for
$\b_{02}$ are close to the nominal level, the length of the CI is
reasonably small, and the CI consistently excludes zero. When the
sample size increases, there is a significant decrease in CI length,
an improvement in coverage probability, and a reduction of the mean
absolute deviation (MAD) for the lower and upper bounds of the CI.
Overall, these results demonstrate the good performance of our approach
in different DGP designs.

\subsection{Dynamic Ordered Response Model}

In this section, we investigate a dynamic ordered choice model with
one lagged dependent variable $Y_{i,t-1}$. The latent dependent variable
$Y_{it}^{*}$ is generated as follows: 
\[
Y_{it}^{*}=Z_{it}\b_{0}+Y_{i,t-1}\g_{0}+\a_{i}+\e_{it}.
\]
where the endogenous variable is the lagged dependent variable $Y_{i,t-1}$.
We study three periods $T=3$ to illustrate our approach with multiple
periods. The DGP is similar: the exogenous covariate $Z_{it}$ satisfies
$Z_{it}\sim\mathcal{N}(0,\s_{z})$; the fixed effects $\a_{i}$ are
given as $\a_{i}=\sum_{t=1}^{T}Z_{it}/(4*\s_{z}*T)$; the error term
$(\e_{i1},\e_{i2},\e_{i3})$ follows the normal distribution $\mathcal{N}(\mu,\Sigma)$
with $\mu=(0,0,0)$ and $\Sigma=(0.5\ c\ c;c\ 0.5\ c;c\ c\ 0.5)$,
where $c=0.5*\rho$. The true parameter is $\t_{0}:=(\b_{0},\g_{0})'=(1,1)'$,
the repetition number is $B=200$, and the sample size is $n\in\{2000,8000\}$.
We consider three specifications for $\s_{z}\in\{1,1.5,2\}$ and $\rho\in\{0,0.25,0.5\}$.

The observed dependent variable $Y_{it}$ is given as 
\[
Y_{it}=1*(Y_{it}^{*}\leq b_{2})+2*(b_{2}<Y_{it}^{*}\leq b_{3})+3*(Y_{it}^{*}>b_{3}),
\]
for $1\leq t\leq T$. The initial value $Y_{i0}\in\{1,2,3\}$ is generated
independently of all variables and follows the distribution $\P(Y_{i0}=1)=0.6,\P(Y_{i0}=2)=\P(Y_{i0}=3)=0.2$.

In this dynamic model, the covariates $Z_{i}:=(Z_{it})_{t=1}^{T}$
and the initial value $Y_{i0}$ are exogenous, while the lagged variable
$Y_{i,t-1}$ is endogenous. Proposition \ref{prop:order} characterizes
the identified set for $\t_{0}$ with the following conditional moment
inequalities:
\begin{itemize}
\item[(1)] When $s\in\{2,3\}$.
\[
\begin{aligned} & \sum_{j=1}^{2}\P\left(Y_{is}=y_{j},b_{j+1}-z_{s}'\b-Y_{is-1}\g\leq c\mid z,y_{0}\right),\\
\leq\  & 1-\sum_{j=2}^{3}\P\left(Y_{i1}=y_{j}\mid z,y_{0}\right)*\ind\left\{ b_{j}-z_{1}'\beta-y_{0}\g\geq c\right\} 
\end{aligned}
\]
\[
\begin{aligned} & \sum_{j=1}^{2}\P\left(Y_{i1}=y_{j}\mid z,y_{0}\right)*\ind\{b_{j+1}-z_{1}'\b-y_{0}\g\leq c\},\\
\leq\  & 1-\sum_{j=2}^{3}\P\left(Y_{is}=y_{j},b_{j}-z_{s}'\beta-Y_{is-1}\g\geq c\mid z,y_{0}\right)
\end{aligned}
\]
for any $c\in\{b_{j}-z_{1}'\b-y_{0}\g,b_{j}-z_{s}'\b-\g,b_{j}-z_{s}'\b-2\g,b_{j}-z_{s}'\b-3\g\}_{j=2}^{T}$;
\item[(2)] When $s,t\in\{2,3\}$, 
\[
\begin{aligned} & \sum_{j=1}^{2}\P\left(Y_{it}=y_{j},b_{j+1}-z_{t}'\b-Y_{it-1}\g\leq c\mid z,y_{0}\right),\\
\leq\  & 1-\sum_{j=2}^{3}\P\left(Y_{is}=y_{j},b_{j}-z_{s}'\beta-Y_{is-1}\g\geq c\mid z,y_{0}\right)
\end{aligned}
\]
for any $c\in\{b_{j}-z_{s}'\b-\g,b_{j}-z_{s}'\b-2\g,b_{j}-z_{s}'\b-3\g,b_{j}-z_{t}'\b-\g,b_{j}-z_{t}'\b-2\g,b_{j}-z_{t}'\b-3\g\}_{j=2}^{3}$.
\end{itemize}
We normalize the first parameter $\b_{0}$ to one, and report the
performance of the coefficient $\g_{0}$ for the lagged dependent
variable. Tables \ref{table:dyn_sigma} and \ref{table:dyn_rho} illustrate
that our approach yields robust and informative results for the dynamic
ordered choice model across various DGP specifications. The coverage
probability of the CI nearly reaches 95\%, and the CI consistently
excludes zero, producing significant coefficients. These results remain
similar across different values of correlation coefficients. When
the standard deviation $\s_{z}$ increases, the length of the CI also
experiences a slight increase. This phenomenon occurs because, in
the dynamic model, only partial identification is achieved, and the
bound for $\g_{0}$ depends on the variation in $\D z'\b_{0}$. A
larger variation in $\D z'\b_{0}$ may result in a wider identified
set in this specification, but it still provides informative results.
As the sample size increases, the confidence interval shrinks, and
concurrently, the coverage probability improves in all specifications.

%\begin{table}[!htbp] 
%\centering
%%\setlength{\tabcolsep}{2mm}{
%\caption{Performance of $\g_0$ under different specifications}
% \label{table:dyn_sigma}
%\begin{tabular}{c|cccccc}
%\hline
%\hline
% Specifications & CI & CP & length & Power & $l_{MAD}$ & $u_{MAD}$   \\
%\hline 
%&\multicolumn{6}{c}{ $\rho=0.25$} \\
%\hline
%$\s_z =1$ &[0.446, 1.606] & 0.935 & 1.160&  1.000&  0.565 & 0.625
%  \\[1.5ex]
% $\s_z =1.5$ & [0.375,  1.673]   &0.959  & 1.298 &  1.000  & 0.629 &  0.693
%   \\[1.5ex]
% $\s_z =2 $ & [0.311, 1.730] & 0.960 & 1.418 & 1.000&  0.700&  0.739
%     \\[1.5ex]
%  \hline
%  &\multicolumn{6}{c}{ $\s=1$}  \\
%  \hline
%$\rho=0$ & [0.472,  1.593] &  0.932 &  1.121 &  1.000  & 0.550 & 0.607
%  \\[1.5ex]
% $\rho=0.25$ &[0.446, 1.606] & 0.935 & 1.160&  1.000&  0.565 & 0.625
%   \\[1.5ex]
% $\rho=0.5 $ &[0.457,  1.631] &  0.943 &  1.173  & 1.000&   0.548 &  0.648
%     \\[1.5ex]
%  \hline
% \hline  
%\end{tabular}
%\end{table}
\begin{table}[!htbp]
\centering %\setlength{\tabcolsep}{2mm}{
 \caption{Performance of $\protect\g_{0}$ under different values of $\protect\s_{z}$
($\rho=0.25$)}
\label{table:dyn_sigma} %
\begin{tabular}{c|cccccc}
\hline 
$\s_{z}$  & CI  & CP  & length  & power  & $l_{MAD}$  & $u_{MAD}$ \tabularnewline
\hline 
 & \multicolumn{6}{c}{ $N=2000$}\tabularnewline
\hline 
$\s_{z}=1$  & {[}0.446, 1.606{]}  & 0.935  & 1.160 & 1.000 & 0.565  & 0.625 \tabularnewline
$\s_{z}=1.5$  & {[}0.375, 1.673{]}  & 0.959  & 1.298  & 1.000  & 0.629  & 0.693 \tabularnewline
$\s_{z}=2$  & {[}0.311, 1.730{]}  & 0.960  & 1.418  & 1.000 & 0.700 & 0.739 \tabularnewline
\hline 
 & \multicolumn{6}{c}{ $N=8000$}\tabularnewline
\hline 
$\s_{z}=1$  & {[}0.529, 1.495{]}  & 0.969  & 0.966  & 1.000  & 0.473  & 0.504 \tabularnewline
$\s_{z}=1.5$  & {[}0.460, 1.559{]}  & 0.965  & 1.100  & 1.000  & 0.548  & 0.564 \tabularnewline
$\s_{z}=2$  & {[}0.427, 1.585{]}  & 0.985  & 1.158  & 1.000 & 0.573 & 0.589 \tabularnewline
\hline 
\end{tabular}
\end{table}

\begin{table}[!htbp]
\centering %\setlength{\tabcolsep}{2mm}{
 \caption{Performance of $\protect\g_{0}$ under different values of $\rho$
($\sigma_{z}=1$)}
\label{table:dyn_rho} %
\begin{tabular}{c|cccccc}
\hline 
$\rho$  & CI  & CP  & length  & power  & $l_{MAD}$  & $u_{MAD}$ \tabularnewline
\hline 
 & \multicolumn{6}{c}{ $N=2000$}\tabularnewline
\hline 
$\rho=0$  & {[}0.472, 1.593{]}  & 0.932  & 1.121  & 1.000  & 0.550  & 0.607 \tabularnewline
$\rho=0.25$  & {[}0.446, 1.606{]}  & 0.935  & 1.160 & 1.000 & 0.565  & 0.625 \tabularnewline
$\rho=0.5$  & {[}0.457, 1.631{]}  & 0.943  & 1.173  & 1.000 & 0.548  & 0.648 \tabularnewline
\hline 
 & \multicolumn{6}{c}{ $N=8000$}\tabularnewline
\hline 
$\rho=0$  & {[}0.528, 1.472{]}  & 0.958  & 0.945 & 1.000 & 0.475  & 0.487 \tabularnewline
$\rho=0.25$  & {[}0.529, 1.495{]}  & 0.969  & 0.966  & 1.000  & 0.473  & 0.504 \tabularnewline
$\rho=0.5$  & {[}0.535, 1.515{]}  & 0.975 & 0.980  & 1.000 & 0.467 & 0.519 \tabularnewline
\hline 
\end{tabular}
\end{table}

\section{Empirical Application}

\label{sec:appl}

In this section, we apply our proposed approach to explore the empirical
analysis of income categories using the NLSY79 dataset. The dependent
variable is three categories of (log) income, denoted by the three
values $\{1,2,3\}$, indicating whether an individual falls within
the top 33.3\% highest income bracket, the 33.3\%-66.6\% highest income
range, and the lowest 33.3\% income tier, respectively. We include
two covariates in this analysis: one is tenure, defined as the total
duration (in weeks) with the current employer, and the other is a
residence indicator for whether one lives in an urban or rural area.\footnote{This dataset also contains other crucial factors for income such as
gender and race. However, these variables are time-invariant and cannot
be included for panel models with fixed effects.} We use two periods of panel data from the years 1982 and 1983 as
well as the income data from 1981 as initial values, and there are
$n=5259$ individuals in each period. The following table presents
the summary statistics of these variables.

\begin{table}[!htbp]
\centering %\setlength{\tabcolsep}{2mm}{
 \caption{Application: Summary Statistics}
\label{table: sum} %
\begin{tabular}{c|ccc}
\hline 
 & income category  & residence  & tenure /100 \tabularnewline
mean  & 1.990  & 0.799  & 0.825 \tabularnewline
s.d. & 0.810  & 0.401  & 0.738 \tabularnewline
25\% quantile  & 1.000  & 1.000  & 0.220 \tabularnewline
median  & 2.000  & 1.000  & 0.605 \tabularnewline
75\% quantile  & 3.000  & 1.000  & 1.280 \tabularnewline
minimum  & 1.000  & 0.000  & 0.010 \tabularnewline
maximum  & 3.000  & 1.000  & 4.850 \tabularnewline
\hline 
\end{tabular}
\end{table}

We adopt various ordered response models introduced in Section \ref{subsec:order}
to analyze the income category. The first model is the standard static
model without any endogeneity. The second is the static model, while
treating residence as an endogenous covariate. Residence is potentially
endogenous since the choice of living area is typically endogenously
determined and may be correlated with individuals' unobserved ability
or preference. The last model considers the dynamic model with one
lagged dependent variable, allowing people's income in current periods
to depend on their income in the last period. All three models allow
for individual fixed effects and do not impose any parametric distributions
on time-changing shocks. Proposition \ref{prop:order} characterizes
the identified set of the model coefficients for these three models
using conditional moment inequalities. Similar to Section \ref{sec:simu},
we exploit the kernel-based CLR inference method to construct confidence
intervals. The coefficient of the variable ``residence'' is normalized
to one. Table \ref{table: app} reports the confidence intervals for
the coefficients of the covariate ``tenure'' and the lagged dependent
variable (when applicable).

\begin{table}[!htbp]
\centering %\setlength{\tabcolsep}{2mm}{
 \caption{Application: Income Categories}
\label{table: app} %
\begin{tabular}{c|ccc}
\hline 
 & $\b_{0,1}$ (residence)  & $\b_{0,2}$ (tenure)  & $\g_{0}$ (lag) \tabularnewline
exogenous static model  & 1  & {[}0.612, 0.939{]}  & - \tabularnewline
endogenous static model  & 1  & {[}0.041, 0.939{]}  & - \tabularnewline
dynamic model  & 1  & {[}0.531, 0.694{]}  & {[}0.286, 0.612{]} \tabularnewline
\hline 
\end{tabular}
\end{table}

As shown in Table \ref{table: app}, tenure exhibits a significantly
positive effect on the income category across all specifications.
When allowing for the endogeneity of residence, the confidence interval
for tenure becomes wider, as we need to account for all possible correlations
between residence and unobserved heterogeneity. The results from the
dynamic model show that the income category in the current period
is also positively affected by last period's income, and this effect
is significant. Furthermore, this analysis demonstrates the flexibility
of our approach, which can not only allow for endogeneity introduced
by dynamics but also account for contemporaneous endogeneity.

\section{Conclusion}

\label{sec:conc}

We introduce a general method to (partially) identify index parameters in nonlinear panel data models
based on a partial stationarity condition. This approach accommodates
dynamic models with an arbitrary finite number of lagged outcome variables
and other types of endogenous covariates. We demonstrate how our key
identification strategy can be applied to obtain informative identifying
restrictions in various limited dependent variable models, including
binary choice, ordered response, multinomial choice, as well as censored
outcome models. Finally, we further extend this approach to study
general nonseparable models.

There are some natural directions for follow-up research. In this
paper we focus on the identification of model parameters, but it would
also be interesting to investigate how our identification strategy
can be exploited to obtain informative bounds on average marginal
effects and other counterfactual parameters, say, following the approach
proposed in \citet{botosaru2022identification}.\footnote{\citet{botosaru2022identification} proposes an approach to obtain
bounds on counterfactual CCPs in semiparametric dynamic panel data
models, assuming that the index parameters are (partially) identified.} Also, our identification strategy should be adaptable to exploit
additional restrictions imposed by time-exchangeability assumptions
such as in \citet*{mbakop2023identification}, which not only impose
homogeneity on per-period marginals of errors but also on their intertemporal
dependence structures. Additionally, the idea of bounding an endogenous
object (parametric index in our case) by an arbitrary constant so
as to obtain an object free of endogeneity issues may have broader
applicability beyond the models studied in this work, and it remains
to see whether our key identification strategy can be further adapted
to other structures.

\bibliographystyle{ecta}
\bibliography{stationarity}

\newpage{}

\appendix
\noindent \textbf{\LARGE{}Appendix}{\LARGE\par}

\section{Main Proofs}

\label{sec:appen}

\subsection{Proof of Proposition \ref{prop:disc}}
\begin{proof}
Clearly, $\T_{I}\subseteq\T^{disc}$. Below we show $\T^{disc}\subseteq\T_{I}$
when $X_{it}$ is discrete. Suppose that $\t$ satisfies condition
\eqref{eq:ID_Set} at all 
\[
c\in{\cal C}\left(\t\right):=\left\{ z_{t}^{'}\b+\ol x_{k}^{'}\g:k=1,...,K,t=1,...,T\right\} 
\]
for any realization $z=\left(z_{1},...,z_{T}\right).$ We seek to
show that $\t$ must also satisfy condition \eqref{eq:ID_Set} for
any $c\in\R\backslash{\cal C}\left(\t\right)$. Without loss of generality,
we order elements in ${\cal C}\left(\t\right)$ from the smallest
to the largest as 
\[
\ol c_{1}\leq\ol c_{2}\leq...\leq\ol c_{KT}.
\]
For $c<\ol c_{1}$, we must have 
\[
\P\left(\rest{Y_{it}=1,\ z_{t}^{'}\b+X_{it}^{'}\g\leq c}z\right)\equiv0,
\]
so \eqref{eq:ID_Set} holds trivially. Similarly, for $c>\ol c_{KT}$,
we must have 
\[
\P\left(\rest{Y_{is}=0,\ z_{s}^{'}\b+X_{is}^{'}\g\geq c}z\right)\equiv0,
\]
so \eqref{eq:ID_Set} again holds trivially. For any $c$ s.t. $\ol c_{j}<c<\ol c_{j+1}$
for some $j$, we have 
\[
z_{t}^{'}\b+X_{it}^{'}\g\leq c\quad\iff\quad z_{t}^{'}\b+X_{it}^{'}\g\leq\ol c_{j}
\]
and 
\[
z_{s}^{'}\b+X_{is}^{'}\g\geq c\quad\iff\quad z_{s}^{'}\b+X_{is}^{'}\g\geq\ol c_{j+1}.
\]
which implies 
\begin{align}
\P\left(\rest{Y_{it}=1,\ z_{t}^{'}\b+X_{it}^{'}\g\leq c}Z_{i}=z\right)= & \ \P\left(\rest{Y_{it}=1,\ z_{t}^{'}\b+X_{it}^{'}\g\leq\ol c_{j}}z\right)\label{eq:c_cj_LB}
\end{align}
and 
\begin{align*}
\P\left(\rest{Y_{is}=0,\ z_{s}^{'}\b+X_{is}^{'}\g\geq c}z\right) & =\P\left(\rest{Y_{is}=0,\ z_{s}^{'}\b+X_{is}^{'}\g\geq\ol c_{j+1}}z\right)\\
 & \leq\P\left(\rest{Y_{is}=0,\ z_{s}^{'}\b+X_{is}^{'}\g\geq\ol c_{j}}z\right),
\end{align*}
or equivalently, 
\begin{equation}
1-\P\left(\rest{Y_{is}=0,\ z_{s}^{'}\b+X_{is}^{'}\g\geq\ol c_{j}}z\right)\leq1-\P\left(\rest{Y_{is}=0,\ z_{s}^{'}\b+X_{is}^{'}\g\geq\ol c_{j+1}}z\right).\label{eq:c_cj_UB}
\end{equation}
Since \eqref{eq:ID_Set} holds at $\ol c_{j}$, we have 
\[
\max_{t}\P\left(\rest{Y_{it}=1,\ z_{t}^{'}\b+X_{it}^{'}\g\leq\ol c_{j}}z\right)\leq1-\max_{s}\P\left(\rest{Y_{is}=0,\ z_{s}^{'}\b+X_{is}^{'}\g\geq\ol c_{j}}z\right).
\]
Combining the above with \eqref{eq:c_cj_LB} and \eqref{eq:c_cj_UB},
we have 
\[
\max_{t}\P\left(\rest{Y_{it}=1,\ z_{t}^{'}\b+X_{it}^{'}\g\leq c}z\right)\leq1-\max_{s}\P\left(\rest{Y_{is}=0,\ z_{s}^{'}\b+X_{is}^{'}\g\geq c}z\right).
\]
\end{proof}

\subsection{\label{subsec:pf_sharp_disc}Proof of Theorem \ref{thm:sharp_disc}}

We first clarify the rigorous meaning of ``sharpness'' in Theorem
\ref{thm:sharp_disc} through the following definition.
\begin{defn}
\label{def:sharp}We say that $\T_{I}$ is \emph{sharp} under model
\eqref{eq:bin} and Assumption \ref{assu:PartStat} if, for any
$\t\equiv\left(\b^{'},\g^{'}\right)^{'}\in\T_{I}^{disc}\backslash\left\{ \t_{0}\right\} $,
there exist well-defined latent random variables $\left(\e_{i}^{*},\a_{i}^{*}\right)$
such that:
\begin{itemize}
\item Assumption \ref{assu:PartStat} (partial stationarity) is satisfied,
i.e.,
\[
\rest{\e_{it}^{*}\sim\e_{is}^{*}}Z_{i},\a_{i}^{*},\ \forall t,s=1,...,T.
\]
\item (CCP-J) $\left(\t,\e_{i}^{*},\a_{i}^{*}\right)$ are observationally
equivalent to $\left(\t_{0},\e_{i},\a_{i}\right)$, i.e., formally,
$\left(\t,\e_{i}^{*},\a_{i}^{*}\right)$ produces the following conditional
choice probabilities under model \eqref{eq:bin}:
\begin{equation}
\begin{aligned} & \P\left(v_{it}^{*}\leq w_{t}^{'}\t \ \forall t\text{ s.t. }y_{t}=1,\ v_{is}^{*}>w_{s}^{'}\t \ \forall s\text{ s.t. }y_{s}=0\mid w\right)=p\left(y\mid w\right),\end{aligned}
\label{eq:CCP_Joint}
\end{equation}
where $v_{it}^{*}:=-\left(\e_{it}^{*}+\a_{i}^{*}\right)$ and $p\left(\rest{\cd}w\right)$
denotes the \emph{true} conditional probability
\begin{align*}
p\left(\rest yw\right) & :=\P\left(Y_{it}=y_{t}\,\forall t=1,...,T\mid W_{i}=w\right)\\
 & \equiv\P\left(v_{it}\leq w_{t}^{'}\t_{0} \ \forall t\text{ s.t. }y_{t}=1,\ v_{is}>w_{s}^{'}\t \ \forall s\text{ s.t. }y_{s}=0\mid W_{i}=w\right),
\end{align*}
for any outcome realization $y\equiv\left(y_{1},...,y_{T}\right)\in\left\{ 0,1\right\} ^{T}$,
for almost every realization $w$ of $W_{i}$ (except in a set of
probability measure zero). 
\end{itemize}
\end{defn}
\begin{proof}
We prove Theorem \ref{thm:sharp_disc} by providing a construction
of $\left(\e_{i}^{*},\a_{i}^{*}\right)$ in Definition \ref{def:sharp}
for any candidate parameter $\t\in\T_{I}^{disc}\backslash\left\{ \t_{0}\right\} $.
Under discreteness of $X_{i}$, note that the CCP matching condition
(CCP-J) needs to be satisfied for each realization $x$ of $X_{i}$
and a.s.-$Z_{i}$. Set $\a_{i}^{*}\equiv0$ and $\e_{i}^{*}:=-v_{i}^{*}$.
Then the conclusion follows from Lemma \ref{lem:SharpMarginal} and
\ref{lem:SharpJoint} below.
\end{proof}
\begin{lem}[Relaxed Discrete Problem]
\label{lem:SharpMarginal} Suppose that $\bigcup_{t=1}^{T}\text{Supp}\left(X_{it}\right)$
is finite. For any $\t\equiv\left(\b^{'},\g^{'}\right)^{'}\in\T_{I}^{disc}\backslash\left\{ \t_{0}\right\} $,
there exist well-defined latent random variables $v_{i1}^{*},...,v_{iT}^{*}$
with marginal CDFs $F_{1}^{*},...,F_{T}^{*}$ such that
\begin{equation}
\rest{F_{t}^{*}\left(\rest{\cd}Z_{i}=z\right)=F_{s}^{*}\left(\rest{\cd}Z_{i}=z\right)}\label{eq:PS_v}
\end{equation}
and
\begin{equation}
F_{t}^{*}\left(\rest{w_{t}^{'}\t}W_{i}=w\right)=p_{t}\left(w\right),\quad\forall t,\forall w,\label{eq:CCP_M}
\end{equation}
where
\[
p_{t}\left(w\right):=\P\left(\rest{Y_{it}=1}W_{i}=w\right).
\]
\end{lem}
\begin{proof}
For any $\t\equiv\left(\b^{'},\g^{'}\right)^{'}\in\T_{I}^{disc}\backslash\left\{ \t_{0}\right\} $,
below we show how to construct $v_{i1}^{*},...,v_{iT}^{*}$, or equivalently,
the conditional CDFs $F_{1}^{*}\left(\rest cW_{i}=w\right),...,F_{T}^{*}\left(\rest cW_{i}=w\right)$
for each realization $w$ and each $c\in\R$ so that (i) condition
\eqref{eq:PS_v} is satisfied so that partial stationarity holds;
and (ii) condition \eqref{eq:CCP_M} is satisfied so that per-period
marginal CCPs are matched.

Fix a specific realization of the exogenous covariates at $z\equiv\left(z_{1},...,z_{T}\right)$.
We construct the (conditional) CDF $F_{t}^{*}$ of $v_{it}^{*}$ for
each $t=1,...,T$ and each given $z$ in the following manner. 

From now on, we suppress ``$|Z_{i}=z$'' from all functions that
are defined conditional on $z$. However, we will write out $F_{t}^{*}\left(\rest{\cd}z\right)$
and $F_{t}^{*}\left(\rest{\cd}w\right)$ explicitly to emphasize the
difference in the conditioning variables.

Define
\begin{align*}
L_{t}\left(c\right) & :=\P\left(\rest{Y_{it}=1,\ z_{t}^{'}\b+X_{it}^{'}\g\leq c}Z_{i}=z\right),\\
U_{t}\left(c\right) & :=1-\P\left(\rest{Y_{it}=0,\ z_{t}^{'}\b+X_{it}^{'}\g\geq c}Z_{i}=z\right),
\end{align*}
and 
\begin{align*}
\ol L\left(c\right):=\max_{s}L_{s}\left(c\right),\quad & \ul U\left(c\right):=\min_{s}U_{s}\left(c\right).
\end{align*}
Since $\t\equiv\left(\b^{'},\g^{'}\right)^{'}\in\T_{I}^{disc}\backslash\left\{ \t_{0}\right\} $,
by \ref{eq:ID_Set} we have, 
\[
\ol L\left(c\right)\leq\ul U\left(c\right),\quad\forall c\in\R.
\]
Observe that both $\ol L\left(c\right)$ and $\ul U\left(c\right)$
are weakly increasing in $c$.

Since$X_{it}$ can only take $K$ values $\ol x_{1},...,\ol x_{K}$,
the parametric index $w_{t}^{'}\t\equiv z_{t}^{'}\b+x_{t}^{'}\g$
can only take values in the set
\begin{align*}
{\cal C} & :=\left\{ z_{t}^{'}\b+\ol x_{k}^{'}\g:t=1,...,T,k=1,...,K\right\} .\\
 & =\left\{ c_{\left(1\right)},...,c_{\left(\kappa\right)}:c_{\left(1\right)}<...<c_{\left(\kappa\right)}\right\} 
\end{align*}
and write
\[
\ul c:=\min{\cal C},\quad\ol c:=\max{\cal C}
\]
so that $\ul c\leq W_{it}^{'}\t\leq\ol c$ for all $t$.

Let $\d>0$ be a sufficiently small positive constant.\footnote{The small positive constant $\d>0$ is used to ensure the right continuity
of CDFs defined afterwards. Let $\ul{\d}$ to be smallest distance
between any two \emph{distinct} points in ${\cal C}$. If $\ul{\d}>0$,
then we may set $\d:=\ul{\d}/2$. If $\ul{\d}=0$, $\d$ can be set
as any positive number, say, $\d:=1$. However, it is worth pointing
out that, if $\ul{\d}=0$, then $W_{it}^{'}\t$ is degenerate once
conditional given $z$, and is thus a deterministic function of $z$,
which would correspond to a degenerate case where there is effectively
no endogenous covariate $X_{it}$. Sharpness in such fully exogenous
case is easier to establish and does not require our new proof. That
said, for technical comprehensiveness, in the case of $\ul{\d}=0$, } For each $t=1,...,T,$ we show how to construct $v_{t}^{*}$ with
CDF $F_{t}^{*}$ 
\begin{equation}
F_{t}^{*}\left(\rest cz\right)\equiv\begin{cases}
0, & \text{if }c<\ul c,\\
\ol L\left(c\right), & \text{if }\ul c\leq c<\ol c+\d,\\
1, & \text{if }c\geq\ol c+\d,
\end{cases}\label{eq:PS_vt}
\end{equation}
and
\begin{equation}
F_{t}^{*}\left(\rest cw\right)=p_{t}\left(w\right)\quad\forall c\in{\cal C}.\label{eq:CCP_Mt}
\end{equation}

Clearly, partial stationarity \eqref{eq:PS_v} will be satisfied under
\eqref{eq:PS_vt}, the right-hand side of which does not depend on
the time index $t$. Furthermore, since $w_{t}^{'}\t\in{\cal C}$
by the definition of ${\cal C},$ \eqref{eq:CCP_Mt} would imply \eqref{eq:CCP_M},
i.e., the marginal CCPs will be matched for each $t$.

\textbf{Step 1:}

We construct the conditional CDF of $v^{*}|W_{i}=w$ using two auxiliary
CDFs $F_{t}^{L}$ and $F_{t}^{U}$, defined by
\begin{align*}
F_{t}^{L}\left(\rest cw\right) & =\begin{cases}
0, & c<w_{t}^{'}\t,\\
p_{t}\left(w\right), & w_{t}^{'}\t\leq c<\ol c_{t}+\d,\\
1, & c\geq\ol c_{t}+\d,
\end{cases}
\end{align*}
and
\[
F_{t}^{U}\left(\rest cw\right)=\begin{cases}
0, & c<\ul c_{t},\\
p_{t}\left(w\right), & \ul c_{t}\leq c<w_{t}^{'}\t+\d,\\
1, & c\geq w_{t}^{'}\t+\d.
\end{cases}
\]
where
\[
\ol c_{t}:=\max{\cal C}_{t},\quad\ul c_{t}:=\min{\cal C}_{t},\quad{\cal C}_{t}:=\left\{ z_{t}^{'}\b+\ol x_{k}^{'}\g:k=1,...,K\right\} .
\]
Clearly, by construction we have
\begin{equation}
F_{t}^{L}\left(\rest{w_{t}^{'}\t}w\right)=F_{t}^{U}\left(\rest{w_{t}^{'}\t}w\right)=p_{t}\left(w\right).\label{eq:CCP_Mt_UL}
\end{equation}
Furthermore, for any $c\in\left[\ul c_{t},\ol c_{t}\right]$, we have
\begin{align*}
F_{t}^{L}\left(\rest cz\right) & =\E\left[\rest{F_{t}^{L}\left(\rest cW_{i}\right)}Z_{i}=z\right]\\
 & =\E\left[\ind\left\{ W_{it}^{'}\t\leq c\right\} p_{t}\left(W_{i}\right)\right]\\
 & =\E\left[\P\left(\rest{Y_{i}=1\text{ and }W_{it}^{'}\t\leq c}W_{i}=w\right)\right]\\
 & =\P\left(\rest{Y_{i}=1\text{ and }W_{it}^{'}\t\leq c}Z_{i}=z\right)\\
 & =L_{t}\left(\rest cz\right),
\end{align*}
and similarly
\begin{align*}
F_{t}^{U}\left(\rest cz\right) & =\E\left[F_{t}^{U}\left(\rest cW_{i}\right)\right]\\
 & =\E\left[\rest{1-\left(1-p_{t}\left(W_{i}\right)\right)\ind\left\{ W_{it}^{'}\t\geq c-\d\right\} }Z_{i}=z\right]\\
 & =\E\left[\rest{1-\P\left(\rest{Y_{it}=0\text{ and }W_{it}^{'}\t\geq c-\d}W_{i}\right)}Z_{i}=z\right]\\
 & =1-\P\left(\rest{Y_{i}=0\text{ and }W_{it}^{'}\t\geq c-\d}Z_{i}=z\right)\\
 & =1-\P\left(\rest{Y_{i}=0\text{ and }W_{it}^{'}\t\geq c}Z_{i}=z\right)\\
 & =U_{t}\left(\rest cz\right),
\end{align*}
where the second last equality holds for sufficiently small $\d>0$
due to the discreteness of ${\cal C}$. 

In summary, we have
\begin{align}
F_{t}^{L}\left(\rest cz\right) & =\begin{cases}
L_{t}\left(c\right), & \forall c<\ol c_{t}+\d,\\
1, & \forall c\geq\ol c_{t}+\d,
\end{cases}\nonumber \\
F_{t}^{U}\left(\rest cz\right) & =\begin{cases}
0, & \forall c<\ul c_{t},\\
U_{t}\left(c\right), & \forall c\geq\ul c_{t},
\end{cases}\label{eq:Ft_z_UL}
\end{align}
Furthermore, observe that 
\[
L_{t}\left(\cd\right)\leq F_{t}^{L}\left(\rest{\cd}z\right)\leq F_{t}^{U}\left(\rest{\cd}z\right)\leq U_{t}\left(\cd\right).
\]

\textbf{Step 2:}

We now construct $F_{t}^{*}\left(\rest cw\right)$ for $c\in{\cal C}$.
Define
\begin{align*}
{\cal U}_{t} & :=\left\{ U_{t}\left(c\right):c\in{\cal C}\right\} \equiv\left\{ U_{t}\left(c\right):c\in{\cal C}_{t}\right\} \\
{\cal L}_{t} & :=\left\{ L_{t}\left(c\right):c\in{\cal C}\right\} \equiv\left\{ L_{t}\left(c\right):c\in{\cal C}_{t}\right\} 
\end{align*}
Notice that ${\cal L}_{t}\leq{\cal U}_{t}$ and 
\begin{align*}
{\cal L}_{t}\cap{\cal U}_{t} & =\left\{ q^{*}:=U_{t}\left(\ul c_{t}\right)=L_{t}\left(\ol c_{t}\right)\right\} .
\end{align*}
In addition, since
\[
U_{t}\left(\ul c_{t}\right)\leq L_{t}\left(\ul c_{t}\right)\leq L_{t}\left(c\right)\leq\ol L\left(c\right)\leq\ul U\left(c\right)\leq U_{t}\left(c\right)\leq U_{t}\left(\ol c_{t}\right)
\]
we have
\[
\ol L\left({\cal C}\right):=\left\{ \ol L\left(c\right):c\in{\cal C}\right\} \subseteq{\cal L}_{t}\cup{\cal U}_{t}.
\]

Hence, for each $c\in{\cal C}$, there are two exhaustive cases:
\begin{itemize}
\item (i) $c$ is such that $\ol L\left(c\right)>q^{*}$. 

For such $c$, there exists some $1\leq j\leq\kappa$ such that
\[
\ul c_{t}\leq c_{\left(j-1\right)}<c_{\left(j\right)}\leq\ol c_{t}
\]
 and
\[
U_{t}\left(c_{\left(j-1\right)}\right)\leq\ol L\left(c\right)\leq U_{t}\left(c_{\left(j\right)}\right).
\]
Since the inequalities above are weak, in principle there could be
multiple such $j$'s, in which case we take $j$ to be the smallest
one. 

Now, we set 
\[
F_{t}^{*}\left(\rest cw\right)=\a F_{t}^{U}\left(\rest{c_{\left(j-1\right)}}w\right)+\left(1-\a\right)F_{t}^{U}\left(\rest{c_{\left(j\right)}}w\right)
\]
with 
\[
\a:=\begin{cases}
1, & \text{if }U_{t}\left(c_{\left(j-1\right)}\right)=U_{t}\left(c_{\left(j\right)}\right)\\
\frac{U_{t}\left(c_{\left(j\right)}\right)-\ol L\left(c\right)}{U_{t}\left(c_{\left(j\right)}\right)-U_{t}\left(c_{\left(j-1\right)}\right)}, & \text{if }U_{t}\left(c_{\left(j-1\right)}\right)<U_{t}\left(c_{\left(j\right)}\right)
\end{cases}
\]
Then we have the partial stationarity condition satisfied at $c$
\[
F_{t}^{*}\left(\rest cz\right)=\a U_{t}\left(c_{\left(j-1\right)}\right)+\left(1-\a\right)U_{t}\left(c_{\left(j\right)}\right)=\ol L\left(c\right).
\]

Furthermore, since $\ol L\left(c\right)\leq U_{t}\left(c\right)$,
we must have $c_{\left(j\right)}\leq$$c$. Thus if $w$ is such that
$w_{t}^{'}\t=c$, then we must have 
\[
\ul c_{t}\leq c_{\left(j-1\right)}<c_{\left(j\right)}\leq\min\left\{ \ol c_{t},c=w_{t}^{'}\t\right\} <w_{t}^{'}\t+\d
\]
Hence, by the definition of $F_{t}^{U}$, we have
\begin{align*}
F_{t}^{*}\left(\rest{w_{t}^{'}\t}w\right) & =\a F_{t}^{U}\left(\rest{c_{\left(j-1\right)}}w\right)+\left(1-\a\right)F_{t}^{U}\left(\rest{c_{\left(j\right)}}w\right)\\
 & =\a p_{t}\left(w\right)+\left(1-\a\right)p_{t}\left(w\right)\\
 & =p_{t}\left(w\right),
\end{align*}
which satisfies the period-$t$ CCP matching condition at $w$.
\item (ii) $c$ is such that $\ol{L}\left(c\right)\leq q^{*}.$

For such $c$, there exists some $1\leq j\leq\kappa$ such that
\[
\ul c_{t}\leq c_{\left(j-1\right)}<c_{\left(j\right)}\leq\ol c_{t}
\]
 and
\[
L_{t}\left(c_{\left(j-1\right)}\right)\leq\ol L\left(c\right)\leq L_{t}\left(c_{\left(j\right)}\right).
\]
Since the inequalities above are weak, in principle there could multiple
such $j$'s, in which case we take $j$ to be the largest one. 

Now, we set 
\[
F_{t}^{*}\left(\rest cw\right)=\a F_{t}^{L}\left(\rest{c_{\left(j-1\right)}}w\right)+\left(1-\a\right)F_{t}^{L}\left(\rest{c_{\left(j\right)}}w\right)
\]
with 
\[
\a:=\begin{cases}
1, & \text{if }L_{t}\left(c_{\left(j-1\right)}\right)=L_{t}\left(c_{\left(j\right)}\right)\\
\frac{L_{t}\left(c_{\left(j\right)}\right)-\ol L\left(c\right)}{L_{t}\left(c_{\left(j\right)}\right)-L_{t}\left(c_{\left(j-1\right)}\right)}, & \text{if }L_{t}\left(c_{\left(j-1\right)}\right)<L_{t}\left(c_{\left(j\right)}\right)
\end{cases}
\]
Then we have the partial stationarity condition satisfied at $c$
\[
F_{t}^{*}\left(\rest cz\right)=\a L_{t}\left(c_{\left(j-1\right)}\right)+\left(1-\a\right)L_{t}\left(c_{\left(j\right)}\right)=\ol L\left(c\right).
\]

Furthermore, since $\ol{L}\left(c\right)\geq L_{t}\left(c\right)$, we
must have $c_{\left(j-1\right)}\geq$$c$. Thus if $w$ is such that
$w_{t}^{'}\t=c$, then we must have 
\[
\max\left\{ \ul c_{t},c=w_{t}^{'}\t\right\} \leq c_{\left(j-1\right)}<c_{\left(j\right)}\leq\ol c_{t}
\]
Hence, by the definition of $F_{t}^{L}$, we have
\begin{align*}
F_{t}^{*}\left(\rest{w_{t}^{'}\t}w\right) & =\a F_{t}^{L}\left(\rest{c_{\left(j-1\right)}}w\right)+\left(1-\a\right)F_{t}^{L}\left(\rest{c_{\left(j\right)}}w\right)\\
 & =\a p_{t}\left(w\right)+\left(1-\a\right)p_{t}\left(w\right)\\
 & =p_{t}\left(w\right),
\end{align*}
which satisfies the period-$t$ CCP matching condition at $w$
\end{itemize}
~

\textbf{Step 3:}

We now construct $F_{t}^{*}\left(\rest cw\right)$ for $c\in\R\backslash{\cal C}$.
We set $F_{t}^{*}\left(\rest cw\right)=0$ for $c<\ul c$ and $F_{t}^{*}\left(\rest cw\right)=1$
for $c\geq\ol c+\d$. This guarantees \eqref{eq:PS_vt} at any $c\in\R\backslash{\cal C}$.

~

This completes the construction $F_{t}^{*}\left(\rest cw\right)$
for all $c\in\R$ at each $t=1,...,T$. Together, we have ensured
that: 

(a) $F_{t}^{*}\left(\rest{\cd}w\right)$ is a proper conditional CDF;

(b) partial stationarity holds since \eqref{eq:PS_vt} is satisfied
for all $c\in\R$;

(c) period-$t$ marginal CCPs are matched since \eqref{eq:CCP_M}
holds for all $c\in{\cal C}_{t}$ (in Step 2).

\noindent Observe also that each $F_{t}^{*}\left(\rest{\cd}w\right)$
defines a discrete distribution with finite support points. 

\end{proof}
\begin{lem}[Marginal to Joint]
\label{lem:SharpJoint} There exists a well-defined joint distribution
of $\left(v_{i1}^{*},...,v_{iT}^{*}\right)$ with period$-t$ marginal
CDF (conditional on $w$) given by
\[
F_{t}^{*}\left(\rest{\cd}w\right)
\]
as constructed in Lemma \ref{lem:SharpMarginal} such that \eqref{eq:CCP_Joint}
holds.
\end{lem}
\begin{rem}
\label{rem:mar_joint_cts}For each $w$, the constructed per-period
marginals $F_{t}^{*}\left(\rest{\cd}w\right)$ from Lemma \ref{lem:SharpMarginal}
defines a discrete distribution with finite support points. This remains
true for the $F_{t}^{*}\left(\rest{\cd}w\right)$ constructed in the
proof of sharpness in the continuous case (Theorem \ref{thm:sharp_cts}):
even though $w$ is continuously distributed, the CDF $F_{t}^{*}\left(\rest{\cd}w\right)$
remains a discrete one. Since the subsequent proof for Lemma \ref{lem:SharpJoint}
is conditional on $w$ and only utilizes the discreteness of $F_{t}^{*}\left(\rest{\cd}w\right)$,
Lemma \ref{lem:SharpJoint} also holds for $F_{t}^{*}\left(\rest{\cd}w\right)$
constructed in the proof of Theorem \ref{thm:sharp_cts} as well.
\end{rem}
\begin{proof}
For each $w$, the constructed per-period marginals $F_{t}^{*}\left(\rest{\cd}w\right)$
from Lemma \ref{lem:SharpMarginal} defines a discrete distribution
with finite support points. Let ${\cal C}_{w}$ denote the union of
support points of $F_{t}^{*}\left(\rest{\cd}w\right)$ across all
$t=1,...,T$, and let $f_{t}^{*}\left(\rest{\cd}w\right)$ denote
the corresponding probability mass function for $F_{t}^{*}\left(\rest{\cd}w\right)$.
Then, by definition,
\[
F_{t}^{*}\left(\rest cw\right)=\sum_{\tilde{c}\in{\cal C}_{w}:\tilde{c}\leq c}f_{t}^{*}\left(\rest{\tilde{c}}w\right),\quad\forall c.
\]
We now show how to construct a joint pmf $f^{*}\left(\rest{\cd}w\right)$
whose period-$t$ marginals are given by $f_{t}^{*}\left(\rest{\cd}w\right)$. 

For each $t$, define
\begin{equation}
c_{t}^{*}:=\max\left\{ c\in{\cal C}_{w}:\ F_{t}^{*}\left(\rest cw\right)=F_{t}^{*}\left(\rest{w_{t}^{'}\t}w\right)\right\} ,\label{eq:ct*}
\end{equation}
which exists and is unique by the construction in Lemma \ref{lem:SharpMarginal}.

For each ${\bf c}\equiv\left(c_{1},...,c_{T}\right)\in{\cal C}_{w}^{T}$,
write
\begin{align*}
y_{t}\left(c_{t}\right) & :=\ind\left\{ c_{t}\leq c_{t}^{*}\right\} ,\\
y\left({\bf c}\right) & :=\left(y_{1}\left(c_{1}\right),...,y_{T}\left(c_{T}\right)\right)^{'}.
\end{align*}
and define
\begin{equation}
f^{*}\left(\rest{{\bf c}}w\right):=p\left(\rest{y\left({\bf c}\right)}w\right)\prod_{t=1}^{T}\frac{f_{t}^{*}\left(\rest{c_{t}}w\right)}{p_{t}\left(w\right)^{y_{t}\left(c_{t}\right)}\left(1-p_{t}\left(w\right)\right)^{1-y_{t}\left(c_{t}\right)}},\label{eq:p*cw}
\end{equation}
under the convention $0^{0}=1.$

We show that $f^{*}\left(\rest{\cd}w\right)$ is a probability mass
function that characterizes a well-defined joint distribution of $\left(v_{i1}^{*},...,v_{iT}^{*}\right)$
and satisfies the requirements in Lemma \ref{lem:SharpJoint}.

\textbf{Step 1:}

First, note that the right-hand \eqref{eq:p*cw} only involves known
(observed or constructed) quantities. In particular:
\begin{itemize}
\item $p\left(\rest yw\right):=\P\left(\rest{Y_{it}=y_{t}\forall t=1,...,T}W_{i}=w\right)$
is the (observed) joint CCP of observing a particular path of outcomes
$y$ across all periods, given $W_{i}=w$.
\item $f_{t}^{*}\left(\rest cw\right)$ is the period-$t$ marginal pmf
corresponding to $F_{t}^{*}\left(\rest cw\right)$ defined in Lemma
\ref{lem:SharpMarginal}.
\item $f_{t}\left(w\right)=\P\left(\rest{Y_{it}=1}W_{i}=w\right)$ is the
observed period-$t$ marginal CCP, with 
\begin{equation}
p_{t}\left(w\right)=F_{t}\left(\rest{c_{t}^{*}}w\right)=\sum_{\tilde{c}\in{\cal C}_{w}^{T}:\tilde{c}\leq c_{t}^{*}}f_{t}^{*}\left(\rest{\tilde{c}}w\right).\label{eq:pt_f*t}
\end{equation}
\end{itemize}
\textbf{$\quad\ $Step 2:}

We show that the period-$t$ marginal pmf implied by $f^{*}\left(\rest{\cd}w\right)$
coincides with $f_{t}^{*}\left(\rest{\cd}w\right)$. To see this,
observe that, for any $t$ and $y_{t}\in\left\{ 0,1\right\} $, we
have
\begin{align}
 & \sum_{c_{t}\in{\cal C}_{w}^{T}:y_{t}\left(c_{t}\right)=y_{t}}\frac{f_{t}^{*}\left(\rest{c_{t}}w\right)}{p_{t}\left(w\right)^{y_{t}\left(c_{t}\right)}\left(1-p_{t}\left(w\right)\right)^{1-y_{t}\left(c_{t}\right)}}\nonumber \\
= & y_{t}\sum_{c_{t}\leq c_{t}^{*}}\frac{f_{t}^{*}\left(\rest{c_{t}}w\right)}{p_{t}\left(w\right)}+\left(1-y_{t}\right)\sum_{c_{t}>c_{t}^{*}}\frac{f_{t}^{*}\left(\rest{c_{t}}w\right)}{1-p_{t}\left(w\right)}\nonumber \\
= & y_{t}\frac{\sum_{c_{t}\leq c_{t}^{*}}f_{t}^{*}\left(\rest{c_{t}}w\right)}{\sum_{c_{t}\leq c_{t}^{*}}f_{t}^{*}\left(\rest{c_{t}}w\right)}+\left(1-y_{t}\right)\frac{\sum_{c_{t}>c_{t}^{*}}f_{t}^{*}\left(\rest{c_{t}}w\right)}{\sum_{c_{t}>c_{t}^{*}}f_{t}^{*}\left(\rest{c_{t}}w\right)}\text{ by }\eqref{eq:pt_f*t}\nonumber \\
= & y_{t}\cd1+\left(1-y_{t}\right)\cd1\nonumber \\
= & 1,\label{eq:sum_to_1}
\end{align}
Hence, for any $c_{t}\in\ol{{\cal C}},$the period-$t$ marginal implied
by $f^{*}\left(\rest{\cd}w\right)$ is
\begin{align*}
 & \sum_{c_{-t}\in\ol{{\cal C}}^{T-1}}f^{*}\left(\rest{c_{t},c_{-t}}w\right)\\
= & \frac{f_{t}^{*}\left(\rest{c_{t}}w\right)}{p_{t}\left(w\right)^{y_{t}\left(c_{t}\right)}\left(1-p_{t}\left(w\right)\right)^{1-y_{t}\left(c_{t}\right)}}\sum_{c_{-t}}p\left(\rest{y\left(c_{t},c_{-t}\right)}w\right)\prod_{s\neq t}\frac{f_{s}^{*}\left(\rest{c_{s}}w\right)}{p_{s}\left(w\right)^{y_{s}\left(c_{s}\right)}\left(1-p_{s}\left(w\right)\right)^{1-y_{s}\left(c_{s}\right)}}\\
= & \frac{f_{t}^{*}\left(\rest{c_{t}}w\right)}{p_{t}\left(w\right)^{y_{t}\left(c_{t}\right)}\left(1-p_{t}\left(w\right)\right)^{1-y_{t}\left(c_{t}\right)}}\sum_{y_{-t}}p\left(\rest{y_{t}\left(c_{t}\right),y_{-t}}w\right)\sum_{c_{-t}:y_{-t}\left(c_{-t}\right)=y_{-t}}\prod_{s\neq t}\frac{f_{s}^{*}\left(\rest{c_{s}}w\right)}{p_{s}\left(w\right)^{y_{s}\left(c_{s}\right)}\left(1-p_{s}\left(w\right)\right)^{1-y_{s}\left(c_{s}\right)}}\\
= & \frac{f_{t}^{*}\left(\rest{c_{t}}w\right)}{p_{t}\left(w\right)^{y_{t}\left(c_{t}\right)}\left(1-p_{t}\left(w\right)\right)^{1-y_{t}\left(c_{t}\right)}}\sum_{y_{-t}}p\left(\rest{y_{t}\left(c_{t}\right),y_{-t}}w\right)\prod_{s\neq t}\sum_{c_{s}:y_{s}\left(c_{s}\right)=y_{s}}\frac{f_{s}^{*}\left(\rest{c_{s}}w\right)}{p_{s}\left(w\right)^{y_{s}\left(c_{s}\right)}\left(1-p_{s}\left(w\right)\right)^{1-y_{s}\left(c_{s}\right)}}\\
= & \frac{f_{t}^{*}\left(\rest{c_{t}}w\right)}{p_{t}\left(w\right)^{y_{t}\left(c_{t}\right)}\left(1-p_{t}\left(w\right)\right)^{1-y_{t}\left(c_{t}\right)}}\sum_{y_{-t}}p\left(\rest{y_{t}\left(c_{t}\right),y_{-t}}w\right)\prod_{s\neq t}1\text{ by \eqref{eq:sum_to_1}}\\
= & \frac{f_{t}^{*}\left(\rest{c_{t}}w\right)}{p_{t}\left(w\right)^{y_{t}\left(c_{t}\right)}\left(1-p_{t}\left(w\right)\right)^{1-y_{t}\left(c_{t}\right)}}p_{t}\left(w\right)^{y_{t}\left(c_{t}\right)}\left(1-p_{t}\left(w\right)\right)^{1-y_{t}\left(c_{t}\right)}\\
= & f_{t}^{*}\left(\rest{c_{t}}w\right).
\end{align*}

\textbf{Step 3:}

We show that $f^{*}\left(\rest{\cd}w\right)$ is a valid joint pmf.
Clearly, $f^{*}\left(\rest{{\bf c}}w\right)\geq0$, since all quantities
on the right-hand side of \eqref{eq:p*cw} are nonnegative. In addition,
since the period-$t$ marginal of $f^{*}\left(\rest{\cd}w\right)$
coincides with $f_{t}^{*}\left(\rest{\cd}w\right)$ as established
in (2), we must have
\[
\sum_{{\bf c}}f^{*}\left(\rest{{\bf c}}w\right)=\sum_{c_{t}}f_{t}^{*}\left(\rest{c_{t}}w\right)=1.
\]
Hence, $f^{*}\left(\rest cw\right)$ is a valid pmf and thus characterizes
a well-defined joint distribution of $\left(v_{i1}^{*},...,v_{iT}^{*}\right)$.

\textbf{Step 4:}

Lastly, we show that \eqref{eq:CCP_Joint} holds under $f^{*}\left(\rest{\cd}w\right)$.
For any $y\in\left\{ 0,1\right\} ^{T},$
\begin{align*}
\begin{aligned} & \P\left(v_{it}^{*}\leq w_{t}^{'}\t\forall t\text{ s.t. }y_{t}=1,\ v_{is}^{*}>w_{s}^{'}\t\forall s\text{ s.t. }y_{s}=0\mid w\right),\\
= & \sum_{{\bf c}}f^{*}\left(\rest{{\bf c}}w\right)\ind\left\{ c_{t}\leq c_{t}^{*}\,\forall t\text{ s.t. }y_{t}=1,\ c_{s}>c_{s}^{*}\ \forall s\text{ s.t. }y_{s}=0\right\} \\
= & \sum_{{\bf c}:y\left({\bf c}\right)=y}f^{*}\left(\rest{{\bf c}}w\right)\\
= & \sum_{{\bf c}:y\left({\bf c}\right)=y}p\left(\rest{y\left({\bf c}\right)}w\right)\prod_{t=1}^{T}\frac{f_{t}^{*}\left(\rest{c_{t}}w\right)}{p_{t}\left(w\right)^{y_{t}\left(c_{t}\right)}\left(1-p_{t}\left(w\right)\right)^{1-y_{t}\left(c_{t}\right)}},\\
= & p\left(\rest yw\right)\sum_{{\bf c}:y\left({\bf c}\right)=y}\prod_{t=1}^{T}\frac{f_{t}^{*}\left(\rest{c_{t}}w\right)}{p_{t}\left(w\right)^{y_{t}\left(c_{t}\right)}\left(1-p_{t}\left(w\right)\right)^{1-y_{t}\left(c_{t}\right)}}\\
= & p\left(\rest yw\right)\prod_{t=1}^{T}\left(\sum_{c_{t}:y_{t}\left(c_{t}\right)=y_{t}}\frac{f_{t}^{*}\left(\rest{c_{t}}w\right)}{p_{t}\left(w\right)^{y_{t}}\left(1-p_{t}\left(w\right)\right)^{1-y_{t}}}\right)\\
= & p\left(\rest yw\right)\prod_{t=1}^{T}1\text{ by }\eqref{eq:sum_to_1}\\
= & p\left(\rest yw\right).
\end{aligned}
\end{align*}
\end{proof}

\subsection{\label{subsec:pf_sharp_cts}Proof of Theorem \ref{thm:sharp_cts}}
\begin{proof}
Since conditional distributions are only defined up to (probability)
measure-zero sets, for sharpness in the continuous case, we only need
to construct the latent distribution so that CCP-J in \ref{def:sharp}
holds almost surely under $\P_{\rest{W_{i}}Z_{i}=z}$.

We now show how the construction and the proof in the discrete case
(Theorem \ref{thm:sharp_disc}) can be adapted to the continuous case.

Let $\t\equiv\left(\b^{'},\g^{'}\right)^{'}\in\T_{I}\backslash\left\{ \t_{0}\right\} $.
Define 
\begin{align*}
L_{t}\left(c\right) & :=\P\left(\rest{Y_{it}=1,\ z_{t}^{'}\b+X_{it}^{'}\g\leq c}z\right),\\
U_{t}\left(c\right) & :=1-\P\left(\rest{Y_{it}=0,\ z_{t}^{'}\b+X_{it}^{'}\g\geq c}z\right),
\end{align*}
and 
\begin{align*}
\ol L\left(c\right):=\max_{s}L_{s}\left(c\right),\quad & \ul U\left(c\right):=\min_{s}U_{s}\left(c\right).
\end{align*}

By Assumption \ref{assu:Cts}(a), $\rest{z_{t}^{'}\b+X_{it}^{'}\g}Z_{i}=z$ is continuously
distributed with a density function on a bounded interval support.
Write $\pi_{t}\left(c\right)$ for this density (conditional on $z$)
and write ${\cal C}_{t}=\left[\ul c_{t},\ol c_{t}\right]$ as its
support. Then, $L_{t}\left(c\right)$ has an integral representation
\begin{align*}
L_{t}\left(c\right) & =\E\left[\rest{p_{t}\left(W_{i}\right)\ind\left\{ z_{t}^{'}\b+X_{it}^{'}\g\leq c\right\} }z\right]\\
 & =\int_{\ul c_{t}}^{c}p_{t}\left(\tilde{c}\right)\pi_{t}\left(\rest{\tilde{c}}z\right)d\tilde{c}
\end{align*}
so that its derivative, by Assumption \ref{assu:Cts}(b), is given by 
\[
L_{t}^{'}\left(c\right)=p_{t}\left(c\right)\pi_{t}\left(\rest cz\right)>0
\]
Hence $L_{t}\left(c\right)$ is continuous and strictly increasing
on $\left[\ul c_{t},\ol c_{t}\right]$. Similarly, 
\begin{align*}
U_{t}\left(c\right) & =1-\P\left(\rest{Y_{it}=0,\ z_{t}^{'}\b+X_{it}^{'}\g\geq c}z\right),\\
 & =1-\E\left[\rest{\left(1-p_{t}\left(W_{i}\right)\right)\ind\left\{ z_{t}^{'}\b+X_{it}^{'}\g\geq c\right\} }z\right]\\
 & =1-\int_{c}^{\ol c_{t}}\left(1-p_{t}\left(\tilde{c}\right)\right)\pi_{t}\left(\rest{\tilde{c}}z\right)d\tilde{c}
\end{align*}
with derivative 
\[
U_{t}^{'}\left(c\right):=\left(1-p_{t}\left(c\right)\right)\pi_{t}\left(\rest cz\right)>0.
\]
Hence $U_{t}\left(c\right)$ is also continuous and strictly increasing
on $\left[\ul c_{t},\ol c_{t}\right]$.

\textbf{Step 1: }

Let $F_{t}^{L}\left(\rest cw\right)$ and $F_{t}^{U}\left(\rest cw\right)$
be defined as before. Again, we have 
\begin{align*}
F_{t}^{L}\left(\rest cz\right) & =L_{t}\left(c\right)
\end{align*}
but now

\begin{align*}
F_{t}^{U}\left(\rest cz\right) & =1-\P\left(\rest{Y_{i}=0\text{ and }z_{t}^{'}\b+X_{it}^{'}\g\geq c-\d}Z_{i}=z\right)\\
 & =U_{t}\left(\rest{c-\d}z\right)\\
 & <U_{t}\left(\rest cz\right)\text{ on }c\in\left[\ul c_{t},\ol c_{t}+\d\right]
\end{align*}
Hence, the key step in adapting the discrete-case construction to
the continuous case is to ensure the mismatch between $F_{t}^{U}\left(\rest cz\right)$
and $U_{t}\left(c\right)$ can be properly handled.

\textbf{Step 2: }

Let ${\cal {\cal L}}_{t}$, ${\cal U}_{t}$ and $q^{*}$ be defined
as before. For any $c\in\left[\ul c_{t},\ol c_{t}\right]$, we again
consider the following two cases:

Case 1: $\ol L\left(c\right)<q^{*}=L_{t}\left(\ol c_{t}\right)$.

Since $L_{t}$ and $\ol L$ are both continuous and strictly increasing,
we can define 
\[
\psi\left(c\right):=L_{t}^{-1}\left(\ol L\left(c\right)\right)
\]
and set 
\[
F_{t}^{*}\left(\rest cw\right):=F_{t}^{L}\left(\rest{\psi\left(c\right)}w\right)
\]
which ensures partial stationarity at $c$, since 
\begin{align*}
F_{t}^{*}\left(\rest cz\right) & =F_{t}^{L}\left(\rest{\psi\left(c\right)}z\right)=L_{t}\left(\psi\left(c\right)\right)=L_{t}\left(L_{t}^{-1}\left(\ol L\left(c\right)\right)\right)=\ol L\left(c\right).
\end{align*}
In addition, notice that since $L_{t}\left(c\right)\leq\ol L\left(c\right)<L_{t}\left(\ol c_{t}\right)$,
we must have 
\[
c\leq\psi\left(c\right)<\ol c_{t}.
\]
Hence, if $w$ is such that $w_{t}^{'}\t=c$, we have $w_{t}^{'}\t\leq\psi\left(w_{t}^{'}\t\right)<\ol c_{t}$
and thus by the definition of $F_{t}^{L}$ 
\[
F_{t}^{*}\left(\rest{w_{t}^{'}\t}w\right):=F_{t}^{L}\left(\rest{\psi\left(w_{t}^{'}\t\right)}w\right)=p_{t}\left(w\right).
\]
Lastly, observe that $\psi$ is increasing, and hence $F_{t}^{*}\left(\rest cw\right)$
is weakly increasing with 
\[
\text{\ensuremath{F_{t}^{*}\left(\rest cw\right)}}=F_{t}^{L}\left(\rest{\psi\left(c\right)}w\right)\leq p_{t}\left(w\right)
\]
in this case.

Case 2: $\ol L\left(c\right)\geq q^{*}$.

As before, since $U_{t}$ and $\ol L$ are both continuous and strictly
increasing, we can define 
\[
\psi\left(c\right):=U_{t}^{-1}\left(\ol L\left(c\right)\right)
\]
and set 
\[
F_{t}^{*}\left(\rest cw\right):=F_{t}^{U}\left(\rest{\psi\left(c\right)+\d}w\right).
\]
This construction again ensures the partial stationarity condition
at $c$: 
\begin{align*}
F_{t}^{*}\left(\rest cz\right) & =F_{t}^{U}\left(\rest{\psi\left(c\right)+\d}z\right)=U_{t}\left((\psi\left(c\right)+\d)-\d\right)\\
 & =U_{t}\left(\psi\left(c\right)\right)=U_{t}\left(U_{t}^{-1}\left(\ol L\left(c\right)\right)\right)\\
 & =\ol L\left(c\right)
\end{align*}
Furthermore, notice that $q^{*}=U_{t}\left(\ul c_{t}\right)\leq\ol L\left(c\right)\leq U_{t}\left(c\right)$,
we must have 
\[
\ul c_{t}\leq\psi\left(c\right)\leq c.
\]
Hence, $F_{t}^{*}\left(\rest cw\right)$ is weakly increasing given
that $\psi$ is increasing, with 
\[
F_{t}^{*}\left(\rest cw\right)\geq F_{t}^{U}\left(\rest{\ul c_{t}+\d}w\right)\geq F_{t}^{U}\left(\rest{\ul c_{t}}w\right)=p_{t}\left(w\right)
\]
in this case.

We now investigate the period-$t$ CCP matching condition. We consider
two subcases.

Subcase 1a: $\ol L\left(c\right)<U_{t}\left(c\right)$. In this subcase,
we must have $\ul c_{t}\leq\psi\left(c\right)<c$ and thus 
\[
\ul c_{t}<\psi\left(c\right)+\d<c+\d.
\]
Then, if $w$ is such that $w_{t}^{'}\t=c$, we have $\ul c_{t}<\psi\left(w_{t}^{'}\t\right)+\d<w_{t}^{'}\t+\d$
and thus by the definition of $F_{t}^{U}$ 
\[
F_{t}^{*}\left(\rest{w_{t}^{'}\t}w\right):=F_{t}^{U}\left(\rest{\psi\left(w_{t}^{'}\t\right)+\d}w\right)=p_{t}\left(w\right),
\]
which verifies the period-$t$ CCP matching condition in this subcase.

Subcase 1b: $\ol L\left(c\right)=U_{t}\left(c\right)$. In this subcase,
CCP matching will not be satisfied, since $\ol L\left(c\right)=U_{t}\left(c\right)$
implies $\psi\left(c\right)=c$. Hence, if $w$ is such that $w_{t}^{'}\t=c$,
we will have 
\[
F_{t}^{*}\left(\rest{w_{t}^{'}\t}w\right):=F_{t}^{U}\left(\rest{w_{t}^{'}\t+\d}w\right)=1\neq p_{t}\left(w\right).
\]
However, we will argue that this mismatch can be ignored under Condition
(C3), which essentially implies that such mismatch happens with probability
zero and is thus ignorable.

~

We now argue that $F_{t}^*\left(\rest cw\right)$ must be weakly increasing
on ${\cal C}_{t}=\left[\ul c_{t},\ol c_{t}\right]$. For $c\in\left[\ul c_{t},\ol L^{-1}\left(q^{*}\right)\right)$,
we have $\ol L\left(c\right)<q^{*}$ as in Case 2, where we have established
$F_t^*\left(\rest cw\right)$ is weakly increasing with $F_t^*\left(\rest cw\right)<p_{t}\left(w\right)$
in this region. For $c\in\left[\ol L^{-1}\left(q^{*}\right),\ol c_{t}\right)$,
we have $\ol L\left(c\right)>q^{*}$ as in Case 2 and again we have
established $F_t^*\left(\rest cw\right)$ is weakly increasing with $F\left(\rest cw\right)\geq p_{t}\left(w\right)$
in this region. Hence, $F_{t}^*\left(\rest cw\right)$ must be weakly
increasing on ${\cal C}_{t}=\left[\ul c_{t},\ol c_{t}\right]$.

The rest of the construction of $F^{*}\left(\rest{\cd}w\right)$,
as well as the corresponding proof, proceed exactly the same as in
the discrete case. In particular, notice that Lemma \ref{lem:SharpJoint}
continues to apply to the $F^{*}\left(\rest{\cd}w\right)$ constructed
here, as discussed in Remark \ref{rem:mar_joint_cts}.

~ 

In summary, in the continuous case, conditional on $z$, we have constructed
$F^{*}\left(\rest{\cd}w\right)$ that: (1) the partial stationarity
condition exactly (2) the CCP matching condition at any $w$ except
for those such that 
\[
{\cal W}^{\circ}=\left\{ w:\ul L\left(w_{t}^{'}\t\right)=U_{t}\left(w_{t}^{'}\t\right)\quad\text{for some }t\right\} .
\]
However, by Assumption \eqref{assu:Cts}(a), $W_{it}^{'}\t|z$ is
continuously distributed with a density function on its support, and
thus it follows from Assumption \eqref{assu:Cts}(c) that ${\cal W}^{\circ}$
is a probability-zero set under $\P_{\rest{W_{i}}z}$. Hence, the
CCP matching condition is satisfied almost surely under $\P_{\rest{W_{i}}z}$,
which suffices for sharpness. 
\end{proof}

\subsection{Reconciliation with KPT}

\label{appe:RelateKPT} We show that under Assumption \ref{assu:PartStat}
and $X_{it}=Y_{i,t-1}$, our identifying condition \eqref{eq:ID_Set}
implies the following result in KPT:

KPT(i): $\P\left(\rest{Y_{it}=1}z\right)>\P\left(\rest{Y_{is}=1}z\right)$
$\imp$ $\left(z_{t}-z_{s}\right)^{'}\b_{0}+\left|\g_{0}\right|>0.$

KPT(ii): $\P\left(\rest{Y_{it}=1}z\right)>1-\P\left(\rest{Y_{i,s}=0,Y_{i,s-1}=1}z\right)$
$\imp$ $\left(z_{t}-z_{s}\right)^{'}\b_{0}-\min\left\{ 0,\g_{0}\right\} >0.$

KPT(iii): $\P\left(\rest{Y_{it}=1}z\right)>1-\P\left(\rest{Y_{i,s}=0,Y_{i,s-1}=0}z\right)$
$\imp$ $\left(z_{t}-z_{s}\right)^{'}\b_{0}+\max\left\{ 0,\g_{0}\right\} >0.$

KPT(iv): $\P\left(\rest{Y_{it}=1,Y_{it-1}=1}z\right)>\P\left(\rest{Y_{is}=1}z\right)$
$\imp$ $\left(z_{t}-z_{s}\right)^{'}\b_{0}+\max\left\{ 0,\g_{0}\right\} >0.$

KPT(v): $\P\left(\rest{Y_{it}=1,Y_{it-1}=1}z\right)>1-\P\left(\rest{Y_{is}=0,Y_{i,s-1}=1}z\right)$
$\imp$ $\left(z_{t}-z_{s}\right)^{'}\b_{0}>0.$

KPT(vi): $\P\left(\rest{Y_{it}=1,Y_{it-1}=1}z\right)>1-\P\left(\rest{Y_{is}=0,Y_{i,s-1}=0}z\right)$
$\imp$ $\left(z_{t}-z_{s}\right)^{'}\b_{0}+\g_{0}>0.$

KPT(vii): $\P\left(\rest{Y_{it}=1,Y_{it-1}=0}z\right)>1-\P\left(\rest{Y_{is}=0}z\right)$
$\imp$ $\left(z_{t}-z_{s}\right)^{'}\b_{0}-\min\left\{ 0,\g_{0}\right\} >0.$

KPT(viii): $\P\left(\rest{Y_{it}=1,Y_{it-1}=0}z\right)>1-\P\left(\rest{Y_{is}=0\text{\ensuremath{,}}Y_{i,s-1}=1}z\right)$
$\imp$ $\left(z_{t}-z_{s}\right)^{'}\b_{0}-\g_{0}>0.$

KPT(ix): $\P\left(\rest{Y_{it}=1,Y_{it-1}=0}z\right)>1-\P\left(\rest{Y_{is}=0,Y_{i,s-1}=0}z\right)$
$\imp$ $\left(z_{t}-z_{s}\right)^{'}\b_{0}>0.$

~
\begin{proof}
With $X_{it}=Y_{i,t-1}$, our inequality restriction \eqref{eq:ID_ts_KPT}
can be equivalently rewritten as follows: 
\begin{align}
 & \P\left(\rest{Y_{it}=1,\ Y_{i,t-1}=1}z\right)\ind\left\{ z_{t}^{'}\b_{0}+\g_{0}\leq c\right\} +\P\left(\rest{Y_{it}=1,\ Y_{i,t-1}=0}z\right)\ind\left\{ z_{t}^{'}\b_{0}\leq c\right\} \nonumber \\
\leq\  & 1-\P\left(\rest{Y_{is}=0,\ Y_{i,s-1}=1}z\right)\ind\left\{ z_{s}^{'}\b_{0}+\g_{0}\geq c\right\} -\P\left(\rest{Y_{is}=0,\ Y_{i,s-1}=0}z\right)\ind\left\{ z_{s}^{'}\b_{0}\geq c\right\} ,\label{eq:dyn_ineq}
\end{align}
by enumerating the realization of $Y_{i,t-1}$.

Note that the lower and upper expressions in the inequality \eqref{eq:dyn_ineq}
both have three possible (informative) outcomes depending on the value
of $c$, leading to the 9 inequalities in KPT. We derive the first
two inequalities KPT(i) and KPT(ii), and the rest of inequalities
can be derived in the same way.

KPT(i): consider the event that all indicators in condition \eqref{eq:dyn_ineq}
are equal to one, saying that 
\[
\max\{z_{t}^{'}\b_{0}+\g_{0},z_{t}^{'}\b_{0}\}\leq c\leq\min\{z_{s}^{'}\b_{0}+\g_{0},z_{s}^{'}\b_{0}\},
\]
which is equivalent to 
\[
z_{t}^{'}\b_{0}+\max\{0,\g_{0}\}-(z_{s}^{'}\b_{0}+\min\{0,\g_{0}\})=(z_{t}-z_{s})^{'}\b_{0}+|\g_{0}|\leq0.
\]

Then, when $(z_{t}-z_{s})^{'}\b_{0}+|\g_{0}|\leq0$, condition \eqref{eq:dyn_ineq}
becomes 
\begin{align*}
\P\left(\rest{Y_{it}=1}z\right) & =\P\left(\rest{Y_{it}=1,\ Y_{i,t-1}=1}z\right)+\P\left(\rest{Y_{it}=1,\ Y_{i,t-1}=0}z\right)\\
 & \leq\ 1-\P\left(\rest{Y_{is}=0,\ Y_{i,s-1}=1}z\right)-\P\left(\rest{Y_{is}=0,\ Y_{i,s-1}=0}z\right)\\
 & =1-\P(Y_{is}=0\mid z)=\P\left(\rest{Y_{is}=1}z\right).
\end{align*}

By contraposition, it implies the same restriction in KPT(i): 
\[
\P\left(\rest{Y_{it}=1}z\right)>\P\left(\rest{Y_{is}=1}z\right)\Longrightarrow(z_{t}-z_{s})^{'}\b_{0}+|\g_{0}|>0.
\]

~

KPT(ii): we first relax condition \eqref{eq:dyn_ineq} by dropping
the last term in the upper expression $\P\left(\rest{Y_{is}=0,\ Y_{i,s-1}=0}z\right)\ind\left\{ z_{s}^{'}\b_{0}\geq c\right\} $
and have the following relaxed inequality: 
\begin{align}
 & \P\left(\rest{Y_{it}=1,\ Y_{i,t-1}=1}z\right)\ind\left\{ z_{t}^{'}\b_{0}+\g_{0}\leq c\right\} +\P\left(\rest{Y_{it}=1,\ Y_{i,t-1}=0}z\right)\ind\left\{ z_{t}^{'}\b_{0}\leq c\right\} \nonumber \\
\leq\  & 1-\P\left(\rest{Y_{is}=0,\ Y_{i,s-1}=1}z\right)\ind\left\{ z_{s}^{'}\b_{0}+\g_{0}\geq c\right\} .\label{eq:dyn_ineq_relax}
\end{align}

Now, consider the event that the indicators in the above restriction
are all equal to one, which implies that 
\[
\max\{z_{t}^{'}\b_{0}+\g_{0},z_{t}^{'}\b_{0}\}\leq c\leq z_{s}^{'}\b_{0}+\g_{0},
\]
and it is equivalent to the following condition: 
\[
(z_{t}-z_{s})^{'}\b_{0}+\max\{0,\g_{0}\}-\g_{0}=(z_{t}-z_{s})^{'}\b_{0}-\min\{0,\g_{0}\}\leq0.
\]
Given the above event, condition \eqref{eq:dyn_ineq_relax} becomes
\begin{align*}
\P\left(\rest{Y_{it}=1}z\right) & =\P\left(\rest{Y_{it}=1,\ Y_{i,t-1}=1}z\right)+\P\left(\rest{Y_{it}=1,\ Y_{i,t-1}=0}z\right)\\
 & \leq\ 1-\P\left(\rest{Y_{is}=0,\ Y_{i,s-1}=1}z\right).
\end{align*}

Similarly, we can derive the same restriction in KPT(ii) by contraposition:
\[
\P\left(\rest{Y_{it}=1}z\right)>1-\P\left(\rest{Y_{is}=0,\ Y_{i,s-1}=1}z\right)\ \imp\ (z_{t}-z_{s})^{'}\b_{0}-\min\{0,\g_{0}\}>0.
\]
\end{proof}

\subsection{Proof of Proposition \ref{prop:mul}}
\begin{proof}
Let $v_{ijt}:=\a_{ij}+\e_{ijt}$, for any set $K\subset\mathcal{J}$,
the probability of selecting a choice $j\in K$ conditional on $W_{i}=w$
is given as: 
\[
\P(Y_{it}^{K} = 1\mid w)=\P(Y_{it}\in K\mid w)=\P\left(\exists j\in K\text{ s.t. }w_{ijt}'\t_{0}+v_{ijt}\geq w_{ikt}'\t_{0}+v_{ikt}\ \forall k\in K^{c}\mid w\right).
\]

The above observed probability restricts the conditional distribution
of $v_{ikt}-v_{ijt}\mid w$ and can be exploited to bound this distribution.

We define $Q_{t}(c_{jk}\mid w)$ as follows: for $c_{jk}\in\mathcal{R}$,
\[
\begin{aligned}Q_{t}(c_{jk}\mid w):=\P\left(\exists j\in K\text{ s.t. }v_{ikt}-v_{ijt}\leq c_{jk}\ \forall k\in\mathcal{J}\setminus K\mid w\right).\end{aligned}
\]

Then, we can derive lower and upper bounds for the above probability
using variations in observed choice probabilities. When $c_{jk}$
satisfies $c_{jk}\geq(w_{ijt}-w_{ikt})'\t_{0}$ for any $j\in K$
and $k\in\mathcal{J}\setminus K$, then $Q_{t}(c_{jk}\mid w)$ can
be bounded below as 
\[
\begin{aligned}Q_{t}(c_{jk}\mid w) & \geq\P\left(\exists j\in K\text{ s.t. }v_{ikt}-v_{ijt}\leq(w_{ijt}-w_{ikt})'\t_{0}\ \forall k\in\mathcal{J}\setminus K\mid w\right)\\
 & =\P(Y_{it}\in K\mid w).
\end{aligned}
\]

Therefore, the lower bound for $Q_{t}(c_{jk}\mid w)$ is established
as 
\[
Q_{t}(c_{jk}\mid w)\geq\P(Y_{it}\in K,\ c_{jk}\geq(w_{ijt}-w_{ikt})'\t_{0}\ \forall j\in K,k\in k\in\mathcal{J}\setminus K\mid w).
\]
The above inequality holds since either $c_{jk}\geq(w_{ijt}-w_{ikt})'\t_{0}$
or the lower bound is zero.

By taking expectation of $X_{i}$ given $z$, we can bound the conditional
distribution $Q_{t}(c_{jk}\mid z)$ as 
\[
\begin{aligned}Q_{t}(c_{jk}\mid z) & \geq\P\left(Y_{it}\in K,c_{jk}\geq(z_{ijt}-z_{ikt})'\b_{0}+(X_{ijt}-X_{ikt})'\g_{0}\ \forall j\in K,k\in\mathcal{J}\setminus K\mid z\right)\\
 & =\P\left(Y_{it}^{K}=1,c_{jk}\geq(z_{ijt}-z_{ikt})'\b_{0}+(X_{ijt}-X_{ikt})'\g_{0}\ \forall j\in K,k\in\mathcal{J}\setminus K\mid z\right).
\end{aligned}
\]

Similarly, the conditional probability $Q_{t}(c_{jk}\mid w)$ can
be bounded above as 
\begin{multline*}
Q_{t}(c_{jk}\mid w)\leq\P(Y_{it}^{K}=1\mid w)\ind\{c_{jk}\leq(w_{ijt}-w_{ikt})'\t_{0}\ \forall j\in K,\mathcal{J}\setminus K\}+\\
1-\ind\{c_{jk}\leq(w_{ijt}-w_{ikt})'\t_{0}\ \forall j\in K,k\in\mathcal{J}\setminus K\}.
\end{multline*}
The above inequality holds since either $c_{jk}\leq(w_{ijt}-w_{ikt})'\t_{0}$
or the upper bound is one with $c_{jk}>(w_{ijt}-w_{ikt})'\t_{0}$.
After taking expectation of $X_{i}$ given $z$, the upper bound for
$Q_{t}(c_{jk}\mid z)$ is obtained as 
\begin{multline*}
Q_{t}(c_{jk}\mid z)\leq\P\left(Y_{it}^{K}=1,c_{jk}\leq(z_{ijt}-z_{ikt})'\b_{0}+(X_{ijt}-X_{ikt})'\g_{0}\ \forall j\in K,k\in K^{c}\mid z\right)\\
+1-\P\left(c_{jk}\leq(z_{ijt}-z_{ikt})'\b_{0}+(X_{ijt}-X_{ikt})'\g_{0}\ \forall j\in K,k\in\mathcal{J}\setminus K\mid z\right).
\end{multline*}

Rearranging the above formula yields 
\[
Q_{t}(c_{jk}\mid z)\leq1-\P\left(Y_{it}^{K}=0,c_{jk}\leq(z_{ijt}-z_{ikt})'\b_{0}+(X_{ijt}-X_{ikt})'\g_{0}\ \forall j\in K,k\in\mathcal{J}\setminus K\mid z\right).
\]

Under Assumption \ref{assu:PartStat}, the conditional probability
$Q_{t}(c_{jk}\mid z)$ is the same for any $t$. Therefore, the smallest
upper bound of $Q_{t}(c_{jk}\mid z)$ should be larger than the largest
lower bound over all periods, yielding the identifying condition \eqref{eq:ind_mul}
as follows: 
\begin{multline*}
1-\max_{s=1,...,T}\P(Y_{is}^{K}=0,(z_{js}-z_{ks})'\b_{0}+(X_{ijs}-X_{iks})'\g_{0}\geq c_{jk}\ \forall j\in K,k\in\mathcal{J}\setminus K\mid z)\\
\geq\max_{t=1,...,T}\P(Y_{it}^{K}=1,(z_{jt}-z_{kt})'\b_{0}+(X_{ijt}-X_{ikt})'\g_{0}\leq c_{jk}\ \forall j\in K,k\in\mathcal{J}\setminus K\mid z).
\end{multline*}
\end{proof}

\section{\label{sec:Supplement}Supplemental Results and Discussions}

\subsection{Binary Choice: Point Identification}

\label{subsec:point}

Proposition \ref{prop:bin} characterizes the sharp identified set
for $\t_{0}$ by only imposing Assumption \ref{assu:PartStat}. This
section provides sufficient conditions to achieve point identification
for $\b_{0}$ (up to scale) and the sign of $\g_{0}$ under support
conditions on the exogenous covariate $Z_{it}$. We focus on the scenario
where the endogenous covariate $X_{it}$ is discrete $X_{it}\in\mathcal{X}\equiv\left\{ \ol x_{1},...,\ol x_{K}\right\} $
and there are only two periods $T=2$.

We start by noting that Section 3 of KPT contains a detailed discussion
about point identification for the AR(1) setting $X_{it}=Y_{i,t-1}$.
Since our identification result becomes equivalent to theirs in the
AR(1) setting, the sufficient conditions they provide there still
apply. Hence, in this section, we seek to provide some sufficient
condition with a general $X_{it}$ that may not be the one-period
lag $Y_{i,t-1}$.

To point identify $\b_{0}$, the first step is to determine the sign
of the covariate index $(Z_{i2}-Z_{i1})'\b_{0}$ under certain variation
of observed choice probability. To identify the sign of $(Z_{i2}-Z_{i1})'\b_{0}$,
we define the following two sets: 
\[
\begin{aligned} & \mathcal{Z}_{1}:=\Big\{(z_{1},z_{2})\mid\exists x\in\mathcal{X}\text{ s.t. }1-\P(Y_{i1}=0,X_{i1}=x\mid z)<\P(Y_{i2}=1,X_{i2}=x\mid z)\Big\},\\
 & \mathcal{Z}_{2}:=\Big\{(z_{1},z_{2})\mid\exists x\in\mathcal{X}\text{ s.t. }1-\P(Y_{i1}=1,X_{i1}=x\mid z)<\P(Y_{i2}=0,X_{i2}=x\mid z)\Big\}.
\end{aligned}
\]

Let $\mathcal{Z}:=\mathcal{Z}_{1}\cup\mathcal{Z}_{2}$. Let $\D Z_{i}=Z_{i2}-Z_{i1}$
and $\D\mathcal{Z}$ be defined as 
\[
\D\mathcal{Z}:=\Big\{\D z:=z_{2}-z_{1}\mid(z_{1},z_{2})\in\mathcal{Z}\Big\}.
\]

As shown in the proof below, when $\D z$ satisfies $\D z\in\D\mathcal{Z}$,
the sign of $\D z'\b_{0}$ is identified. In the definition of the
two sets $\mathcal{Z}_{1},\mathcal{Z}_{2}$, we only need the existence
of one value in the support of $\mathcal{X}$ such that the choice
probability in the two sets are observed. When observing such choice
probability, the sign of $\D z'\b_{0}$ is identified. Then $\b_{0}$
can be identified up to scale under the standard large support condition
on $\D z$.

Let $\D z^{j}$ denote the $j$-th element of $\D z$. The following
is the support condition on the covariate. 
\begin{assumption}[Support Condition]
 \label{assu:supp} (1) $\D\mathcal{Z}$ is not contained in any
proper linear subspace of $\R^{d_{z}}$; (2) for any $\D z\in\D\mathcal{Z}$,
there exists one element $\D z^{j^{*}}$ such that $\b_{0}^{j*}\neq0$,
and the conditional support of $\D z^{j^{*}}$ is $\R$ given $\D z\setminus\D z^{j^{*}}$,
where $\D z\setminus\D z^{j^{*}}$ denote the remaining components
of $\D z$ besides $\D z^{j^{*}}$. 
\end{assumption}
\begin{prop}
\label{prop:point} Under Assumptions \ref{assu:PartStat}-\ref{assu:supp},
the parameter $\b_{0}$ is point identified up to scale. 
\end{prop}
We provide point identification for $\b_{0}$ with two periods $T=2$.
When there are more than two periods, then we only require the existence
of two periods, satisfying Assumption \ref{assu:supp}. As shown in
\citet{manski1987}, the large support assumption is necessary to
point identify $\b_{0}$, as without it, there exists some $b\neq k\b_{0}$
such that $\D z'b$ has the same sign with $\D z'\b_{0}$ when $\D z$
has bounded support.

The parameter $\g_{0}$ in general can be only partially identified
given potential endogeneity of $X_{it}$ and flexible structures on
$(\a_{i},\e_{it})$. Nevertheless, we can still bound the value $(x_{1}-x_{2})'\g_{0}$
and identify the sign of $\g_{0}$ under certain choice probabilities.
We derive the sufficient conditions to identify the sign of $\g_{0}$.

Let $x^{j}$ denote the $j$-th element of $x$ and $\g_{0}^{j}$
denote the $j$-th coefficient of $\g_{0}$. We define the following
two sets: 
\[
\begin{aligned}\mathcal{Z}_{3}^{j}:=\Big\{(z_{1},z_{2})\mid\exists x_{1},x_{2} & \in\mathcal{X}\text{ with }x_{1}^{j}\neq x_{2}^{j},x_{1}^{m}=x_{2}^{m}\ \forall m\neq j\text{ s.t. }\\
 & 1-\P(Y_{i1}=0,X_{i1}=x_{1}\mid z)<\P(Y_{i2}=1,X_{i2}=x_{2}\mid z)\Big\};\\
\mathcal{Z}_{4}^{j}:=\Big\{(z_{1},z_{2})\mid\exists x_{1},x_{2} & \in\mathcal{X}\text{ with }x_{1}^{j}\neq x_{2}^{j},x_{1}^{m}=x_{2}^{m},\ \forall m\neq j\text{ s.t. }\\
 & 1-\P(Y_{i1}=1,X_{i1}=x_{1}\mid z)<\P(Y_{i2}=0,X_{i2}=x_{2}\mid z)\Big\}.
\end{aligned}
\]

From the identifying results in Proposition \ref{prop:bin}, the value
of $(x_{1}^{j}-x_{2}^{j})\g_{0}^{j}$ can be bounded when $(z_{1},z_{2})$
belong to the two sets: 
\[
\begin{aligned} & (z_{1},z_{2})\in\mathcal{Z}_{3}^{j}\ \Longrightarrow\ (x_{1}^{j}-x_{2}^{j})\g_{0}^{j}<\D z'\b_{0},\\
 & (z_{1},z_{2})\in\mathcal{Z}_{4}^{j}\ \Longrightarrow\ (x_{1}^{j}-x_{2}^{j})\g_{0}^{j}>\D z'\b_{0}.
\end{aligned}
\]

Then the sign of $\g_{0}^{j}$ is identified if either the sign of
$\D z'\b_{0}$ is identified as negative when $(z_{1},z_{2})\in\mathcal{Z}_{2}$
or as positive when $(z_{1},z_{2})\in\mathcal{Z}_{1}$. 
\begin{prop}
\label{prop:point_gamma} Under Assumptions \ref{assu:PartStat},
and for any $1\leq j\leq d_{x}$, either $\mathcal{Z}_{3}^{j}\cap\mathcal{Z}_{2}\neq\emptyset$
or $\mathcal{Z}_{4}^{j}\cap\mathcal{Z}_{1}\neq\emptyset$, then the
sign of $\g_{0}$ is identified.
\end{prop}
When the endogenous variable $X_{it}$ is a scalar, e.g., the lagged
dependent variable $X_{it}=Y_{i,t-1}$, then the definition of the
two sets $\mathcal{Z}_{3}^{j},\mathcal{Z}_{4}^{j}$ can be simplified
as there existing $x_{1}\neq x_{2}$ such that the corresponding choice
probability is observed. Besides the sign of $\g_{0}$, the identification
results can also bound the value of $\g_{0}$ from variation in the
exogenous covariates.

When $X_{it}$ is multi-dimensional such as including two lagged dependent
variable $X_{it}=(Y_{i,t-1},Y_{i,t-2})$ with $\g_{0}=(\g_{0}^{1},\g_{0}^{2})$,
then $\g_{0}^{1}$ is identified when the required choice probability
in the two sets $\mathcal{Z}_{3}^{1},\mathcal{Z}_{4}^{1}$ are observed
for $(Y_{i,1},Y_{i,0})=(1,1),(Y_{i,2},Y_{i,1})=(0,1)$ or $(Y_{i,1},Y_{i,0})=(0,0),(Y_{i,2},Y_{i,1})=(1,0)$.
We provide general sufficient conditions to identify the sign of $\g_{0}$,
which may be stronger than necessary and can be relaxed in certain
scenarios. For example, when we know that $\g_{0}^{1}+\g_{0}^{2}>0$
while $\g_{0}^{1}<0$, we can infer that $\g_{0}^{2}>0$ without requiring
additional assumptions on the two sets $\mathcal{Z}_{3}^{2},\mathcal{Z}_{4}^{2}$.

~
\begin{proof}[Proof of Propositions \ref{prop:point} and \ref{prop:point_gamma}]
The proof for the point identification of $\b_{0}$ consists of two
steps: we first show that when $\D z\in\D\mathcal{Z}$, the sign of
$\D z'\b_{0}$ is identified from the identifying condition \eqref{eq:ID_Set}
in Proposition \ref{prop:bin}. Then, the large support condition
in Assumption \ref{assu:supp} ensures that $\b_{0}$ is point identified
up to scale.

When $X_{it}$ is discrete and there are two periods $T=2$, the identifying
condition \eqref{eq:ID_Set} is given as 
\[
1-\P(Y_{i1}=0,z_{1}'\b_{0}+X_{i1}'\g_{0}\geq c\mid z)\geq\P(Y_{i2}=1,z_{2}'\b_{0}+X_{i2}'\g_{0}\leq c\mid z),
\]
for $c\in\{z_{t}'\b_{0}+x_{k}'\g_{0},\ t=1,2,k=1,...,K\}$, and another
identifying condition switches the order of period 1 and 2.

Let $c=z_{1}'\b_{0}+x_{k}'\g_{0}$,\footnote{The value of $c=z_{2}'\b_{0}+x_{k}'\g_{0}$ leads to the same identifying
condition.}, then the above upper bound can be further bounded as 
\[
1-\P(Y_{i1}=0,z_{1}'\b_{0}+X_{i1}'\g_{0}\geq z_{1}'\b_{0}+x_{k}'\g_{0}\mid z)\leq1-\P(Y_{i1}=0,X_{i1}=x_{k}\mid z).
\]

When $z_{1}'\b_{0}-z_{2}'\b_{0}\geq0$ which implies $z_{1}'\b_{0}+x_{k}'\g_{0}\geq z_{2}'\b_{0}+x_{k}'\g_{0}$,
then the lower bound can be bounded below as 
\[
\P(Y_{i2}=1,z_{2}'\b_{0}+X_{i2}'\g_{0}\leq z_{1}'\b_{0}+x_{k}'\g_{0}\mid z)\geq\P(Y_{i2}=1,X_{i2}=x_{k}\mid z).
\]

Combining the above results leads to 
\[
\text{If }z_{1}'\b_{0}-z_{2}'\b_{0}\geq0\Longrightarrow1-\P(Y_{i1}=0,X_{i1}=x_{k}\mid z)\geq\P(Y_{i2}=1,X_{i2}=x_{k}\mid z).
\]

The contraposition of the above inequality yields 
\[
1-\P(Y_{i1}=0,X_{i1}=x_{k}\mid z)<\P(Y_{i2}=1,X_{i2}=x_{k}\mid z)\Longrightarrow\D z'\b_{0}>0.
\]

Switching the order of the time period leads to another identifying
restriction as follows: 
\[
1-\P(Y_{i1}=1,X_{i1}=x_{k}\mid z)<\P(Y_{i2}=0,X_{i2}=x_{k}\mid z)\Longrightarrow\D z'\b_{0}<0.
\]

Therefore, when $\D z\in\D\mathcal{Z}$, the sign of $\D z'\b_{0}$
is identified.

~

Next, we show that $\b_{0}$ is point identified under the large support
assumption. To prove it, we will show that for any $\beta\neq k\b_{0}$
for some $k$, there exists some value $\D z$ such that $\D z'b$
has different signs from $\D z'\b_{0}$.

From Assumption \ref{assu:supp}, the conditional support of $\Delta z^{j^{*}}$
is $\mathcal{R}$ and $\b_{0}^{j^{*}}\neq0$. We focus on the case
where $\b_{0}^{j^{*}}>0$, and the analysis also applies to the other
case. Let $\Delta\tilde{z}:=\Delta z\setminus\Delta z^{j^{*}}$ denote
the remaining covariates in $\Delta z$ and $\tilde{\beta}_{0}$ denote
its coefficient. For any candidate $b$, we discuss three cases: $b^{j^{*}}<0$,
$b^{j^{*}}=0$, and $b^{j^{*}}>0$.

Case 1: $b^{j^{*}}<0$. When the covariate $\Delta z^{j^{*}}$ takes
a large positive value $\Delta z^{j^{*}}\rightarrow+\infty$ and the
remaining covariates take bounded values in their support, it implies
that $\Delta z'\b_{0}>0$ and $\Delta z'b<0$.

Case 2: $b^{j^{*}}=0$. For any value $\Delta z$, the value of $\Delta z'b$
is either positive or nonpositive. When $\Delta z'b>0$ is positive,
then let $\Delta z^{j*}$ take a large negative value $\Delta z^{j*}\rightarrow-\infty$
such that $\Delta z'\b_{0}<0$, which has a different sign from $\Delta z'b$.
Similarly, if $\Delta z'b\leq0$, there exists $\Delta z^{j^{*}}\rightarrow+\infty$
such that $\Delta z'\beta_{0}>0$.

Case 3: $b^{j^{*}}>0$. Assumption \ref{assu:supp} requires that
$\D\mathcal{Z}$ is not contained in any proper linear subspace, so
there exists $\Delta\tilde{z}$ such that $\Delta\tilde{z}'\tilde{\b}_{0}/\b_{0}^{j^{*}}\neq\Delta\tilde{z}'\tilde{b}/b^{j^{*}}$.
Suppose that $\Delta\tilde{z}'\tilde{\b}_{0}/\b_{0}^{j^{*}}-\Delta\tilde{z}'\tilde{b}/b^{j^{*}}=k>0$,
then when the covariate takes the value $\Delta Z_{i}=-\Delta\tilde{z}'\tilde{b}/b^{j^{*}}-\epsilon$
with $0<\epsilon<k$. The sign of the covariate index satisfies: $\Delta z'\b_{0}=\b_{0}^{j^{*}}(k-\e)>0$
and $\Delta z'b=-b^{j^{*}}\e<0$. The construction is similar when
$k<0$.

For the identification of $\g_{0}$, under the similar analysis for
$\b_{0}$, we have 
\[
\begin{aligned} & (z_{1},z_{2})\in\mathcal{Z}_{3}^{j}\ \Longrightarrow\ (x_{1}^{j}-x_{2}^{j})\g_{0}^{j}<\D z'\b_{0},\\
 & (z_{1},z_{2})\in\mathcal{Z}_{4}^{j}\ \Longrightarrow\ (x_{1}^{j}-x_{2}^{j})\g_{0}^{j}>\D z'\b_{0}.
\end{aligned}
\]

As previously shown, when $(z_{1},z_{2})\in\mathcal{Z}_{2}$, it implies
that $\D z'\b_{0}<0$. Therefore, when $(z_{1},z_{2})\in\mathcal{Z}_{2}\cap\mathcal{Z}_{3}^{j}$,
we have $(x_{1}^{j}-x_{2}^{j})\g_{0}^{j}<\D z'\b_{0}<0$ and the sign
of $\g_{0}^{j}$ is identified given $x_{1}^{j}\neq x_{2}^{j}$. Similarly,
when $(z_{1},z_{2})\in\mathcal{Z}_{1}\cap\mathcal{Z}_{4}^{j}$, the
sign of $\g_{0}^{j}$ is also identified given $(x_{1}^{j}-x_{2}^{j})\g_{0}^{j}>\D z'\b_{0}>0$.
Proposition \ref{prop:point_gamma} requires that for any $j\leq d_{x}$,
either $\mathcal{Z}_{2}\cap\mathcal{Z}_{3}^{j}\neq\emptyset$ or $\mathcal{Z}_{1}\cap\mathcal{Z}_{4}^{j}\neq\emptyset$
so that the sign of $\g_{0}^{j}$ is identified for any $j$. 
\end{proof}

\subsection{\label{subsec:ID_CF}Binary Choice: Counterfactual Parameters}

In previous subsections, we have focused on the (partial) identification
of the index parameters $\t_{0}$. Here we show how our identification
results can also be leveraged to (partially) identify counterfactual
parameters.

Write $W_{i}:=\left(Z_{i},X_{i}\right)$ in short, and correspondingly
$w:=\left(z,x\right)$, and $w_{t}^{'}\t=z_{t}^{'}\b+x_{t}^{'}\g$.
Consider a general counterfactual change in the observable covariates
$W_{i}$ from $w$ to $\tilde{w}$, and the consequent counterfactual
period-$t$ conditional choice probability of the form 
\begin{equation}
\tilde{p}_{t}\left(\tilde{w}\right):=\P\left(\rest{v_{it}\leq\tilde{w}^{'}\t_{0}}W_{i}=w\right).\label{eq:pt_CF}
\end{equation}
Importantly, in the definition above, the utility index is changed
from $w_{t}^{'}\t_{0}$ to the counterfactual $\tilde{w}_{t}^{'}\t_{0}$,
while the conditional distribution of the latent $v_{it}$ is held
unchanged at $\rest{v_{it}}W_{i}=w$. Hence, $\tilde{p}_{t}\left(\tilde{w}\right)$
can be interpreted as a counterfactual CCP induced by an exogenous
policy intervention that only changes the characteristics from $w$
to $\tilde{w}$, but leaves all other unobserved individual heterogeneity
reflected in the distribution of $v_{it}$ unchanged. In particular,
note that the (partial) derivative of $\tilde{p}_{t}\left(w\right)$
can be interpreted as average marginal effects.\footnote{Here, ``average'' refers to the averaging over unobserved individual
heterogeneity in $\rest{\left(\a_{i},\e_{it}\right)}w_{i}$. The counterfactual
CCP $\tilde{p}_{t}\left(w\right)$, or its derivative, can be further
averaged over (subvectors of) $w$ to produce additional average effects
that are averaged over observed individual heterogeneity. }

Note that the (partial) identification of counterfactual CCP $\tilde{p}_{t}\left(\tilde{w}\right)$
relies on the identification of the index parameter $\t_{0}$ as well
as the identification of the latent conditional distribution $\rest{v_{it}}W_{i}=w$,
which also involves the endogenous covariates $X_{i}$. It turns out
that, our key identification strategy in Section \ref{subsec:ind}
also provides a straightforward way to derive bounds on $F_{t}\left(\rest cw\right)$,
the CDF of $\rest{v_{it}}W_{i}=w$ at any point $c$, by taking conditional
expectations of \eqref{eq:vleqc_LB} and \eqref{eq:vleqc_UB} given
$W_{i}=w$ (instead of $Z_{i}=z$ as in Section \ref{subsec:ind}):
\begin{equation}
\P\left(\rest{Y_{it}=1,\ w_{t}^{'}\t_{0}\leq c}W_{i}=w\right)\leq F_{t}\left(\rest cw\right)\leq1-\P\left(\rest{Y_{it}=0,\ w_{t}^{'}\t_{0}\geq c}W_{i}=w\right),\label{eq:Bounds_Ft(c|w)}
\end{equation}
which can then be combined with Proposition \ref{prop:bin} to derive
the bounds in Proposition \ref{prop:CFbound}.
\begin{prop}[Bounds on Counterfactual CCP]
 \label{prop:CFbound} Under model \ref{eq:bin} and Assumption
\ref{assu:PartStat}, 
\begin{equation}
\inf_{\t\in\T_{I}}\P\left(\rest{Y_{it}=1,\ w_{t}^{'}\t\leq\tilde{w}_{t}^{'}\t}W_{i}=w\right)\leq\tilde{p}_{t}\left(\tilde{w}\right)\leq1-\inf_{\t\in\T_{I}}\P\left(\rest{Y_{it}=0,\ w_{t}^{'}\t\geq\tilde{w}_{t}^{'}\t}W_{i}=w\right).\label{eq:CFbound}
\end{equation}
\end{prop}
\begin{proof}
By \eqref{eq:Bounds_Ft(c|w)}, we have 
\[
\P\left(\rest{Y_{it}=1,\ w_{t}^{'}\t_{0}\leq c}W_{i}=w\right)\leq F_{t}\left(\rest cw\right)\leq1-\P\left(\rest{Y_{it}=0,\ w_{t}^{'}\t_{0}\geq c}W_{i}=w\right).
\]
Since $\tilde{p}_{t}\left(\tilde{w}\right)=F_{t}\left(\rest{\tilde{w}_{t}^{'}\t_{0}}w\right)$,
we have 
\[
\P\left(\rest{Y_{it}=1,\ w_{t}^{'}\t_{0}\leq\tilde{w}_{t}^{'}\t_{0}}W_{i}=w\right)\leq\tilde{p}_{t}\left(\tilde{w}\right)\leq1-\P\left(\rest{Y_{it}=0,\ w_{t}^{'}\t_{0}\geq\tilde{w}_{t}^{'}\t_{0}}W_{i}=w\right),
\]
and hence 
\[
\inf_{\t\in\T_{I}}\P\left(\rest{Y_{it}=1,\ w_{t}^{'}\t\leq\tilde{w}_{t}^{'}\t}W_{i}=w\right)\leq\tilde{p}_{t}\left(\tilde{w}\right)\leq1-\inf_{\t\in\T_{I}}\P\left(\rest{Y_{it}=0,\ w_{t}^{'}\t\geq\tilde{w}_{t}^{'}\t}W_{i}=w\right).
\]
\end{proof}
The lower and upper bounds in Proposition \ref{prop:CFbound} above
are identified since the involved conditional probabilities are all
about observed data $\left(Y_{i},W_{i}\right)$ for each $\t\in\T_{I}$,
while the set $\T_{I}$ is identified by Proposition \ref{prop:bin}.
Hence, Proposition \ref{prop:CFbound} establishes the partial identification
of the counterfactual CCP $\tilde{p}_{t}\left(\tilde{w}\right)$.

\subsection{\label{subsec:Yi0}Binary Choice: Initial Conditions}

In the main text, we treat all covariates in $X_{it}$ as observed
and endogenous. In the specific context of dynamic binary choice model,
say, the AR(1) model with $X_{it}=Y_{i,t-1}$, we are effectively
treating the initial condition $Y_{i0}$ as observed and endogenous
(and thus not conditioned upon). In this appendix, we consider some
alternative setups, and explain how our approach can be adapted accordingly.

~

For illustration, we focus on the AR(1) dynamic binary choice.
\[
Y_{it}=\ind\left\{ Z_{it}^{'}\b_{0}+\g_{0}Y_{i,t-1}+X_{it}^{'}\l_{0}+\a_{i}+\e_{it}\geq0\right\} 
\]
with $Y_{i,t-1}$ explicitly written out.

If $Y_{i,0}$ is \emph{observed }and treated as ``\emph{exogenous}'',
i.e., if we impose the partial stationarity condition conditional
on $Y_{i0}$ in addition to $Z_{i}$, i.e.,
\[
\rest{\e_{it}}Z_{i},Y_{i0},\a_{i}\sim\rest{\e_{is}}Z_{i},Y_{i0},\a_{i},
\]
then we can replicate our identification arguments conditional on
$Z_{i}=z,Y_{i0}=y_{0}$. Then Proposition \ref{prop:bin} holds with
the same forms of CCPs conditioned on $Y_{i0}=y_{0}$ in addition
to $Z_{i}=z$. In particular, the parametric index in the first-period
CCP
\[
W_{i1}^{'}\t=Z_{i1}^{'}\b+\g Y_{i,0}+X_{i1}^{'}\l=z_{1}^{'}\b+\g y_{0}+X_{i1}^{'}\l
\]
would its $\g Y_{i0}$ component degenerate to $\g y_{0}$, but, for
$t=2,...,T$, this index
\[
W_{it}^{'}\t=Z_{i1}^{'}\b+\g Y_{i,t-1}+X_{it}^{'}\l
\]
will still involve randomness in $Y_{i,t-1}$, conditional on $Z_{i}=z,Y_{i0}=y_{0}$.

If $Y_{i,0}$ is \emph{unobserved, }then we still work with the same
partial stationarity assumption conditional on $Z_{i}=z$ only, and
Proposition \ref{prop:bin} holds with the following specially adapted
bounds for $\P(v_{i1}\leq c\mid z)$ using observations from period
$t=1$:

Specifically, for period $t=1$, we exploit
\[
\ind\left\{ Y_{i1}=1\right\} \ind\left\{ z_{1}^{'}\b_{0}+\max\{0,\g_{0}\}+X_{i1}^{'}\l_{0}\leq c\right\} \leq\ind\left\{ v_{i1}\leq c\right\} 
\]
which does not involve the unobserved $Y_{i0}$ but nevertheless produces
a valid lower bound in the form of
\[
\P(Y_{i1}=1,\ z_{1}^{'}\b_{0}+\max\{0,\g_{0}\}+X_{i1}^{'}\l_{0}\leq c\mid z)\leq\P(v_{i1}\leq c\mid z.)
\]
Similarly, we can also provide an upper bound in the form of 
\[
\P(v_{i1}\leq c\mid z)\leq1-\P(Y_{i1}=0,z_{1}'\b_{0}+\min\{0,\g_{0}\}+X_{i1}^{\prime}\lambda_{0}\geq c\mid z).
\]
Again, notice that $\P(v_{i1}\leq c\mid z)=\P\left(\rest{v_{it}\leq c}z\right)$
for all $t=2,...,T$ by partial stationarity, so the above special
bounds for $t=1$ can be aggregated with bounds derived from other
periods to produce bounds
on $\P(v_{it}\leq c\mid z)$ as before.

\subsection{Binary Choice: Sufficient Conditions for Assumption \ref{assu:Cts}(c)}\label{subsec:cond_A2c}

In this subsection, we provide some examples and primitive conditions for Assumption \ref{assu:Cts}(c) to illustrate our point that Assumption \ref{assu:Cts}(c) is a rather mild assumption that is likely to be ``generically" satisfied.

Recall that Assumption \ref{assu:Cts}{(c) rules out cases in which the lower and upper envelopes
\[
L(c \mid z,\theta_0)
\;=\;
\max_{t=1,\dots,T} L_t(c \mid z,\theta_0),
\qquad
U(c \mid z,\theta_0)
\;=\;
\min_{t=1,\dots,T} U_t(c \mid z,\theta_0),
\]
coincide on a non-trivial interval of the index $c$. Under Assumption \ref{assu:Cts}(a)(b),
each $L_t(\cdot \mid z,\theta_0)$ and $U_t(\cdot \mid z,\theta_0)$ is continuous and
strictly increasing in $c$, so $L(c \mid z,\theta_0)=U(c \mid z,\theta_0)$
can hold on an interval only if some pair $L_t$ and $U_s$ actually coincide on that interval.

We first give some simple examples where Assumption \ref{assu:Cts}(c) is automatically
satisfied, and then provide a set of primitive sufficient conditions.

\subsubsection{A Simple Example where Assumption \ref{assu:Cts}(c) Trivially Holds}

Consider the leading dynamic binary panel model
\[
Y_{it}
=
\mathbf{1}\{W_{it}'\theta_0 + \alpha_i + \varepsilon_{it} \ge 0\},
\qquad
W_{it}'\theta_0 = X_{it}'\gamma_0,
\]
with $T=2$ and no fixed effect $\a_i$ and time-varying $Z$ for simplicity. Suppose:
\begin{enumerate}
  \item[(i)] The shocks $\varepsilon_{it}$ are i.i.d.\ logistic and independent of
  $(\alpha_i,x_{i1},x_{i2})$, so
  \[
  P(Y_{it}=1 \mid W_{it}'\theta_0 = c)
  = F(c) := \frac{e^c}{1+e^c},
  \]
  where $F$ is strictly increasing and continuous in $c$.

  \item[(ii)] The index $W_{it}'\theta_0$ has a common conditional density
  $g(c)$ across $t=1,2$, with $g(c)>0$ on some bounded interval
  (this is a special case of Assumption~2(a)).
\end{enumerate}

In this case, the derivative expressions for the envelopes reduce to
\[
L_1'(c)
= F(c) g(c),
\qquad
U_2'(c)
= [1-F(c)] g(c).
\]
If $L(c)=U(c)$ holds on a non-degenerate interval, then necessarily
$L_1(c)=U_2(c)$ on (a subinterval of) that region, which implies
$L_1'(c)=U_2'(c)$ there. Thus
\[
F(c) g(c) = [1-F(c)] g(c)
\quad\text{for all $c$ in that interval},
\]
or equivalently $F(c)=1/2$ on that interval. But $F$ is strictly increasing,
so the equation $F(c)=1/2$ has at most one solution. Hence the set of
$c$ for which $L(c)=U(c)$ has no interval and, in fact, consists of at most
a single point. Assumption \ref{assu:Cts}(c) therefore holds automatically in the familiar
dynamic logit (and similarly dynamic probit) case with a continuous index.

This illustrates the nature of Assumption \ref{assu:Cts}(c): it excludes knife-edge cases
in which the CCPs line up in such a way that the lower and upper envelopes
touch over a non-trivial index interval, which would require a flat segment
of the choice probability function.

\subsubsection{Another Simple Example}
We now provide another simple example to illustrate Assumption \ref{assu:Cts}(c) in the fully-exogenous case where
\( W_{it}'\theta_0 = z_{it}'\beta_0 \) and \( T = 2 \).

In this setting, the lower and upper bounds simplify to
\[
\begin{aligned}
L_{t}\left(\rest c z,\theta\right)
&:= \Pr\left( Y_{it}=1,\; z_t' \beta \le c \mid z \right), \\
U_{t}\left(\rest c z,\theta\right)
&:= 1 - \Pr\left( Y_{it}=0,\; z_t' \beta \ge c \mid z \right).
\end{aligned}
\]
Define \( P_t(z) := \Pr(Y_{it} = 1 \mid z) \).
Without loss of generality, assume \( z_1' \beta < z_2' \beta \) (the analysis is symmetric when \( z_1' \beta > z_2' \beta \)). We consider the following three cases.

\noindent\textbf{Case 1:} \( c > z_2' \beta \).
In this case,
\[
\underline{L}(\rest c z) = \max \{ P_1(z), P_2(z) \},
\qquad
\overline{U}(\rest c z) = 1.
\]
Assumption \ref{assu:Cts}(c) is satisfied.

\medskip

\noindent
\textbf{Case 2:} \( c < z_1' \beta \).
In this case,
\[
\underline{L}(\rest c z) = 0,
\qquad
\overline{U}(\rest c z) = \min \{ P_1(z), P_2(z) \}.
\]
Assumption \ref{assu:Cts}(c) is satisfied.

\medskip

\noindent
\textbf{Case 3:} \( z_1' \beta < c < z_2' \beta \).
In this case,
\[
\underline{L}(\rest c z) = P_1(z),
\qquad
\overline{U}(\rest c z) = P_2(z).
\]
Assumption \ref{assu:Cts}(c) is satisfied provided that \( P_1(z) \neq P_2(z) \). A sufficient condition for this inequality is that the conditional density of error term $\epsilon_{it}\mid z$ is positive everywhere. 

\subsubsection{Primitive Sufficient Conditions for Assumption \ref{assu:Cts}(c)}

We now give conditions under which Assumption \ref{assu:Cts}(c) follows from more primitive
restrictions on the conditional choice probabilities and the index densities.

Fix $z$ and let
\[
m_t(c,z)
\;:=\;
P\big(Y_{it}=1 \mid W_{it}'\theta_0=c,\ Z_i=z\big)
\]
and let $g_t(c,z)$ denotes the p.d.f. of  $W_{it}'\theta_0 \mid Z_i=z$.
Under Assumption \ref{assu:Cts}(a)(b), each $g_t(\cdot,z)$ is continuous and strictly
positive on a bounded interval, and the CCPs $m_t(\cdot,z)$ take values
strictly between $0$ and $1$. Using the derivative expressions,
we can write
\begin{equation}
L_t'(c \mid z,\theta_0)
= m_t(c,z)\,g_t(c,z),
\qquad
U_s'(c \mid z,\theta_0)
= [1 - m_s(c,z)]\,g_s(c,z).
\label{eq:LtUs-derivatives}
\end{equation}

Define the following condition:

\medskip\noindent
\textbf{Condition S (Regularity and No Exact Mirror Equality).} For each $z$ and each pair
$t \neq s$, the mappings
\[
c \longmapsto m_t(c,z)\,g_t(c,z)
\quad\text{and}\quad
c \longmapsto [1-m_s(c,z)]\,g_s(c,z)
\]
are piecewise real analytic\footnote{A real-valued function is real analytic if it can be represented as a convergent power series locally around every point of its domain. A function is piecewise real analytic if the function is real analytic on each cell of a finite partition of its domain. This regularity condition is imposed to rule out pathological level sets that are nowhere dense but have strictly positive Lebesgue measures, such as the Fat Cantor set. Such sets are of limited relevance in economics and econometrics, which is why we are imposing the piecewise real analytic condition to rule them out. That said, we note that piecewise real analytic is sufficient but unnecessary.} and not identical on any non-degenerate interval in $c$.
\medskip

\noindent
We show that Assumption \ref{assu:Cts}(c) follows from Assumption \ref{assu:Cts}2(a)(b) plus Condition S.

\medskip\noindent
\begin{prop} Suppose Assumption \ref{assu:Cts}(a)(b) and Condition S hold. Then Assumption \ref{assu:Cts}(c) holds.
\end{prop}
\medskip\noindent

\begin{proof}
Fix $z$ and suppress $(z,\theta_0)$ in the notation when there is no ambiguity. Write
$\ol{L}(c):=\ol{L}(c\mid z,\theta_0)$, $\ul{U}(c):=\ul{U}(c\mid z,\theta_0)$, and similarly
$L_t(c):=L_t(c\mid z,\theta_0)$ and $U_s(c):=U_s(c\mid z,\theta_0)$.
Define the contact set
\[
E \;:=\; \{c:\, \ol{L}(c)=\ul{U}(c)\}.
\]
For any $c\in E$, choose $t(c)\in\arg\max_{t=1,\dots,T} L_t(c)$ and
$s(c)\in\arg\min_{s=1,\dots,T} U_s(c)$. Then
\[
L_{t(c)}(c)=\ol{L}(c)=\ul{U}(c)=U_{s(c)}(c).
\]
Moreover, $t(c)\neq s(c)$ because Assumption~2(a)(b) implies $L_t(c)<U_t(c)$ for all $t$ and all $c$
in the support. Hence
\[
E \subseteq \bigcup_{t\neq s} E_{ts},
\qquad
E_{ts} \;:=\; \{c:\, L_t(c)=U_s(c)\}.
\]
Since there are finitely many pairs $(t,s)$, it suffices to show that each $E_{ts}$ has Lebesgue
measure zero.

Fix $t\neq s$ and define $H_{ts}(c):=L_t(c)-U_s(c)$. By \eqref{eq:LtUs-derivatives}, for all $c$ in the interior of the
support,
\[
H'_{ts}(c)
= L_t'(c)-U_s'(c)
= m_t(c,z)g_t(c,z)-\bigl[1-m_s(c,z)\bigr]g_s(c,z)
=: h_{ts}(c).
\]
Under Condition~S, $h_{ts}$ is piecewise real analytic: there exists a finite partition of the
support into non-overlapping intervals $\{I_k\}_{k=1}^K$ such that $h_{ts}$ is real analytic on each
$I_k$. Fix $k$ and pick any $c_k\in I_k$. For $c\in I_k$,
\[
H_{ts}(c) \;=\; H_{ts}(c_k) + \int_{c_k}^{c} h_{ts}(u)\,du.
\]
Since the antiderivative of a real analytic function is real analytic, it follows that $H_{ts}$ is
real analytic on each $I_k$.

Suppose that $E_{ts}$ has strictly positive Lebesgue measure. Then for some
$k$ we must have $Leb(E_{ts}\cap I_k)>0$. But $E_{ts}\cap I_k$ is exactly the zero set of the
real analytic function $H_{ts}$ restricted to $I_k$. A real analytic function on an interval that
vanishes on a set of positive Lebesgue measure must be identically zero on that interval (equivalently,
zeros of a nontrivial real analytic function are isolated). Therefore $H_{ts}(c)=0$ for all $c\in I_k$,
and hence $h_{ts}(c)=H'_{ts}(c)=0$ for all $c\in I_k$. This implies
\[
m_t(c,z)g_t(c,z)=\bigl[1-m_s(c,z)\bigr]g_s(c,z)\qquad\text{for all }c\in I_k,
\]
contradicting Condition~S (which rules out such equality on any non-degenerate interval).

Thus $Leb(E_{ts})=0$ for every $t\neq s$. Since $E$ is contained in a finite union of the
$E_{ts}$, we conclude $Leb(E)=0$, which is Assumption 2(c).
\end{proof}

Condition S has a simple interpretation. If the index densities differ across periods
($g_t(\cdot,z) \neq g_s(\cdot,z)$), or if the CCPs $m_t(\cdot,z)$ and $m_s(\cdot,z)$
are not related in such an exact ``mirror'' way, then the equality
\[
m_t(c,z)\,g_t(c,z) = [1-m_s(c,z)]\,g_s(c,z)
\]
cannot hold on a whole interval. Even when $g_t=g_s$, Condition~S reduces to ruling out
\[
m_t(c,z) = 1 - m_s(c,z)
\quad\text{for all $c$ in an interval,}
\]
i.e.\ having one period's success probability be exactly the complement of another's
over a range of the index. In standard parametric logit or probit models, such a relation
would require very special (non-generic) parameter restrictions, and hence Condition S
is automatically satisfied under mild cross-time variation in either the conditional
densities or the conditional choice probabilities.

\subsection{Discussion about 
General-Case Sharpness}\label{subsec:sharp_gen}

We now provide a discussion of the issue of sharpness in the general case as considered in Section \ref{subsec:generic}, by first converting the
general model \eqref{eq:model_gen} into a family of binary outcome
models. Specifically, under weak monotonicity of $G$ in its first argument
(Assumption 3), we can define the pseudo-inverse of $G$ in its first
argument as
\[
G^{-1}\left(y,\a,\e\right):=\inf\left\{ c:G\left(c,\a,\e\right)\geq y\right\} ,\quad\forall y\in{\cal Y}.
\]
Then, given any $y\in{\cal Y}$, we can write
\begin{align*}
v_{it}\left(y\right):=G^{-1}\left(y,\a_{i},\e_{it}\right),\quad & Y_{it}\left(y\right):=\ind\left\{ Y_{it}\geq y\right\} 
\end{align*}
and obtain the binarized model
\begin{equation}
Y_{it}\left(y\right)=\ind\left\{ W_{it}^{'}\t_{0}\geq v_{it}\left(y\right)\right\} ,\label{eq:gen_bin}
\end{equation}
which, at a each given $y$, is the same as binary choice model \eqref{eq:bin}
written in terms of $v_{it}$.\footnote{This also shows that scalar-additivity of $\a_{i}$ and $\e_{it}$
is not a binding restriction in the binary choice model.} Note that since $G$ is weakly increasing in $c$, $G^{-1}$ must
be weakly increasing in $y$ as well. Hence, $v_{it}\left(y\right)$
is a stochastic process weakly increasing in $y\in{\cal Y}$, and
its CDF given $W_{i}=w$, denoted by $F\left(\rest{c,y}w\right):=\P\left(\rest{v_{it}\left(y\right)\leq c}w\right)$
must be decreasing in $y$ at any given $c$.

Given the binary representation \eqref{eq:bin}, any $\t\in\T_{I,gen}\backslash\left\{ \t_{0}\right\} $,
we can follow the proof of Theorem \ref{thm:sharp_disc} or \ref{thm:sharp_cts}
to construct a latent distribution (CDF) $F^{*}\left(\rest{c,y}w\right)$
for each given $y$, which satisfies the partial stationarity assumption
and matches all observable CCPs $\P\left(\rest{Y_{it}\geq y}w\right)$
at each $y$. This essentially asserts ``sharpness at each $y$''
separately. 

What remains is to establish sharpness across all $y\in{\cal Y}$
jointly, and the key issue here is to ensure that the constructed
latent CDF $F^{*}\left(\rest{c,y}w\right)$ is weakly decreasing in
$y$ as the $F\left(\rest{c,y}w\right)$ is. Such monotonicity ensures
the existence of a unified stochastic process $v_{it}^{*}\left(y\right)$
with CDF $F^{*}\left(\rest{c,y}w\right)$ at each $y$. However, even
though it is straightforward to establish that our construction of
$F^{*}\left(\rest{c,y}w\right)$ ensures the associated $F^{*}\left(\rest{c,y}z\right)$
is weakly decreasing in $y$, it is less obvious whether $F^{*}\left(\rest{c,y}w\right)$
is. In particular, a key step to establish the weak motonocity of
$F^{*}\left(\rest{c,y}w\right)$ requires certain conditions on the
curvature of the functions $L_{t}$ and $U_{t}$ \eqref{eq:Lt_Ut}
in $y$. It is not obvious to us whether such conditions are plausible or not under the current general setup.

\subsection{Static Ordered Response Model}\label{subsec:order_stat}

This section studies the static ordered response model without endogeneity and provides simplified identifying conditions under this structure. In the static model, the full stationarity assumption holds conditional on all covariates $W_{i}$:
\[
\e_{is} \mid W_{i}, \a_i \stackrel{d}{\sim} \e_{it}  \mid W_{i}, \a_i.
\]

Then, the identifying restriction in Proposition \ref{prop:order} is given as
\[  \max_{t=1,...,T} \sum_{j=1}^J \P\left(Y_{it}=y_j, \ b_{j+1}-w_t'\t_0 \leq c \mid  w \right)\leq 1-\max_{s=1,...,T} \sum_{j=1}^J \P(Y_{is}=y_j, \ b_{j}-w_s'\t_0 \geq c \mid w). 
\]

The above condition is informative only if there exists $j_1, j_2$ such that $b_{j_2+1}-w_t'\t_0 \leq c \leq b_{j_1}-w_s'\t_0$, leading to
\[\begin{aligned}
 \max_{t=1,...,T} \sum_{j=1}^{j_2} \P\left(Y_{it}=y_j \mid  w\right) &\leq 1-\max_{s=1,...,T} \sum_{j=j_1}^J \P(Y_{is}=y_j\mid w) \\
 &=\min_{s=1,...,T} \sum_{j=1}^{j_1-1} \P(Y_{is}=y_j\mid w).
 \end{aligned}
\]

 The following proposition presents the identification results for the static ordered choice model by changing the notation of $j_1-1$ with $j_1$.

\begin{prop}\label{prop:static_ordered}
Assume that $ \e_{is} \mid (W_{i}, \a_i) \stackrel{d}{\sim} \e_{it} \mid (W_{i}, \a_i)$, then $\T_{I, order}$ consists of all $\t=\left(\b^{'},\g^{'}\right)^{'}\in\R^{d_{z}}\times\R^{d_{x}}$
such that
\[
b_{j_1+1}-w_s'\t \geq  b_{j_2+1}-w_t'\t \Longrightarrow 
\min_{s=1,...,T} \sum_{j=1}^{j_1} \P(Y_{is}=y_j \mid w)\geq  \max_{t=1,...,T} \sum_{j=1}^{j_2}\P(Y_{it}=y_j\mid w), 
\]
for any $1 \leq j_1, j_2 \leq J-1$, and any realization $w=(w_1,..., w_T)$ in the support of $W_i$.
\end{prop}
The results in Proposition \ref{prop:static_ordered}  are analogous to the maximum-score type result in \cite{manski1987}, with the distinction being that we can exploit variations in the sum of multiple choices rather than investigating a single choice to identify $\t_0$. Furthermore, with multiple choices, we can utilize variations in the sum of different choices across various periods for identification.

\subsection{Censored Outcome Model \label{subsec:cont}}

The two previous examples primarily investigate discrete choice models.
However, our approach also applies to models with continuous dependent
variables, including those with censored or interval outcomes. To
illustrate, we focus on the following panel model with censored outcomes
as studied in \citet{honore2004estimation} :
\begin{equation}
\begin{aligned}Y_{it}^{*} & =Z_{it}'\b_{0}+X_{it}'\g_{0}+\a_{i}+\e_{it},\\
Y_{it} & =\max\{Y_{it}^{*},0\},
\end{aligned}
\label{eq:model_cen}
\end{equation}
where $Y_{it}^{*}$ denotes the latent outcome which is not observed
in the data, and $Y_{it}$ represents the observed outcome, censored
at zero. 

The endogenous covariate $X_{it}$ can again incorporate lagged dependent
variable $Y_{it-1}$ and other endogenous covariates. With $X_{it}=Y_{i,t-1}$,
model \eqref{eq:model_cen} specializes to the one in \citet{honore1993orthogonality}.
Both \citet{honore1993orthogonality} and \citet{honore2004estimation}
develop orthogonality conditions for these models under the assumption
of conditionally i.i.d. errors $\e_{it}$. 

Alternatively, \citet{hu2002estimation} considers a slightly different
model setup where the dynamic dependence is fully specified on the
latent outcome variable:
\begin{equation}
\begin{aligned}Y_{it}^{*} & =Z_{it}'\b_{0}+Y_{i,t-1}^{*}\g_{0}+\a_{i}+\e_{it},\\
Y_{it} & =\max\{Y_{it}^{*},0\}
\end{aligned}
\label{eq:model_cen_latent}
\end{equation}
Since $Y_{i,t-1}^{*}$ is not observed when $Y_{i,t-1}^{*}<0$, this
model does not fit into our framework directly. However, our key identification
strategy can still be adapted to handle the potential unobservability
of $Y_{i,t-1}^{*}$. 

Below, we consider the two models above separately.

\subsubsection*{Analysis of Model \eqref{eq:model_cen}}

Here we focus on model \eqref{eq:model_cen}, where the endogenous
covariates $X_{it}$ is observed. The identification strategy is still
to exploit the partial stationarity assumption and bound the conditional
distribution of $v_{it}\mid Z_{i}=z$. This censored outcome model
imposes an additional structure between the outcome and the parametric
index: when $Y_{it}>0$, we have $Y_{it}=Y_{it}^{*}$ and 
\[
v_{it}\geq c\ \iff\ Y_{it} \leq Z_{it}'\b_{0}+X_{it}'\g_{0} - c.
\]
This specific structure can be exploited to further tighten the identified
set for $\t_{0}$, and we provide the details of the identification
strategy in Appendix \ref{proof:censor}. The following proposition
presents the identification results of $\t_{0}$ with censored outcomes.
\begin{prop}
\label{prop:censor} Under model \eqref{eq:model_cen} and Assumption
\ref{assu:PartStat}, $\t_{0}\in\T_{I,cen}$, where the identified
set $\T_{I,cen}$ consists of all $\t=\left(\b^{'},\g^{'}\right)^{'}\in\R^{d_{z}}\times\R^{d_{x}}$
such that 
\begin{equation}
\max_{t=1,...,T}\P(Y_{it}\leq z_{t}'\b+X_{it}'\g-c\mid z)\leq\min_{s=1,...,T}\left\{ \P(0<Y_{is}\leq z_{s}'\b+X_{is}'\g-c\mid z)+\P(Y_{is}=0\mid z)\right\} ,\label{eq:censor}
\end{equation}
for any $c\in\R$ and any realization $z=(z_{1},...,z_{T})$ in the
support of $Z_{i}$.
\end{prop}
Similar to discrete choice models studied in previous sections, Proposition
\ref{prop:censor} characterizes an identified set for $\t_{0}$ by
exploiting the variation in the joint distribution of $(Y_{it},X_{it})\mid Z_{i}$
over time and the variation in the exogenous covariates $Z_{i}$.
The bounds on the distribution $v_{it}\mid Z_{i}=z$ can be derived
either from the probability $\P(0<Y_{it}\leq y\mid z)$ or $\P(Y_{it}=0\mid z)$,
depending on the value of the covariate index $z_{t}'\b_{0}+X_{it}'\g_{0}$.
This result still accommodates both static and dynamic models with
censored outcomes.

\label{proof:censor} 
\begin{proof}[Proof of Propositions \ref{prop:censor}]
Since the observed outcome $Y_{it}$ is censored at 0, we either observe
$Y_{it}=y>0$ or $Y_{it}=0$. Let $v_{it}:=-(\a_{i}+\e_{it})$, the
conditional probability of $Y_{it}=0$ is given as, 
\[
\P(Y_{it}=0\mid w)=\P(Y_{it}^{*}\leq0\mid w)=\P(v_{it}\geq w_{t}'\t_{0}\mid w).
\]

When $y>0$, the conditional distribution is given as 
\[
\begin{aligned}\P(Y_{it}\leq y\mid w) & =\P(Y_{it}^{*}\leq0,\ Y_{it}\leq y\mid w)+\P(Y_{it}^{*}>0,\ Y_{it}\leq y\mid w)\\
 & =\P(Y_{it}^{*}\leq0\mid w)+\P(0<Y_{it}^{*}\leq y\mid w)\\
 & =\P(Y_{it}^{*}\leq y\mid w)\\
 & =\P(v_{it}\geq w_{t}'\t_{0}-y\mid w).
\end{aligned}
\]

Combining the two scenarios, the conditional distributional of $Y_{it}\mid w$
is characterized as follows: 
\[
\P(Y_{it}\leq y\mid w)=\left\{ \begin{aligned}\P(v_{it}\geq\  & w_{t}'\t_{0}-y\mid w)\quad & \text{ if }y\geq0,\\
 & 0\quad & \text{ if }y<0.
\end{aligned}
\right.
\]

Given observed distribution of $Y_{it}\mid w$, we can bound the distribution
$\P(v_{it}\geq c\mid w)$ above as 
\begin{align*}
\P(v_{it}\geq c\mid w) & =\P(v_{it}\geq c,\ w_{t}'\t_{0}-c\geq0\mid w)+\P(v_{it}\geq c,\ w_{t}'\t_{0}-c<0\mid w)\\
 & \leq\P(Y_{it}\leq w_{t}'\t_{0}-c,\ w_{t}'\t_{0}-c\geq0\mid w)+\P(v_{it}\geq w_{t}'\t_{0},\ \ w_{t}'\t_{0}-c<0\mid w)\\
 & =\P(Y_{it}\leq w_{t}'\t_{0}-c,\ w_{t}'\t_{0}-c\geq0\mid w)+\P(Y_{it}=0,\ w_{t}'\t_{0}-c<0\mid w)
\end{align*}
where the above condition holds since $v_{it}\geq c,\ w_{t}'\t_{0}-c<0$
implies $v_{it}\geq w_{t}'\t_{0}$. 

Taking expectation over the endogenous covariate $X_{i}$ yields the
upper bound for the distribution $v_{it}\mid z$: 
\begin{multline*}
\P(v_{it}\geq c\mid z)\leq\P(Y_{it}\leq z_{t}'\b_{0}+X_{it}'\g_{0}-c,\ z_{t}'\b_{0}+X_{it}'\g_{0}\geq c\mid z)+\\
\P(Y_{it}=0,\ z_{t}'\b_{0}+X_{it}'\g_{0}<c\mid z).
\end{multline*}

Rearranging the formula, the above upper bound can be equivalently
written as 
\[
\begin{aligned} & \P(Y_{it}\leq z_{t}'\b_{0}+X_{it}'\g_{0}-c,\ z_{t}'\b_{0}+X_{it}'\g_{0}\geq c\mid z)+\P(Y_{it}=0,\ z_{t}'\b_{0}+X_{it}'\g_{0}<c\mid z)\\
= & \P(0<Y_{it}\leq z_{t}'\b_0+X_{it}'\g_0-c,\ z_{t}'\b_0+X_{it}'\g_0\geq c\mid z)+\P(Y_{it}=0\mid z)\\
= & \P(0<Y_{it}\leq z_{t}'\b_0+X_{it}'\g_0-c\mid z)+\P(Y_{it}=0\mid z).
\end{aligned}
\]

Similarly, the conditional distribution $v_{it}\mid w$ can be bounded
below 
\[
\P(v_{it}\geq c\mid w)\geq\P(Y_{it}\leq w_{t}'\t_{0}-c\mid w),
\]
where the above condition holds since either $w_{t}'\t_{0}-c\geq0$
so that there exists $y=w_{t}'\t_{0}-c\geq0$ such that $\P(Y_{it}\leq y\mid w)=\P(v_{it}\geq c\mid w)$,
or the lower bound is zero when $w_{t}'\t_{0}<c$.

Taking expectation over $X_{i}$ leads to the following lower bound:
\[
\P(v_{it}\geq c\mid z)\geq\P(Y_{it}\leq z_{t}'\b_{0}+X_{it}'\g_{0}-c\mid z).
\]

Given the bounds on the distribution $\P(v_{it}\geq c\mid z)$, the
partial stationarity assumption implies the following identifying
restriction for $\t_{0}$: 
\[
\max_{t}\P(Y_{it}\leq z_{t}'\b_{0}+X_{it}'\g_{0}-c\mid z)\leq\min_{s}\{\P(0<Y_{is}\leq z_{s}'\b_0+X_{is}'\g_0-c\mid z)+\P(Y_{is}=0\mid z)\},
\]
for any $c\in\mathcal{R}$ and any $z$. 
\end{proof}

\subsubsection*{Analysis of Model \eqref{eq:model_cen_latent}}

We now turn to Model \eqref{eq:model_cen_latent} and explain how
we can adjust the results in Proposition \ref{prop:censor} to this
case. Given that $Y_{i,t-1}^{*}=Y_{i,t-1}$ when $Y_{i,t-1}>0$, we
can further relax the lower and upper bounds in \eqref{eq:censor}
to identify $\t_{0}$.

The lower bound in condition \eqref{eq:censor} can be bounded below
as follows: 
\begin{align*}
\P\left(\rest{v_{it}\geq c}z\right)\geq\  & \P(Y_{it}\leq z_{t}'\b+Y_{i,t-1}^{*}\g-c\mid z)\\
\geq\  & \P(Y_{it}\leq z_{t}'\b+Y_{i,t-1}\g-c,Y_{i,t-1}>0\mid z)
\end{align*}
which no longer involves the unobserved $Y_{i,t-1}^{*}$. Similarly,
the upper bound in condition \eqref{eq:censor} can be further bounded
above 
\begin{align*}
\P\left(\rest{v_{is}\geq c}z\right)\leq\  & \P(0<Y_{is}\leq z_{s}'\b+Y_{i,s-1}^{*}\g-c\mid z)+\P(Y_{is}=0\mid z),\\
\leq\  & \P(0<Y_{is}\leq z_{s}'\b+Y_{i,s-1}\g-c,\ Y_{i,s-1}>0\mid z)\\
 & +\P(Y_{is}>0,\ Y_{i,s-1}=0\mid z)+\P(Y_{is}=0\mid z)
\end{align*}
which is again free of $Y_{i,s-1}^{*}$.

We can thus aggregate the lower and upper bounds intertemporally to
produce the identified set as before.

% \subsection{Proof of Proposition \ref{prop:censor}}

\subsection{Visualization of Identification Set in Binary Choice Model} \label{subsec:Q_gamma}

In this section, we numerically compute and visualize the identified
set we derived, using the dynamic (endogenous) binary choice models
as an illustration.

Specifically, we focus on the following model
\[
Y_{it}=\ind\left\{ Z_{it}+\g_{0}X_{it}+\a_{i}+\e_{it}\geq0\right\} ,\quad t=1,...,T=3
\]
where $Z_{it}$ and $X_{it}$ are both taken to be scalar valued.
We normalize the coefficient on $Z_{it}$ to $1$ and focus on the
identification of the coefficient $\g_{0}$.\emph{ }

\emph{Conditional on a given value of $Z_{i}=z\in{\cal Z}:=\left[-10,10\right]^{T}$},
we set the error term $\e_{it}\sim_{i.i.d.}Logistic\left(0,2\right)$,
and the fixed effect 
\[
\a_{i}=\rho_{\a}\cd\frac{1}{T}\sum_{t=1}^{T}z_{t,1}+\sqrt{1-\rho_{\a}^{2}}\xi_{i}
\]
with $\rho_{\a}=0.1$ and $\xi_{i}\sim_{i.i.d.}\cN\left(0,1\right)$. 

Based the above, we then consider the following two versions of the
true DGPs: 
\begin{enumerate}
\item Discrete Specification of $X_{it}$: We set
\[
X_{i,t}=Y_{i,t-1}
\]
which corresponds to the benchmark AR(1) dynamic model, and $\g_{0}=10,$
which is set to bring $\g_{0}X_{it}$ roughly to the same order of
magnitude as $z$.
\item Continuous Specification of $X_{it}$: We set
\[
X_{it}:=5\left(2\left(Y_{i,t-1}-0.5\right)+\eta_{it}\right)
\]
where $\eta_{it}\sim_{i.i.d.}U\left(-1,1\right)$, and $\g_{0}=1$,
so that $X_{it}^{'}\g_{0}$ is again of the similar order of magnitude
as $z$. 
\end{enumerate}
In either case, we set the initial condition $Y_{i,0}\sim_{i.i.d}Bernoulli\left(0,5\right)$
and the true parameter 

Write
\[
Q^{*}\left(\g\right):=\max_{c\in{\cal C},z\in{\cal Z}}Q\left(\g,c,z\right),\quad Q\left(\g,c,z\right):=\max_t p_{L,t}\left(\rest cz,\g\right)+p_{U,t}\left(\rest cz,\g\right)-1,
\]
where
\begin{align}
p_{L,t}(\rest{c}z,\g) & := L_t(\rest{c}z,\g) = \P(\rest{Y_{it} = 1 \text{ and } z_t + \g X_{it}\leq c}z)\\ 
p_{U,t}(\rest{c}z,\g) & := 1 - U_t(\rest{c}z,\g) = \P(\rest{Y_{it} = 0 \text{ and } z_t + \g X_{it}> c}z)
\end{align}
The identified set can thus be equivalently be characterized as 
\[
\G_{I}:=\left\{ \g:\ Q\left(\g\right)\leq0\right\} .
\]

We then implement the following exercise:

(i) Compute $\hat{Q}\left(\g,c,z\right)$ as numerical approximation
of $Q\left(\g,c,z\right)$ via simulations. Specifically, we compute
\[
\hat{Q}\left(\g,c,z\right):=\max_{t=1,...,T}\hat{p}_{L,t}\left(\rest cz,\g\right)+\max_{t=1,...,T}\hat{p}_{U,t}\left(\rest cz,\g\right)-1
\]
where 
\begin{align*}
\hat{p}_{L,t}\left(\rest cz,\g\right) & :=\frac{1}{B}\sum_{b=1}^{B}Y_{bt}\ind\left\{ z_{t}+\g X_{bt}\leq c\right\} \\
\hat{p}_{U,t}\left(\rest cz,\g\right) & :=\frac{1}{B}\sum_{b=1}^{B}\left(1-Y_{bt}\right)\ind\left\{ z_{t}+\g X_{bt}\geq c\right\} 
\end{align*}
using $B=2000$ simulations of $\left(Y_{bt},X_{bt}\right)$ based
on the DGP described above, conditional on each value of $z$. We
emphasize that $B$ simulations $\left(Y_{bt},X_{bt}\right)$ are
generated for each $z$, and thus $B$ should not be exactly interpreted as the usual ``sample size $N$'' as in Section \ref{sec:simu}.

(ii) Optimize $\hat{Q}\left(\g,c,z\right)$ over $\left(c,z\right)\in{\cal C}\times{\cal Z}$,
with ${\cal C}:=\left[-30,30\right]$, to obtain $\hat{Q}^{*}\left(\g\right)$
as a numerical approximation of $Q^{*}\left(\g\right)$, using the
R package GenSA, which implements the global optimization algorithm
called Generalized Simulated Annealing \citep{xiang2013generalized}.

(iii) Plot $\hat{Q}^{*}\left(\g\right)$ as a function of $\g$.

~

Figures \ref{fig:Qstar_disc} and \ref{fig:Qstar_cts} contain plots
of $Q^{*}\left(\g\right)$ (more precisely, its numerical approximation
$\hat{Q}^{*}$), and visualize the informativeness of our identified
set. We note that the ``spikes'' and ``wiggliness'' are likely
to be driven by the randomness in the global optimization algorithm,
which is not always guaranteed to find the true global maximum. Consequently,
we should interpret the blue line as 
a lower bound of the numerical approximation of  $Q^{*}\left(\g\right)$.

First, we confirm that $Q^{*}\left(\g_{0}\right)\leq0$, i.e., the
true parameters $\g_{0}=10$ (Design 1) and $1$ (Design 2) indeed lie within the identified set $\G_{I}$.
Second, we observe that the identified set $\G_{I}$ is nontrivial:
(i) $Q\left(\g\right)\leq0$ only in a neighborhood around the true
value $\g_{0}$, (ii) the sign of of $\g_{0}$ is correctly identified
in both figures, and (iii) the identified set under the binary specification
$X_{it}=Y_{i,t-1}$ is much wider than the one under the continuou
specification, which is as expected.

\begin{figure}
\caption{$Q^{*}\left(\protect\g\right)$ under Discrete Specification of $X_{it}$}
\label{fig:Qstar_disc}
\centering{}\includegraphics[scale=0.65]{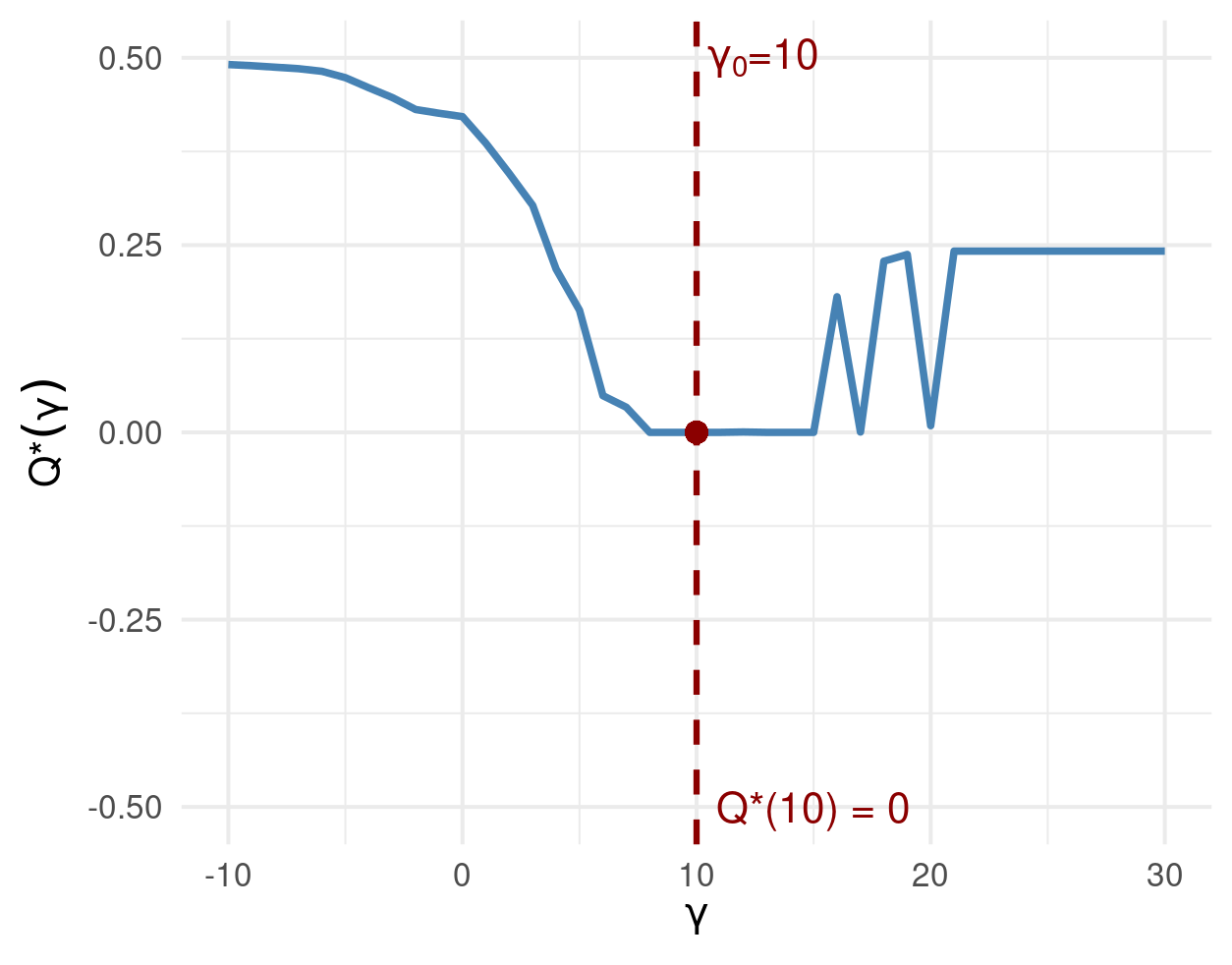}
\end{figure}

\begin{figure}
\caption{$Q^{*}\left(\protect\g\right)$ under Continuous Specification of
$X_{it}$}
\label{fig:Qstar_cts}
\centering{}\includegraphics[scale=0.65]{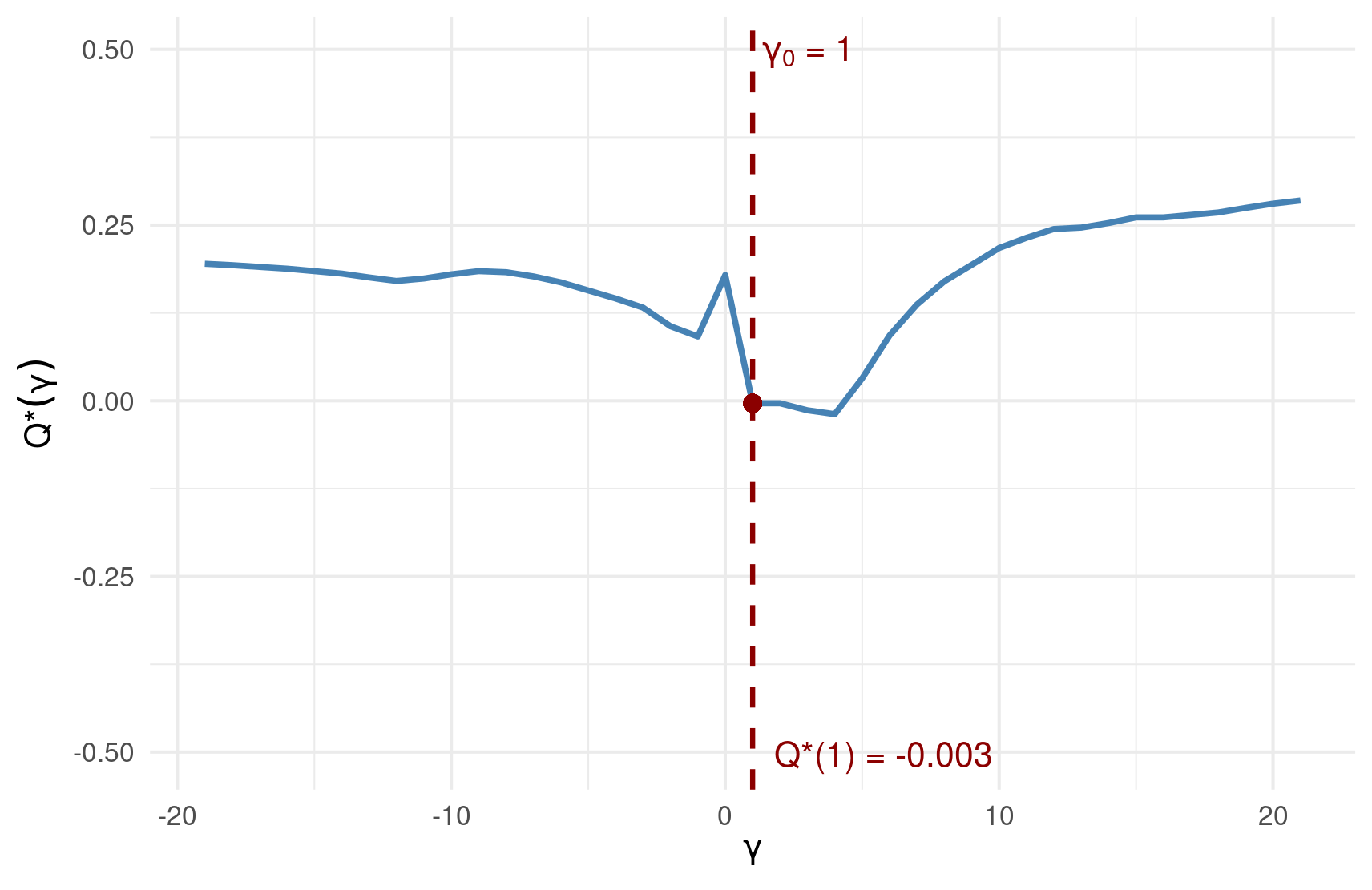}
\end{figure}

We emphasize our visualization of $\G_{I}$ via $Q^{*}\left(\g\right)$
should be interpreted as \emph{conservative, }since we only set ${\cal C}=\left[-20,20\right]$
and the global optimization algorithm (GenSA) may not get the absolute
maximum on ${\cal C}\times{\cal Z}$ (so the maximum value returned
by the algorithm may be strictly smaller than the true maximum). Consequently,
the visualized function $\hat{Q}^{*}\left(\g\right)$ should be interpreted
as a lower bound of an approximation of the true $Q^{*}\left(\g\right)$.
\end{document}